\pdfoutput=1
\RequirePackage{rotating}
\documentclass[acmsmall,screen,nonacm,appendix]{acmart}
\usepackage{hyperref}
\usepackage[capitalize,nameinlink,noabbrev]{cleveref}

\authorsaddresses{}
\renewcommand\footnotetextcopyrightpermission[1]{} \setcopyright{none}

\usepackage{microtype}
\usepackage{wrapfig}
\usepackage{natbib}
\usepackage[inline]{enumitem}
\usepackage{mathpartir}
\let\ifdraft\iffalse
\usepackage{style/notation}
\usepackage{style/pftools}
\usepackage{style/tikzit}
\usepackage{style/appendix}

\begin{document}

\title[A Bunch of Sessions]{A Bunch of Sessions: \texorpdfstring{\\ }{} A Propositions-as-Sessions Interpretation of Bunched Implications in Channel-Based Concurrency}
\ifappendix
  \subtitle{(Extended Version)}
\fi

\author{Dan Frumin}
\email{d.frumin@rug.nl}
\affiliation{\institution{University of Groningen}
\country{The Netherlands}
}

\author{Emanuele D'Osualdo}
\email{dosualdo@mpi-sws.org}
\affiliation{\institution{MPI-SWS}
  \city{Saarbr\"ucken}
  \country{Germany}
}

\author{Bas van den Heuvel}
\email{b.van.den.heuvel@rug.nl}

\author{Jorge A.~Pérez}
\email{j.a.perez@rug.nl}

\affiliation{\institution{University of Groningen}
\country{The Netherlands}
}

\begin{abstract}
  The emergence
  of \emph{propositions-as-sessions},
  a Curry-Howard correspondence between propositions of Linear Logic
  and session types for concurrent processes,
  has settled the logical foundations of message-passing concurrency.
  Central to this approach is the
  \emph{resource consumption} paradigm heralded by Linear Logic.

  In this paper,
  we investigate a new point in the design space
  of session type systems for message-passing concurrent programs.
  We identify O'Hearn and Pym's \emph{Logic of Bunched Implications}~(BI)
  as a fruitful basis for an interpretation of the logic
  as a concurrent programming language.
  This leads to a treatment of \emph{non-linear resources} that is
  radically different from existing approaches based on Linear Logic.
  We introduce a new \picalc\ with sessions,  called \piBI;
  its most salient feature is a construct called \emph{spawn},
  which  expresses new forms of sharing that are induced
  by structural principles in BI.
  We illustrate the expressiveness of \piBI and lay out its fundamental theory:
  type preservation, deadlock-freedom, and weak normalization results for well-typed processes;
  an operationally sound and complete typed encoding of an affine \pre\lambda-calculus;
  and a non-interference result for access of resources.
\end{abstract}

\maketitle

\section{Introduction}
\label{sec:intro}
In this paper, we investigate a new point in the design space of session type systems for message-passing concurrent programs.
We identify the \emph{Logic of Bunched Implications}~(BI) of \citet{ohearn.pym:1999} as a fruitful basis for an interpretation of the logic as a concurrent programming language, in the style of \emph{propositions-as-sessions}~\cite{caires.pfenning:2010,wadler:2012}.
This leads to a treatment of \emph{non-linear resources} that is radically different from existing approaches based on Girard's Linear Logic (LL).
We propose \piBI, the first concurrent interpretation of BI,
and we study the behavioral properties enforced by typing,
laying the meta-theoretical foundations needed,
and clarifying its relation to the other type-theoretic interpretations of BI.

\paragraph{Session types for message-passing concurrency.}
Writing concurrent programs is notoriously hard, as bugs might be caused by subtle undesired interactions between processes.
Statically enforcing the absence of bugs while allowing expressive concurrency patterns is important but difficult.
In the context of message-passing concurrency, type systems based on \emph{session types} provide an effective approach.
Session type systems enforce a communication structure between processes and channels, with the intent of (statically) ruling out races (as in, e.g., two threads sending messages over the same channel at the same time) and other undesirable behaviors, like deadlocks.
This communication structure is formulated at the type level.
For example, the session type $T={!\mathsf{int}.?\mathsf{string}.!\mathsf{bool}.\mathsf{end}}$ (written in the syntax of \cite{vasconcelos:2012}) describes a protocol that first outputs an integer ($!\mathsf{int}$), then inputs a string ($?\mathsf{string}$), and finally outputs a boolean ($!\mathsf{bool}$).
In session-based concurrency, types are assigned to channel names; this way, e.g.,
the assignment $x : T$ dictates that the communications on channel $x$ must adhere to the protocol described by $T$.

The fundamental idea behind session type systems is that an assumption such as~$x:T$ is like a \emph{resource} that can be consumed and produced.
For example, the act of sending an integer on the channel~$x$
consumes~$x:{!\mathsf{int}.?\mathsf{string}.!\mathsf{bool}.\mathsf{end}}$
and produces a new resource $x:{?\mathsf{string}.!\mathsf{bool}.\mathsf{end}}$,
representing the expected continuation of the protocol.
Then, the coordinated use of channels requires a strict discipline on how resources can be consumed and produced:
it is unwise to allow multiple processes to access the same resource~$x:T$,
otherwise simultaneous concurrent outputs by different processes on the same channel will render the protocol invalid.
The type system is thus designed to enforce
that some resources, like those associated with channels, are \emph{linear}: they are consumed \emph{exactly once}.
By enforcing linearity of these resources, session type systems ensure that well-typed programs conform to the protocols encoded as types, and satisfy important correctness properties, such as deadlock-freedom.

\paragraph{Propositions-as-sessions.}
A central theme in this paper is how logical foundations can effectively inform the design of expressive type disciplines for programs.
In the realm of functional programming languages,
such logical foundations have long been understood via type systems obtained
through strong Curry-Howard correspondences with known logical proof systems
(\eg the correspondence between the simply-typed $\lambda$-calculus and intuitionistic propositional logic).
For concurrent languages, on the other hand, such correspondences have been more elusive.
Indeed, although the original works on session types by \citet{honda:1993,honda.vasconcelos.kubo:1998}  feature an unmistakable influence of LL in their formulation, the central question of establishing firm logical foundations for session types remained open until relatively recently.
The first breakthroughs were the logical correspondences based on the concurrent languages \piDILL \cite{caires.pfenning:2010} and CP \cite{wadler:2012} (based on Intuitionistic LL and Classical LL, respectively).
These works define a bidirectional correspondence, in the style of Curry-Howard, which allows us to interpret propositions as session types (protocols),  proofs as $\pi$-calculus processes, and cut elimination as process communication.
These correspondences are often collectively referred to as \emph{propositions-as-sessions}.

Intensely studied in the last decade, the line of work on propositions-as-sessions provides a principled justification to a linear typing discipline.
These correspondences also clarify our understanding of the status of \emph{non-linear} resources, which do not obey resource consumption considerations.
Non-linear resources, such as mutable references, client/server channels, and shared databases, are commonplace in practical programs and systems.
Disciplining non-linear resources is challenging, because there is a tension between flexibility and correctness: ideally, one would like to increase the range of (typable) programs that can be written, while ensuring that such programs treat non-linear resources consistently.

LL allows for a controlled treatment of non-linear resources through the
modality $\bang{A}$.
Within propositions-as-sessions, the idea is that a session of type $\bang{A}$ represents a server providing a session of type $A$ to its clients, and the server itself can be duplicated or dropped.
Those particular features---being able to replicate or drop a session---are achieved through the usage of \emph{structural rules} in the sequent calculus, specifically the rules of contraction and weakening, which are restricted to propositions of the form $\bang{A}$.
A series of recent works have explored quite varied ways of going beyond this treatment of non-linear resources: they have put forward concepts such as \emph{manifest sharing} \cite{balzer:pfenning:2017}, dedicated frameworks such as \emph{client-server logic} \cite{qian.kavvos.birkedal:2021}, and specific constructs for non-deterministic, fail-prone channels \cite{caires.perez:2017}.

\paragraph{The Logic of Bunched Implications}
At their heart, the aforementioned works propose different ways of treating non-linear resources through {modalities}.
Relaxing linearity through a modality allows a clean separation between the worlds of linear and non-linear resources.
This approach relies on rules that act as ``interfaces'' between the two worlds, allowing conversions  between linear and non-linear types only under controlled circumstances.

However, modalities are not the only way in which substructural logics can integrate non-linear resources.
A very prominent alternative is provided by the Logic of Bunched Implications~(BI) of \citet{ohearn.pym:1999}.
BI embeds the pure linear core of LL as multiplicative conjunction~$*$ and implication~$\wand$, but extends it by introducing additive conjunction~$\land$ and implication~$\to$, which are treated non-linearly.
BI can thus be thought of as enabling the free combination of linear and non-linear resources in a single coherent logic.

The result is a logic which admits an interpretation of linearity that is enticingly different from LL.
Conceptually, LL admits a ``number of uses'' interpretation, where types can specify \emph{how many times a resource should be used}: exactly once for linear resources, any number of times for~$\bang{A}$ resources.
On the contrary, BI admits an ``ownership'' interpretation \cite{pym.ohearn.yang:2004}, which focuses on \emph{who has access to which resources}.

The ownership interpretation has positioned BI as the logic of choice for program logics for reasoning about stateful and concurrent programs, under the umbrella of (Concurrent) Separation Logic (see, e.g., the surveys by \citet{ohearn:2019} and \citet{brookes:ohearn:2016}).
While separation logic has received significant attention, the same cannot be said about type-theoretic interpretations of BI as a type system for concurrency.
To our knowledge, the only type-theoretic investigation into the (proof theory of) BI has been the \alcalc \cite{ohearn:2003}---a \pre\lambda-calculus arising from the natural deduction presentation of BI---and its variations \cite{atkey:2004,collinson.pym.robinson:2008}.

\paragraph{Our key idea}
Here we propose \piBI: the first process calculus for the propositions-as-sessions and processes-as-proofs interpretation of BI, based on its sequent calculus formulation.
The result is an expressive concurrent calculus with a new mechanism to handle non-linear resources, which satisfies important behavioral properties, derived from a tight correspondence with BI's proof theory.
The central novelty of \piBI is a process interpretation of the structural rules, which closely follows the proof theory of BI.

Consider the case of contraction/duplication.
Given a session $x:A$, how can we duplicate it into sessions $x_1 : A$ and $x_2 : A$?
The difficulty here is that after duplication,
the two assumptions might be used differently and asynchronously.
We conclude that the actual process implementing those sessions
in the current evaluation context needs to be duplicated,
such that two independent processes can provide the duplicated sessions.
This ``on demand non-local replication'' of a process in the evaluation context
is not something supported natively by the \picalc.
We propose a new process construct, a prefix dubbed \emph{spawn},
which achieves this.

We illustrate the spawn prefix with a simple example.
Let $P$ and $Q$ be two processes, with $P$ providing a service on the channel $x$, and $Q$ requiring two copies of the service.
The spawn prefix $\spawn{x->x_1,x_2}$ denotes a request to the environment to duplicate the service on $x$ into copies on the new channels $x_1$ and $x_2$.
Then, $\spawn{x->x_1,x_2}.Q$ is a process that first performs the request and then behaves as $Q$.
The composition of these processes is denoted $\new x.(P \| \spawn{x->x_1,x_2}.Q)$, where `${}\| {}$' and `$(\restrsym x)$' stand for parallel composition and restriction on $x$, respectively.

In the reduction semantics of  \piBI, obtained from the proof theory of BI, the composed process reduces as follows:
\[
\new x.(P \| \spawn{x->x_1,x_2}.Q) \redd \new x_1. \big(P\subst{x->x_1} \| \new x_2.(P\subst{x->x_2} \| Q)\big).
\]
This way, the duplication request leads to the composition of two copies of $P$ (each with an appropriate substitution  $\subst{x->x_i}$) with the process $Q$ on channels $x_1$ and $x_2$, as desired.

The behavior of the spawn prefix is determined by the context in which it is executed and it communicates with the run-time system to achieve contraction or weakening.
This mechanism reminds us of horizontal scaling in cloud computing, with the spawn prefix playing the role of middleware: it requests the runtime environment to scale up/down a particular resource.
For example, a load balancer might determine that in a certain situation the execution environment has to provide an additional snapshot of a Docker container, and route part of the environment's requests to it.

As we will see, spawn reductions involve the propagation of the effects of duplicating processes (such as $P$ above); 
we give the full definition and illustrate it further in \Cref{{sec:pibi}}.

\paragraph{Contributions}
As mentioned, the spawn prefix provides a direct interpretation of the structural rules in the design of the type system, adopting BI as the underlying logic.
The resulting system is significantly expressive and yet different from  systems derived from propositions-as-sessions, which is not so surprising: as logics, BI and LL are \emph{incomparable}: there are provable formulas of LL that are not provable in BI, and vice versa.
As such, an immediate question is whether \piBI satisfies the expected meta-theoretical properties for session-typed processes: \emph{type preservation} and \emph{deadlock-freedom}.
The key difficulty is that the semantics of the spawn prefix is fundamentally \emph{non-local}---it depends on its execution context.
As a \textbf{first contribution}, we show that type preservation and deadlock-freedom hold for \piBI; moreover, we prove \emph{weak normalization}, which further justifies the semantics of spawn prefixes.

In addition to these meta-theoretical properties, an essential ingredient in the propositions-as-sessions research program is defined by concurrent interpretations  of (typed) functional calculi, in the spirit of Milner's seminal work on \emph{functions-as-processes}~\cite{milner:1992}.
As already mentioned, the only prior type-theoretic interpretation of BI is the (sequential) calculus \alcalc~\cite{ohearn:2003}.
As a \textbf{second contribution}, we define a translation from  \alcalc into \piBI, and prove that it correctly preserves and reflects the operational semantics of terms and processes, respectively.

While insightful and novel, the operational semantics of  \piBI and the translation of the \alcalc do not offer us a direct insight in the meaning of and difference between the types in our system (as is the case in the \alcalc).
A natural question is: what is the difference between multiplicative conjunction $\ast$ and additive conjunction $\wedge$ in \piBI?
As an answer to this question, our \textbf{third contribution} is a \emph{denotational semantics} for \piBI, which interprets processes as functions and describes types in terms of ``provenance tracking''.

Intuitively, our denotational semantics considers that duplication through a spawn prefix generates typed processes with the same provenance.
This notion of provenance then allows us to precisely distinguish between $\ast$ and $\wedge$: in a process with a session of type $A \ast B$ the sub-processes providing sessions $A$ and $B$ have a \emph{different origin}, a property that may not necessarily hold for processes with sessions of type $A \wedge B$.
This is possible because the provenance information can be reconstructed from a typing derivation, and it is made evident through the denotational semantics.

In addition to providing a semantic meaning to types, the denotational semantics is sound with regard to \emph{observational equivalence}.
Two processes are observationally equivalent if no other process can (operationally) distinguish between them.
Establishing observational equivalence of programs directly is hard, because it involves reasoning about process behavior under arbitrary contexts.
On the other hand, a denotational semantics provides a direct way of establishing equivalence: if two processes have the same denotation, then they are observationally equivalent.
As an application of the denotational semantics, we frame the operational correspondence for the \alcalc mentioned above in terms of observational equivalence.

\paragraph{Outline}
The rest of the paper is organized as follows.
\Cref{sec:pibi} presents the syntax, semantics, and type system of  \piBI, and illustrates its expressivity.
In \Cref{sec:meta} we establish key meta-theoretical properties of typable processes: type preservation, deadlock freedom,  and weak normalization.
We formally connect the  \alcalc to \piBI by  defining a translation and proving  operational correspondence for it in \Cref{sec:translation}.
We define the denotational semantics for \piBI processes, define observational equivalence, and formally relate the two in \Cref{sec:denot}.
We discuss further related work in \Cref{sec:rel_work} and conclude in \Cref{sec:conclusion}.
The omitted technical details can be found in the appendix.

 \section{The \piBI Calculus}
\label{sec:pibi}
In this section we formally introduce \piBI, a \picalc with constructs for session-based concurrency~\cite{honda:1993,honda.vasconcelos.kubo:1998} and our new spawn prefix.
We first describe syntax and dynamics (reduction semantics), and then present its associated type system, based on the sequent calculus for BI.
Following \piDILL~\cite{caires.pfenning:2010,caires.etal:2016}, our type system for \piBI admits a ``provide/use'' reading for typable processes, whereby a specific channel \emph{provides} a session by \emph{using} zero or more other sessions.

\paragraph{Notation}
We assume an enumerable set of \emph{names} (or \emph{channels}), $ a,b,c,\ldots,x,y,z\in \Name $ to denote channels.
We make use of finite partial functions~$ f \from A \fpfn B $.
We write~$ f(x) = \bot$ if~$f$ is not defined on~$x$.
We define~$\dom(f) = \set{ x \in A | f(x) \ne \bot } $.
We write $ \map{a_1 -> b_1; \dots ;a_n -> b_n} $ to denote a map,
and~$\emptymap$ for the empty map.
We will also use set comprehensions for finite functions,
\eg $
  \map{a -> b | a \in \set{1,2}, b = a^2 }
$.
For a finite partial function~$f$ and a set~$X$, we write~$f \setminus X$
for the function that coincides with~$f$ except for being undefined on~$X$.

\subsection{Process Syntax}
\label{sec:procs}
The syntax of \piBI processes is given in \Cref{fig:processes}.
\begin{mathfig}[\small]
    \begin{align*}
        P,Q,R &::=
        \out x[y].(P \| Q)
        && \text{output}
        &&\sepr\;\, \inp x(y).P
        && \text{input}
        \\
        &\;\,\sepr\;\, \out x<>
        && \text{close}
        &&\sepr\;\, \inp x().P
        && \text{wait}
        \\
        &\;\,\sepr\;\, \selL{x}.P
        && \text{left selection}
        &&\sepr\;\,  \caseLR{x}{P}{Q}
        && \text{branch}
        \\
        &\;\,\sepr\;\, \selR{x}.P
        && \text{right selection}
        &&\sepr\;\, \fwd[x<-y]
        && \text{forwarder}
        \\
        &\;\,\sepr\;\, \new x.(P \| Q)
        && \text{restriction + parallel}
        &&\sepr\;\, \spw{\spvar}.P
        && \text{spawn}
    \end{align*}
\caption{Syntax of \piBI processes.}
\label{fig:processes}
\end{mathfig}
The structure and conventions of process calculi
based on Curry-Howard correspondences are typically based on
an implicit expectation for how the components of a system are organized
--- an expectation that is ultimately verified by typing.
The idea is that interaction is grouped into a \emph{session},
the sequence of interactions along a single channel.
As hinted at above, a process~$P$ should \emph{provide} a session
at some specific channel~$x \in \fn(P)$,
and there is always a single \emph{user} of the session
exchanging messages with~$P$ along~$x$.
To provide a session, a process can make use of sessions on other channels.

Most constructs are standard and reflect these expectations
of sessions with provide/use roles:
\begin{itemize}
  \item \emph{Input/Output:}
    A process $\inp x(y).P$ receives
    a channel~$y$ from the session at~$x$ and proceeds as~$P$,
    continuing the session at~$x$.

    A process $\out x[y].\bigl(P \| Q\bigr)$
    sends a \emph{fresh} channel~$y$ over the session at~$x$;
    the process~$P$ provides the new session at~$y$,
    while~$Q$ continues the session at~$x$.

  \item \emph{Labelled choice (selection and branching):}
    The processes $\selL{x}.P$ and $\selR{x}.P$ select left/right labels over the session at~$x$, respectively.
    The dual process $\caseLR{x}{P}{Q}$ offers these left/right options,
    which trigger continuation~$P$ or~$Q$, respectively.

  \item \emph{Explicit session closing:}
    The end of a session is expected to be explicitly closed by a final handshake between
    the dual prefixes $\out x<>$ and $\inp x().P$ (empty output/input, respectively).

  \item \emph{Structured parallel composition:}
    Parallel composition, in keeping with \piDILL~\cite{caires.etal:2016},
    is used jointly with restriction.
    In a process~$\new x.( P \| Q )$
    a new session is created at~$x$,
    provided by~$P$ with~$Q$ as its only user.
    To improve readability, we sometimes annotate the parallel operator with the name of the associated restriction, and write $\new x.( P \|_x Q )$.

  \item \emph{Forwarders:}
    A process $\fwd[x<-y]$
    provides a session at~$x$
    as a copycat of the session at~$y$.
\end{itemize}

The key novel construct of \piBI
is the \emph{spawn prefix} $\spw\spb.P$.
It is parametrized by what we call a \emph{spawn binding}~$\spb$.
Spawn bindings, formally defined below, are a unification and generalisation
of prefixes like
$ \spawn{x->x_1,x_2}$ (copy the session at $x$ to $x_1$ and $x_2$)
but also
$ \spawn{x->\emptyset}$ (drop the session at $x$).
Indeed, in addition to allowing the simultaneous mapping of more than one
name~$x$, we allow names to be mapped to sets of names, encompassing the nullary and binary cases above.

\begin{definition}[Spawn binding]
\label{def:spawn-env}
A finite partial function
  $ {\spb \from \Name \fpfn \pset{\Name}} $
  is a \emph{spawn binding}~if:
  \begin{itemize}
    \item
  $
        \A x,y\in\dom(\spb). x \neq y \implies \spb(x) \inters \spb(y) = \emptyset
      $, and
  \item
   $
      \A x\in\dom(\spb). \dom(\spb) \inters \spb(x) = \emptyset
    $.
  \end{itemize}

  We define the \emph{restrictions} of~$\spb$ to be the set
  $ \restrOf(\spb) = \Union_{x \in \dom(\spb)} \spb(x) $.
  We omit redundant delimiters in spawn prefixes,
  \eg we write $ \spawn{x->x_1,x_2; y->y_1} $
      for $ \spawn{\map{x->\set{x_1,x_2}; y->\set{y_1}}} $.

  Given two spawn bindings~$\spb_1$ and~$\spb_2$ we say
  they are \emph{independent}, written~$ \spb_1 \indep \spb_2 $,
  if
  $\dom(\spb_1) \inters \dom(\spb_2) = \emptyset$,
  $\dom(\spb_1) \inters \restrOf(\spb_2) = \emptyset$,
  $\restrOf(\spb_1) \inters \restrOf(\spb_2) = \emptyset$, and
  $\dom(\spb_2) \inters \restrOf(\spb_1) = \emptyset$.
\end{definition}

\paragraph{Free and bound names}
Except for the new spawn construct, the notion of free and bound names is standard: the processes $\out x[y].(P \| Q)$, $\inp x(y).P$, and $\new y.(P \| Q)$ all bind~$y$.
For the spawn prefix, the situation is a bit different.
Given a set of names $X$, a spawn~$\spawn{x->X}.P$ signals to the context that $P$ will use $n=\card{X}$ times the session at~$x$.
The names in~$X$ indicate the new names that $P$ will use instead of~$x$.
As such, these new names are bound in~$P$ by the spawn prefix, whereas the original name $x$ is free in~$P$.
Formally, $
  \fn(\spw\spvar.P) =
    \bigl(
      \fn(P) \setminus \restrOf(\spvar)
    \bigr) \cup \dom(\spvar)
$.

We implicitly identify processes up to \pre\alpha-conversion and we adopt Barendregt's variable convention: all bound names are different, and bound names are different from free names.

\paragraph{Structural congruence}
As usual,
we define a congruence that identifies
processes up to inconsequential syntactical differences.
\emph{Structural congruence}, denoted $\congr$,
is the smallest congruence satisfying the rules in~\Cref{fig:congr}:
\begin{mathfig}[\small]
  \adjustfigure

\begin{proofrules}
  \infer*[lab={cong-assoc-l},rightstyle=\em,right={(when $x \notin \fn(Q) \land y \notin \fn(P)$)}]
  {}{
    \new x.\bigl(P \|_x \new y. (Q  \|_y R)\bigr)
    \congr
    \new y.\bigl(Q \|_y \new x. (P  \|_x R)\bigr)
  }
  \label{cong-assoc-l}

  \infer*[lab=congr-assoc-r,rightstyle=\em,right={(when $x \notin \fn(R) \land y \notin \fn(P)$)}]{}{
    \new x.\bigl(P \|_x \new y. (Q  \|_y R)\bigr)
    \congr
    \new y.\bigl(\new x. (P  \|_x Q) \|_y R\bigr)
  }
  \label{cong-assoc-r}

  \infer*[lab=congr-spawn-swap,rightstyle=\em,right={(when $\spb_1 \indep \spb_2$)}]{}{
    \spawn{\spb_1}. \spawn{\spb_2}. Q
    \equiv
    \spawn{\spb_2}. \spawn{\spb_1}. Q
  }
  \label{cong-spawn-swap}
\end{proofrules}

   \caption{Structural congruence.}
  \label{fig:congr}
\end{mathfig}
the orders of parallel compositions and
independent spawn prefixes do not matter (\cref{cong-assoc-r,cong-assoc-l} and \cref{cong-spawn-swap}, resp.).

Our structural congruence is a bit more fine-grained than is usual for the \picalc.
This is guided by the desire to make typing consistent under structural congruence.
Typing will enforce the expectations of process structure alluded to before, so our  congruence needs to preserve them.
For example, in a process $ \out x[y].(P\|Q) $ we expect $P$ to provide the new session at~$y$ and~$Q$ to continue the session at~$x$.
Admitting commutativity of parallel would break this expectation.
Similarly, in the composition of processes $\new x.(P\| Q)$ it is important that $P$ provides the session that governs $x$, and that $Q$ dually uses the session at $x$.
This choice of structural congruence simplifies the technical development and makes the correspondence between logic and type theory sharper.

\subsection{Reduction Semantics}
\label{sec:red-sem}
The operational semantics of \piBI
is defined in terms of a reduction relation, denoted $\redd$,
which combines the usual reductions of the \picalc\ with reductions
for spawn prefixes.
As usual, we shall write $\reddStar$ to denote the reflexive, transitive closure of $\redd$, and $P \nredd$ when $P$ cannot reduce.

\begin{figure}[t]
  \adjustfigure[\footnotesize]
\begin{proofrules}
  \infer*[lab=red-comm-r]{}{
  \new x.\bigl(\inp x(y).Q \|_x \out x[y].(P_1 \| P_2) \bigr)
  \redd
  \new x.\bigl(\new y.(P_1 \|_y Q) \|_x P_2 \bigr)
}
\label{red-comm-r}

\infer*[lab=red-unit-r]{}{
  \new x.\big(\inp x().Q \|_x \out x<> \big)
  \redd
  Q
}
\label{red-unit-r}

\infer*[lab=red-comm-l]{}{
  \new x.\big(\out x[y].(P_1 \| P_2) \|_x \inp x(y).Q \big)
  \redd
  \new x.\big(P_2 \|_x \new y.(P_1 \|_y Q)\big)
}
\label{red-comm-l}

\infer*[lab=red-unit-l]{}{
  \new x.\big(\out x<> \|_x \inp x().Q \big)
  \redd
  Q
}
\label{red-unit-l}

\infer*[lab=red-case]{\ell \in \set{\inl,\inr}
}{
  \new x.\bigl(\sel{x}{\ell}.P \|_x \caseLR{x}{Q_\inl}{Q_\inr}\bigr)
  \redd
  \new x.\bigl(P \|_x Q_\ell\bigr)
}
\label{red-case}
\\

\infer*[lab=red-fwd-r]{
    x \neq y \and y \notin \fn(P)
}{
   \new x.\left(P \|_x \fwd[y<-x]\right) \redd P\subst{x->y}
}
\label{red-fwd-r}

\infer*[lab=red-fwd-l]{
   x \neq y \and y \notin \fn(P)
}{
   \new x.\left(\fwd[x<-y] \|_x P\right) \redd \substS{P}{x}{y}
}
\label{red-fwd-l}

 \infer*[lab=red-spawn]{
   \spb(x) = \set{x_1,\dots,x_n}
   \\
   \spb' = \bigl(
      (\spb \setminus \set{x})
      \union
      \map{z \mapsto \set{z_1,\dots,z_n} | z \in \fn(P) \setminus \set{x}}
   \bigr)
}{
   \new x.\big(P \|_x \spawn{\spb}. Q\big)
   \redd
   \spawn{\spb'}.
      \new x_1.\big(
         P^{(1)} \|_{x_1} \dots \new x_n.(P^{(n)} \|_{x_n} Q)
      \dots\big)
}
\label{red-spawn}

\infer*[lab=red-spawn-r]{
   x \notin \dom(\spb)}{
   \new x.\big(P \|_x \spawn{\spb}. Q\big)
   \redd
   \spawn{\spb}.\new x.(P \|_x Q\big)
}
\label{red-spawn-r}

\infer*[lab=red-spawn-l]{
   x \notin \dom(\spb)}{
   \new x.\big(\spawn{\spb}. P \|_x Q\big)
   \redd
   \spawn{\spb}.\new x.(P \|_x Q\big)
}
\label{red-spawn-l}

\infer*[lab=red-spawn-merge]{}{
   \spawn{\spb_1}.\spawn{\spb_2}.P
   \redd
   \spawn{\spb_1 \merge \spb_2}.P
}
\label{red-spawn-merge}
 \\
  \begin{grammar*}
  \ectxt[\hole] \is
       [\hole]
    |  \spw\spb.\ectxt[\hole]
    |* \new x.(P \| \ectxt[\hole])
    |* \new x.(\ectxt[\hole] \| P)
\end{grammar*}

\infer*[lab=red-eval-ctxt]{
  P \redto Q
}{
  \ectxt[P] \redd \ectxt[Q]
}
\label{red-eval-ctxt}

\infer*[lab=red-congr]{
  P' \congr P
  \\
  P \redto Q
  \\
  Q \congr Q'
}{
  P' \redto Q'
}
\label{red-cong}
 \end{proofrules}
\caption{Reduction rules for \piBI.}
  \label{fig:red-main}
\end{figure}

\Cref{fig:red-main} gives the reduction rules.
The first seven rules describe interactions along a channel.
\Cref{red-comm-r,red-comm-l} describe the exchange of channel~$y$ along~$x$.
The resulting process contains an explicit restriction for~$y$ with $P_2$ out of scope, reflecting the expectation that~$P_1$ is
the provider of the new session at~$y$.
\Cref{red-unit-r,red-unit-l} describe the closing of a session at~$x$.
\Cref{red-case} shows how a branch offered on~$x$
can be selected by sending $\inl$ or $\inr$.
Finally, \cref{red-fwd-r,red-fwd-l} explain the elimination of a forwarder connected by restriction in terms of a substitution.

The next four rules of \cref{fig:red-main}
define the semantics of spawn.
The crucial rule is \cref{red-spawn}, which we explain by example.
\begin{example}\label{ex:contraction}
    Consider a process $P$ that provides a session on channel $x$.
    Another process $Q$ provides a session on $v$ by relying \emph{twice} on the session provided by $P$, on channels $x_1$ and $x_2$.
    Simple concrete examples are $P \is \out x<>$ and $Q \is \inp x_1().\inp x_2().\out v<>$.
    Now consider the following process:
    \[
        R \is \new x.(\inp z().P \|_x \spw{x->x_1,x_2}.Q)
    \]
    In $R$, the process $P$ is blocked waiting for the session on a channel $z$ to close.
By \Cref{red-spawn},
    \[
        R \redd \spw{z->z_1,z_2}.\new x_1.(\inp z_1().P\subst{x->x_1} \|_{x_1} \new x_2.(\inp z_2().P\subst{x->x_2} \|_{x_2} Q)).
    \]
    The result is two copies of $P$, providing their sessions on $x_1$ and $x_2$ instead of on $x$.
    Since we are also copying the closing prefixes on $z$, an additional spawn is generated, but now on $z$: it signals to the environment that two copies of the process providing the session on $z$ should be created and that they should provide its session on $z_1$ and $z_2$.
\end{example}
In the example above, the channel $z$ is a free name of the process that is copied by the spawn reduction.
Generally, a copied process may rely on arbitrarily many sessions on the free names of the process, and all the processes providing these sessions will have to be copied as well.
To handle the general case, \cref{red-spawn} uses the following definition.
\begin{definition}[Indexed renaming]
\label{def:idx-renaming}
  Given a process~$P$ with
  $ \fn(P) = \set{a,b,\dots,z} $,
  we define
  $\idx{P}{i}$
  to be the process~$P$ where every free name is replaced by a fresh copy of the name indexed by~$i$.
  Formally, assuming $a_i,b_i,\dots,z_i \notin \fn(P)$,
  $
    \idx{P}{i} \is
      P\subst{a->a_i,b->b_i,,z->z_i}
  $.
\end{definition}
Note that \cref{red-spawn} uniformly handles the case where
a session is not used at all.
\begin{example}\label{ex:weakening}
    Consider again $P$ that provides a session on $x$.
    This time, the process $Q'$ provides a session on $v$ \emph{without} relying on the session provided by $P$ (e.g., simply $Q' \is \out v<>$).
    Now consider the following process, obtained by replacing the spawn prefix and $Q$ in $R$ from \Cref{ex:contraction}:
    \[
        R' \is \new x.(\inp z().P \|_x \spw{x->\emptyset}.Q')
    \]
    By \Cref{red-spawn}, $ R' \redd \spw{z->\emptyset}.Q' $.
    In this case, $P$ is dropped.
    Since the empty input prefix on $z$ is also dropped, an additional spawn is generated to signal to the environment that the process providing the session on $z$ should be dropped as well.
\end{example}
\Cref{red-spawn-r,,red-spawn-l,,red-spawn-merge} show how the spawn prefix interacts with independent process compositions and with other spawn prefixes, respectively.
\Cref{red-spawn-r,red-spawn-l} are forms of scope extrusion: spawn prefixes can ``bubble up'' past restrictions that do not capture their bindings, possibly enabling interactions of the spawn with processes in the outer context.
\Cref{red-spawn-merge} describes how two consecutive spawn prefixes
can be combined into a single spawn, by merging the spawn bindings, denoted $\merge$, as follows.

\begin{definition}[Merge]
\label{def:merge}
  Let $ \spb[X] \is \Union \{ \spb(x) \mid x\in X,\, x\in\dom(\spb) \}. $
  The \emph{merge} of two spawn bindings~$\spb_1, \spb_2$,
  written $ \spb_1 \merge \spb_2 $, is defined as:
  \begin{equation*}
    (\spvar_1 \merge \spvar_2)(x) \is
    \begin{cases}
      \spvar_2[\spvar_1(x)] \union (\spvar_1(x) \setminus \dom(\spvar_2))
        \CASE x \in \dom(\spvar_1)
      \\
      \spvar_2(x)
        \CASE x \notin \dom(\spvar_1) \land x \notin \restrOf(\spvar_1)
      \\
      \bot \OTHERWISE
    \end{cases}
  \end{equation*}
\end{definition}

\noindent
  Note that the merge of two independent spawn bindings is just disjoint union (as functions), and $\emptyset$ is the neutral element for $\merge$.
Merge is associative:
  $
    (\spb_1 \merge (\spb_2 \merge \spb_3))
    =
    ((\spb_1 \merge \spb_2) \merge \spb_3)
  $.

The idea behind the merge operation $\spb_1\merge\spb_2$
is to ``connect'' the outputs of~$\spb_1$ to the inputs of~$\spb_2$,
similarly to composition of relations.
However, names that are irrelevant for~$\spb_1$ should
still be subject to the mapping of~$\spb_2$,
unless they are captured by the restrictions of~$\spb_1$.
For example:
\[
  \map*{
    \begin{matrix*}[l]
    x -> \emptyset\\
    y -> \set{y_1, y_2, y_3}
    \end{matrix*}
  }
  \merge
  \map*{
    \begin{matrix*}[l]
    y_2 -> \emptyset\\
    y_3 -> \set{y_4, y_5}\\
    z -> z_1
    \end{matrix*}
  }
  =
  \map*{
    \begin{matrix*}[l]
    x -> \emptyset\\
    y -> \set{y_1, y_4, y_5}\\
    z -> z_1
    \end{matrix*}
  }
\]
This merge can be graphically illustrated as follows:
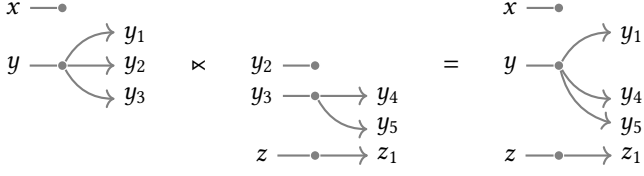
\begin{center}
  \begin{tikzpicture}[
  dot/.style={
    fill=gray,
    draw=white,
    thick,
    circle,
    minimum size=4pt,
    inner sep=0pt,
    outer sep=0pt,
    xshift=-1em,
  },
  up/.style={bend left},
  dw/.style={bend right},
]
\matrix [
  matrix of math nodes,
  column sep=1em,
  row sep=-.2em,
  nodes={anchor=center},
]{
  |(x)| x &|[dot](xm)|&          &        &           &            &            & &|(xr)|x &|[dot](xmr)|&           \\
          &           &|(y1)|y_1 &        &           &            &            & &        &            &|(y1r)|y_1 \\
  |(y)| y &|[dot](ym)|&|(y2)|y_2 & \merge &|(y2')|y_2 &|[dot](y2m)|&            &=&|(yr)|y &|[dot](ymr)|&           \\
          &           &|(y3)|y_3 &        &|(y3')|y_3 &|[dot](y3m)|&|(y4)|y_4   & &        &            &|(y4r)|y_4 \\
          &           &          &        &           &            &|(y5)|y_5   & &        &            &|(y5r)|y_5 \\
          &           &          &        &|(z)|  z   &|[dot](zm)| &|(z1)|  z_1 & &|(zr)|z &|[dot](zmr)|&|(z1r)|z_1 \\
};

\draw[gray,thick]
  (xm)  edge (x)
  (ym)  edge (y)
        edge[->,up] (y1)
        edge[->]    (y2)
        edge[->,dw] (y3)
  (y2m) edge (y2')
  (y3m) edge (y3')
        edge[->]    (y4)
        edge[->,dw] (y5)
  (zm)  edge (z)
        edge[->]    (z1)
  (xmr) edge (xr)
  (ymr) edge (yr)
        edge[->,up] (y1r)
        edge[->,dw] (y4r)
        edge[->,dw] (y5r)
  (zmr) edge (zr)
        edge[->]    (z1r)
;
\end{tikzpicture}
 \end{center}
Note how~$x$ and~$z$ are both in the domain of the result,
and how the mapping to~$y_1$ is preserved by the merge,
although it is not in the restrictions of the second binding.

The last two rules in \Cref{fig:red-main} are purely structural.
\Cref{red-eval-ctxt} closes reduction under
\emph{evaluation contexts}, denoted~$\ectxt$,
consisting of spawn prefixes and structured parallel compositions (cf.\ \Cref{fig:red-main}).
\Cref{red-cong} closes reduction
under structural congruence.

\subsection{Typing}
\label{sec:typing}

The \piBI type system is based on the BI sequent calculus, and follows the approach of \piDILL: propositions are interpreted as session types, where the context governs the \emph{use} of available channels and the conclusion governs the process' behavior on the \emph{provided} channel.
As such, the type system of \piBI uses judgments of the form $\bunch \proves P :: x : A$,
where the process $P$ provides the session $A$ on channel $x$, while using the sessions provided by the typing context $\bunch$.

\begin{mathfig}[\small]
    \rulesection*{Types, bunches, and contexts}
	\begin{grammar}
  A,B,C \is
      \mOne
    | (A * B)
    | (A \wand B)
    \explain{(multiplicatives)}
    |* \ghost{\mOne}{\aOne}
    | (A \land B)
    | (A \to B)
    | (A \vee B)
    \explain{(additives)}
  \\
  \bunch,\bunchB \is
    \mEmpty | \aEmpty | x : A | \bunch \bsep \bunch | \bunch \band \bunch
    \explain{(bunches)}
  \\
  \bctxt(\hole) \is
      (\hole)
    | \bunch \bsep \bctxt(\hole)
    | \bunch \band \bctxt(\hole)
    | \bctxt(\hole) \bsep \bunch
    | \bctxt(\hole) \band \bunch
    \explain{(bunched contexts)}
\end{grammar}
  \rulesection{Typing}
\begin{proofrules}
    \infer*[lab=Sep-r]{
    \bunch_1 \proves P_1 :: y : A
    \\
    \bunch_2 \proves P_2 :: x : B
  }{
    \bunch_1\bsep \bunch_2 \proves
    \out x[y].(P_1 \| P_2) :: x : (A * B)
  }
  \label{rule:sep-r}

  \infer*[lab=Sep-l]{
    \bctxt(x : B\bsep y : A) \proves P :: z : C
  }{
    \bctxt(x : A * B) \proves \inp x(y). P :: z : C
  }
  \label{rule:sep-l}

  \infer*[lab=Wand-r]{
    \bunch\bsep y : A \proves P :: x : B
  }{
    \bunch \proves \inp x(y).P :: x : A \wand B
  }
  \label{rule:wand-r}

  \infer*[lab=Wand-l]{
    \bunch \proves P :: y : A
    \\
    \bctxt(x : B) \proves Q :: z:C
  }{
    \bctxt(\bunch\bsep x : A \wand B) \proves \out x[y].(P \| Q) :: z:C
  }
  \label{rule:wand-l}

  \infer*[lab=Emp-r]{}{
    \mEmpty \proves \out x<> :: x : \mOne
  }
  \label{rule:emp-r}

  \infer*[lab=Emp-l]{
    \bctxt(\mEmpty) \proves P :: x:C
  }{
    \bctxt(x : \mOne) \proves \inp x().P :: x:C
  }
  \label{rule:emp-l}

 \infer*[lab=Conj-r]{
    \bunch_1 \proves P_1 :: y : A
    \\
    \bunch_2 \proves P_2 :: x : B
  }{
    \bunch_1\band \bunch_2 \proves \out x[y].(P_1 \| P_2) :: x : A \land B
  }
  \label{rule:conj-r}

  \infer*[lab=Conj-l]{
    \bctxt(x : B \band y : A) \proves P :: z : C
  }{
    \bctxt(x : A \land B) \proves \inp x(y). P :: z : C
  }
  \label{rule:conj-l}

  \infer*[lab=Impl-r]{
    \bunch \band y : A \proves P :: x : B
  }{
    \bunch \proves \inp x(y).P :: x : A \to B
  }
  \label{rule:impl-r}

  \infer*[lab=Impl-l]{
    \bunch \proves P :: y : A
    \\
    \bctxt(x : B) \proves Q :: z:C
  }{
    \bctxt(\bunch \band x : A \to B) \proves \out x[y].(P \| Q) :: z:C
  }
  \label{rule:impl-l}

  \infer*[lab=True-r]{}{
    \aEmpty \proves \out x<> :: x : \aOne
  }
  \label{rule:true-r}

  \infer*[lab=True-l]{
    \bctxt(\aEmpty) \proves P :: y:A
  }{
    \bctxt(x : \aOne) \proves \inp x().P :: y:A
  }
  \label{rule:true-l}

 \infer*[lab=Disj-r-inl]{
    \bunch \proves P :: x : A
  }{
    \bunch \proves \selL{x}.P :: x : A \vee B
  }
  \label{rule:disj-r-1}
  \label{rule:disj-r-inl}

  \infer*[lab=Disj-r-inr]{
    \bunch \proves P :: x : B
  }{
    \bunch \proves \selR{x}.P :: x : A \vee B
  }
  \label{rule:disj-r-1}
  \label{rule:disj-r-inr}

  \infer*[lab=Disj-l]{
    \bctxt(x : A) \proves P :: z : C
    \\
    \bctxt(x : B) \proves Q :: z : C
  }{
    \bctxt(x : A \vee B) \proves \caseLR{x}{P}{Q} :: z : C
  }
  \label{rule:disj-l}

 \infer*[lab=Fwd]{}{
    y : A \proves \fwd[x<-y] :: x : A
  }
  \label{rule:type-fwd}

  \infer*[lab=Cut]{
    \bunch \proves P :: x : A
    \\
    \bctxt(x : A) \proves Q :: z : C
  }{
    \bctxt(\bunch) \proves \new x.(P \| Q) :: z : C
  }
  \label{rule:cut}

  \infer*[lab=Struct]{
    \bunch_2 \proves P :: z : C
    \\
\senv \from \bunch_1 \tobunch \bunch_2
  }{
    \bunch_1 \proves \spw \senv.P :: z : C
  }
  \label{rule:struct}

  \infer*[lab=Bunch-equiv]{
    \bunch_2 \proves P :: x : C
    \\
    \bunch_2 \buncheq \bunch_1
  }{
    \bunch_1 \proves P :: x : C
  }
  \label{rule:bunch-equiv}

   \end{proofrules}

  \rulesection{Spawn binding}
  \begin{proofrules}
    \infer*[lab=spawn-contract]{}{
  \map*{x -> \set{x_1, \dots, x_n} | {x \in \bunch}}
    \from    \Gamma(\bunch)
    \tobunch \Gamma\bigl(\idx{\bunch}{1}\band \dots\band \idx{\bunch}{n}\bigr)
}
\label{rule:spawn-contract}

\infer*[lab=spawn-weaken]{}{
  \map*{x -> \emptyset | x \in \bunch_1}
  \from    \bctxt(\bunch_1\band \bunch_2)
  \tobunch \bctxt(\bunch_2)
}
\label{rule:spawn-weaken}

\infer*[lab=spawn-merge]{
  \spvar_1 \from \bunch_0 \tobunch \bunch_1
  \\
  \spvar_2 \from \bunch_1 \tobunch \bunch_2
}{
  (\spvar_1 \merge \spvar_2) \from \bunch_0 \tobunch \bunch_2
}
\label{rule:spawn-merge}
   \end{proofrules}
  \rulesectionend
  \caption{Types, typing rules and spawn binding rules for \piBI.}
  \label{fig:typing_rules}
\end{mathfig}
The top of \Cref{fig:typing_rules} gives types, bunches, and bunched contexts; we explain the session behavior associated with types when we discuss the typing rules below.
Bunches~$\bunch$ are binary trees
with internal nodes labelled with either `$\band$' or `$\bsep$', and
with leaves being either unit bunches ($\mEmpty$ or $\aEmpty$)
or typing assignments ($x : A$).
We write~$\ident(\bunch)$ for the set of names occurring in the bunch~$\bunch$, and write $x \in \bunch$ to denote $x \in \ident(\bunch)$.
As is standard for BI, we consider bunches modulo
the least congruence on bunches closed under
commutative monoid laws
for `$\bsep$' with unit $\mEmpty$,
and for `$\band$' with unit $\aEmpty$, denoted~$\buncheq$.
For example,
$
  (\bunch_1 \bsep \mEmpty) \band \bunch_2
  \buncheq
  \bunch_2 \band \bunch_1.
$

Bunched contexts $\bctxt(\hole)$ are bunches with a hole $(\hole)$.
As usual, we write $\bctxt(\bunch)$ for a bunch obtained by replacing $(\hole)$ with $\bunch$ in $\bctxt$.
We write $\bctxt(\hole \mid \dots \mid \hole)$ for a bunched context with multiple holes.

\Cref{fig:typing_rules} also gives the type system for \piBI.
We organize them in four groups:
the first six rules type communication primitives with multiplicative types, and the next six rules with additive types;
the following three rules type branching primitives using disjunction;
the final four rules type
forwarding,
structured parallel composition, and the structural rules.

One key design choice of our typing rules is that
the processes in the multiplicative and the additive groups of rules
are the same.
For example, the same send action can be typed with~$A*B$ or with $A\land B$.
Their difference lays purely in the way they manage their available resources, possibly enabling or restricting the use of \Cref{rule:struct} in other parts of the derivation.

\paragraph{Rules for multiplicative constructs}
The type~$A*B$ is assigned to a session that outputs a channel of type~$A$
and continues as~$B$.
\Cref{rule:sep-r} states that
to \emph{provide} a session of type~$A*B$ on~$x$,
a process must output on~$x$ a new name~$y$
and continue with a process providing a session of type~$A$ on~$y$ in parallel with
a process providing the continuation session~$B$ on~$x$.
\Cref{rule:sep-l} describes how to \emph{use}
a session of type~$A*B$ on~$x$:
a process must input on~$x$ a new name~$y$ which is to be used for the session of type~$A$, after which the process must provide the continuation session~$B$ on~$x$.

\Cref{rule:wand-r,rule:wand-l} describe
the type~$A\wand B$.
These rules are dual to the rules for $A*B$: providing $A \wand B$ requires an input, and using it requires an output.

\Cref{rule:emp-r} states how to close a session of type~$\mOne$ using an empty output, followed by termination.
The dual \Cref{rule:emp-l} uses the empty input prefix.
Note that \Cref{rule:emp-r} requires the context to be~$\mEmpty$,
effectively forcing processes to consume all the sessions they use before terminating.

\paragraph{Rules for additive constructs}

As already mentioned, the rules for sessions of additive type,
are identical to the ones for multiplicative types, except that the latter (de)composes bunches using `$\bsep$' while the former uses `$\band$'.
In particular, the process interpretation of the rules is identical
for both counterparts.
The difference has effect elsewhere in the derivation,
where the choice between `$\band$' and `$\bsep$' affects the possibility of using \Cref{rule:struct} (explained last).

\paragraph{Rules for disjunction}

Disjunction types branching constructs.
To provide on~$x$ a session of type~$A \lor B$, the process must select either $\inl$/$\inr$ on~$x$ and continues by providing $A$/$B$, respectively.
Using a session of type~$A \lor B$ on~$x$ requires a branching on~$x$, where the left branch uses~$x$ as~$A$ and the right branch as~$B$.
Curiously, there is no dual construct for disjunction in BI, meaning that there is no way to type a selection on a channel that is being used, or a branch on a channel that is being provided.
There is no canonical way of adding such a dual construct; there are however extensions of BI that incorporate one---see, e.g.,~\cite{docherty:2019,pym:2002,brotherston:2012,brotherston.calcagno:2010,brotherston.villard:2015}.

\paragraph{Forwarders, Cut, and structural rules}
\Cref{rule:type-fwd} types the forwarder~$\fwd[x<-y]$ as providing a session of type~$A$ on~$x$ as a copycat of a session of the same type on~$y$ in the context.
\Cref{rule:cut} connects processes $P$ and $Q$ along the channel $x$: $P$ must \emph{provide} a session of type $A$ on $x$, whereas $Q$ must \emph{use} the session of the same type on the same channel.

\Cref{rule:bunch-equiv} closes typing under bunch equivalence.
\Cref{rule:struct} extends indexed renaming (\Cref{{def:idx-renaming}}) to bunches as follows.
\begin{definition}[Indexed bunch renaming]
\label{def:idx-bunch-renaming}
  Let~$\bunch$ be a bunch with
  $ \fn(\bunch) = \set{a,b,\dots,z} $.
  Assuming $a_i,b_i,\dots,z_i \notin \fn(\bunch)$,
  we define
  $
    \idx{\bunch}{i} \is
      \bunch\subst{a->a_i,b->b_i,,z->z_i},
  $
  where $\bunch\theta$ is the bunch obtained by
applying the substitution $\theta$ to all the leaves of~$\bunch$.
\end{definition}
\Cref{rule:struct} subsumes and generalizes the two
structural rules of weakening and contraction.
To unpack the meaning of the rule,
\Cref{fig:weaken_contract_bi} gives rules
for weakening and contraction as usually presented for BI sequent calculi.
\begin{mathfig}[\small]
\begin{proofrules}
\infer*[lab=Weakening]{
    \bctxt(\bunch_2) \proves P :: z : C
    \\
    \senv = \map{x->\emptyset | x\in \bunch_1}
  }{
    \bctxt(\bunch_1 \band \bunch_2) \proves \spw \senv.P :: z : C
  }
  \label{rule:weakening}

  \infer*[lab=Contraction]{
    \bctxt(\idx{\bunch}{1} \band \idx{\bunch}{2}) \proves P :: z : C
    \\
    \senv = \map{x->x_1,x_2 | x\in \bunch}
  }{
    \bctxt(\bunch) \proves \spw \senv.P :: z : C
  }
  \label{rule:contraction}
\end{proofrules}
\caption{Usual presentations of weakening and contraction for BI sequent calculi.}
\label{fig:weaken_contract_bi}
\end{mathfig}
\Cref{rule:weakening} discards the unused resources in~$\bunch_1$.
The process interpretation is a spawn
that terminates the providers of sessions on channels in $\bunch_1$.
\Cref{rule:contraction} allows the duplication of the resources
in~$\bunch$.
These resources need to be renamed to keep the names unique,
hence the substitutions~$\idx{\bunch}{1}$ and $\idx{\bunch}{2}$ in the premise.
The process interpretation is again a spawn prefix that generates
two indexed variants of each name in~$\bunch$,
representing the duplicated resources.
For both rules, it is crucial that the affected bunches are combined using `$\band$'.

Both \Cref{rule:weakening,rule:contraction} transform bunches according to
the spawn binding of the involved names.
The idea behind \Cref{rule:struct} is to generalize weakening and contraction, and allow more general spawn bindings.
As such, the rule combines in a single application a number
of consecutive or independent applications of \Cref{rule:weakening,rule:contraction}.

To relate spawn bindings and their corresponding transformations of bunches,
we define a \emph{spawn binding typing} judgment
$\senv \from \bunch_1 \tobunch \bunch_2$;
the bottom of \Cref{fig:typing_rules} gives their rules.

The idea is to consider a binding~$\spb$
as the merge of a sequence of bindings
$\spb = \spb_1 \merge \dots \merge \spb_k$,
where each $\spb_i$ is either a weakening or a contraction binding.
The weakening and contraction bindings are typed using \Cref{rule:spawn-weaken,rule:spawn-contract}.
In case of contraction, when~$n=2$ we get pure contraction, when~$n>2$ it might represent a number of consecutive contractions applied to the same bunch; the corner case when~$n=1$ just renames the variables in the bunch, and might arise as the by-product of a contraction and a weakening (partially) canceling each other out.

\Cref{rule:spawn-weaken,rule:spawn-contract} combined with \Cref{rule:struct} offer a justification of the specialized \Cref{rule:weakening,rule:contraction}, respectively.
In the former case, the justification is direct.
The latter case holds for $n=2$, i.e., for pure contraction.

We wrap up the explanation of \Cref{rule:struct} by giving an example typing derivation.
\begin{example}
    Consider the following process, with contraction and weakening in one spawn:
    \[
        P \is \new x.(\inp z().Q \|_x \new y.(\out y<> \|_y \spw{x->x_1,x_2;y->\emptyset}.R))
    \]
    This process is well-typed, assuming $\Delta \proves Q :: x:A$ and $\Gamma(x_1:A \band x_2:A) \proves R :: v:B$, as follows:
    \begin{derivation}
        \infer*{
            \infer*{
                \infer*{
                    \Delta \proves Q :: x:A
                }{
                    \Delta \bsep \empM \proves Q :: x:A
                }
            }{
                \Delta \bsep z:\mOne \proves \inp z().Q :: x:A
            }
            \\
            \infer*{
                \infer*{
                    \empA \proves \out y<> :: y:\aOne
                    \\
                    \infer*{
                        \Gamma(x_1:A \band x_2:A) \proves R :: v:B
                        \\
                        \Psi
                    }{
                        \Gamma(x:A \band y:\aOne) \proves \spw{x->x_1,x_2;y->\emptyset}.R :: v:B
                    }
                }{
                    \Gamma(x:A \band \empA) \proves \new y.(\out y<> \|_y \spw{x->x_1,x_2;y->\emptyset}.R) :: v:B
                }
            }{
                \Gamma(x:A) \proves \new y.(\out y<> \|_y \spw{x->x_1,x_2;y->\emptyset}.R) :: v:B
            }
        }{
            \Gamma(\Delta \bsep z:\mOne) \proves P :: v:B
        }
    \end{derivation}

\noindent
        where $\Psi$ is as follows:
        \begin{derivation}
            \infer*{
                \map{x->x_1,x_2} : \Gamma(x:A \band y:\aOne) \tobunch \Gamma(x_1:A \band x_2:A \band y:\aOne)
                \\
                \map{y->\emptyset} : \Gamma(x_1:A \band x_2:A \band y:\aOne) \tobunch \Gamma(x_1:A \band x_2:A)
            }{
                (\map{x->x_1,x_2} \merge \map{y->\emptyset}) : \Gamma(x:A \band y:\aOne) \tobunch \Gamma(x_1:A \band x_2:A)
            }
        \end{derivation}
    \noindent
    Notice how the spawn binding must be split into a contracting and a weakening spawn binding to justify the transformation of the bunch.
\end{example}

It is worth noticing that the typing judgment
$\senv \from \bunch_1 \tobunch \bunch_2$
is not uniquely determined from $\senv$ and $\bunch_1$.
Hence, there is not always a unique derivation tree for a given judgment.
To recover unique typing, it should be sufficient to annotate all bindings with their respective types, including the $\senv$ in the spawn prefixes.

\paragraph{Empty spawn}
We briefly discuss a corner case:
according to the typing rules for spawn bindings, we can type the empty spawn $\spawn{\emptyset}$.
It is tempting to add a structural congruence or reduction that removes it, since an empty spawn does not do much operationally:
an empty spawn can only propagate along cuts and silently merge into other spawns.
However, adding a reduction such as $\spawn{\emptyset}.P \redd P$ will cause complications because the empty spawn prefix, though operationally vacuous, can influence the typing.
An example is the following application of weakening:
\[
  \infer*
  {\bctxt(\empA) \proves P :: x : A}
  {\bctxt(\empM) \proves \spawn{\emptyset}.P :: x : A}
\]
Thus, such a reduction might slightly change the typing of a process across reductions, disproving type preservation.
This would unnecessarily complicate the system and, arguably, would not be in line with the Curry-Howard correspondence.

The empty spawn prefixes are but a minor annoyance: reductions can still happen behind spawn prefixes.
We do have to take extra care of the empty spawn when we show deadlock-freedom in \Cref{sec:type_preservation_and_deadlock_freedom} and weak normalization in \Cref{sec:weak_normalization}.
Next, we discuss additional examples.

\subsection{Examples and Comparisons}
\label{sec:examples-comparisons}

The \piBI\ calculus is expressive enough
to represent many useful concurrency patterns.
Here we show three significant examples and contrast \piBI's
approach to related calculi. Below we write $P \redd^k Q$ to mean that $P$ reduces to $Q$ in $k > 1$ consecutive steps.

\paragraph{Server and clients}
Recall from \Cref{ex:contraction} the process $R = \new x.(\inp z().P \| \spw{x->x_1,x_2}.Q)$.
We can interpret $\inp z().P$ as a server providing a service on~$x$
while relying on another server providing a service on~$z$, and the spawn as a request for two copies of the server to be used in~$Q$ on~$x_1$ and~$x_2$.

In \piDILL and CP, servers and clients are expressed using replicated input
${!}\inp x(y).P$,
which upon receiving a channel~$y$ replicates~$P$ to provide its session on~$y$.
A client must then explicitly request a copy of the server by sending a fresh channel over~$x$.
The \piDILL analog of~$R$ would then be
$
  R' \is
    \new u.(
      {!}\inp u(x).\out z[z'].\inp z'().P \|
      \out u[x_1].\out u[x_2].Q
    )
$.
In general, \piDILL's servers and clients can be expressed in \piBI by removing the replicated inputs (\ie ${!}\inp x(y).P$ becomes $P\subst{y->x}$) and replacing request outputs with spawns (i.e., $\out x[x_1].Q$ becomes $\spw{x->x_1,x_2}.Q\subst{x->x_2}$).

There is a crucial difference in the two models of servers:
in \piDILL, the server itself is responsible for creating a new instance
of the session it provides, and thus needs to make sure
that the sessions on which the new instance depends are themselves provided by servers.
In \piBI\ the responsibility for duplication lies with the client;
the server does not need to make special arrangements to allow for duplication,
and its dependencies are duplicated on-the-fly by the spawn semantics.

The on-the-fly nature of spawn propagation makes the server/clients pattern
more concurrent in \piBI than in \piDILL.
Suppose we connect~$R$ to a process providing~$z$.
The communication on~$z$ can take place before the spawn reduction,
such that the spawn no longer needs to propagate to~$z$:
\[
  \new z.(\out z<> \| R) \redd \new x.(P \| \spw{x->x_1,x_2}.Q).
\]
This is not possible in \piDILL: the replicated input of the server is blocking the communication on~$z$.

\paragraph{Failures}
An important aspect of (distributed) programming is coping with failure.
For example, consider $P \is \inp x(y).\inp x().\out z[w].(\fwd[w<-y] \| \inp z().\out v<>)$, i.e., a process that receives a channel $y$ over $x$ and forwards it over $z$.
Suppose that the process providing $x$ is unreliable, and might not be able to send the channel $y$.
This provider process indicates availability by a selection on $x$: left means availability and right means the converse.
We can then embed $P$ in a branch on $x$, where the right branch propagates the failure to forward a channel by means of spawn: $P' \is \caseLR{x}{P}{\inp x().\spw{z->\emptyset}.\out v<>}$.
Let $\inp z(q).R$ denote the process providing the session on $z$, which expects to receive a channel.
The following is an example where the behavior on $x$ is indeed available:
\begin{align*}
& \new z.(\inp z(q).R \| \new x.(\selL{x}.\out x[u].(\out u<> \| \out x<>) \| P')
\\
 \redd^3~ & \new z.(\inp z(q).R \| \new u.(\out u<> \| \out z[w].(\fwd[w<-u] \| \inp z().\out v<>)))
 \\
\redd^3~ & \new u.(\out u<> \| \new z.(R\subst{q->u} \| \inp z().\out v<>))
    \end{align*}
In contrast, in the following example the behavior on $x$ is not available:
$$\new z.(\inp z(q).R \| \new x.(\selR{x}.\out x<> \| P')) \redd^2 \new z.(\inp z(q).R \| \spw{z->\emptyset}.\out v<>) \redd \spw{\emptyset}.\out v<>$$

The principle sketched in this example is inspired by the typed framework by~\citet{caires.perez:2017}, which supports communication primitives for non-deterministically available or unavailable behavior via a Curry-Howard interpretation of Classical LL with dedicated modalities.

\paragraph{Interaction between session delegation and spawn}
Session delegation (also known as higher-order session communication)
is the mechanism that enables to exchange channels themselves over channels,
dynamically changing the communication topology. In \piBI, delegation interacts with spawn,
in that changing process connections influences the propagation of spawn.
Let $P \is \new x.(\out x[y].(\out y<> \| \out x<>) \| \new z.(\inp x(w).\inp x().\inp w().\out z<> \| \spw{z->\emptyset}.\out v<>))$.
From $P$, we could either reduce the spawn prefix or synchronize on $x$.
If we first reduce the spawn, the spawn propagates to $x$: $$P \redd \new x.(\out x[y].(\out y<> \| \out x<>) \| \spw{x->\emptyset}.\out v<>).$$
However, if we first synchronize on $x$, the spawn propagates to the delegated channel $y$: $$P \redd^2 \new y.(\out y<> \| \new z.(\inp y().\out z<> \| \spw{z->\emptyset}.\out v<>)) \redd \new y.(\out y<> \| \spw{y->\emptyset}.\out v<>).$$

\paragraph{Incomparability with \piDILL}
As shown by~\citet{ohearn:2003},
DILL and BI are incomparable.
Examining two canonical distinguishing examples can shed some light
on the fundamental differences of the two logics,
and their interpretations as session type systems.

As we remarked in \cref{sec:intro},
DILL admits a ``number of uses'' interpretation,
where linear resources have to be used exactly once.
This interpretation is not supported by BI:
\begin{example}
\label{ex:unusual}
  In \piBI it is possible to input linearly (\ie with~$\wand$)
  a session and use it twice.
  The process
  $
    P \is
      \inp z(a).
      \inp z(y).
      \spw a->a_1,a_2.
      \out y[a_1'].
      \bigl(
        \fwd[a_1' <- a_1] \|
        \out y[a_2'].
        (
          \fwd[a_2' <- a_2] \|
          \fwd[z<-y]
        )
      \bigr)
  $
can be typed
  as providing a session~$A \wand (A \to A \to B) \to B$ on~$x$:
  \begin{derivation}
    \infer*{
    \infer*{
    \infer*{
    \infer*{
      \infer*{}{
        a_1:A  \proves \fwd[a_1' <- a_1] :: a_1' : A
      }
\and
\infer*{
        a_2:A  \proves \fwd[a_2' <- a_2] :: a_2' : A
\and
y:B \proves \fwd[z<-y] :: z : B
      }{
        a_2:A \band y:A \to B \proves
          \out y[a_2'].
          (
            \fwd[a_2' <- a_2] \|
            \fwd[z<-y]
          )
        :: z : B
      }
    }{
      a_1:A \band a_2:A \band y:A \to A \to B \proves
        \out y[a_1'].
        \bigl(
          \fwd[a_1' <- a_1] \|
          \out y[a_2'].
          (\dots)
        \bigr)
      :: z : B
    }
      \and
}{
      a:A \band y:A \to A \to B \proves
        \spw a->a_1,a_2.
        \bigl(
          \out y[a_1'].
          (\dots)
        \bigr)
      :: z : B
    }}{
      a:A \proves
        \inp z(y).
        \spw a->a_1,a_2.
        (\dots)
      :: z : (A \to A \to B) \to B
    }}{
      \mEmpty \proves
        \inp z(a).
        \inp z(y).
        \spw a->a_1,a_2.
        (\dots)
      :: z : A \wand (A \to A \to B) \to B
    }
  \end{derivation}
The process receives a single session of type~$A$ over~$a$
  through linear input.
  The session type of~$y$ inputs~$A$ twice,
  but allows these two $A$-typed sessions to share a common origin.
  The process can thus spawn two copies of $a:A$
  and use them to interact with~$y$.
\end{example}

The corresponding LL proposition
$ A \lolli (A \to A \to B) \to B $
is not derivable:
LL forbids using twice a resource obtained through linear input.
However, the notion of linearity in \piBI has a more subtle reading:
it restricts the origin of sessions.
In \cref{ex:unusual}, the use of~$\to$ allows
the duplication of the session at~$a$ into its copies~$a_1$ and~$a_2$;
this information about the ``origin'' of $a_1$ and $a_2$ is recorded in the bunch by the use of `$\band$'.

On the other hand, there are types provable in DILL that are not
provable in BI.
A simple example is $ A \lolli B \proves A \to B $,
converting an implication from linear to non-linear.
A ``number of uses'' interpretation of the conversion makes sense:
$A \lolli B$ promises to use~$A$ exactly once to produce~$B$;
$A \to B$ declares to produce~$B$ using~$A$ an unspecified number of times,
including exactly once.
The corresponding judgment~$ A \wand B \proves A \to B $
is not provable in BI (and thus in \piBI).
Intuitively, this is because $A \to B$ allows~$A$ to be obtained with
resources which share their origin with the resource~$A \wand B$;
however, $A \wand B$ can only be applied to resources that do not share its own origin.

\paragraph{The meaning of multiplicative and additive types}
A natural question arises: if the process interpretation of multiplicatives and
additives coincides, what is the difference in the types representing behaviorally?
The following example addresses the difference between linear and non-linear connectives;
in \cref{sec:denot} we formally elucidate this difference by giving a denotational semantics
which allows tracking the \emph{origin} of sessions.

\begin{example}
\label{ex:db-example}
  Assume an opaque base type~$\Data$ of data.
  The type of a stylized database could be
  $
    \DB \is (\Data \to \DB) \land (\Data \land \DB)
  $
where the first conjunct can receive some new data to overwrite the contents of the database  (the `put' operation), and the second  would provide the current data stored in it (the `get' operation).
  This is a recursive type, which is not currently supported by our calculus; for the purposes of this discussion,
  it is enough to consider some finite unfolding of the type
  (terminated with~$\aOne$).

  Just by looking at the type $\DB$,
  we can identify possible interactions with the database.
  A typical usage pattern of a resource~$\var{db} : \DB$ would be
  to input the `put' and the `get' components and weaken the
  one we are not intending to use in the current step.
  Imagine we want to put some $d : \Data$:
  then we would weaken the `get',
  and send~$d$ over $ \var{put} : (\Data \to \DB) $ to obtain
  a continuation of type~$\DB$ that represents the updated database.

  A second pattern of usage afforded by \piBI\ is to use contraction
  to spawn independent snapshots of the database.
  For example, using contraction we can obtain,
  from $ \var{db} : \DB  $, a copy $ \var{db}' : \DB  $.
  From then on, the two copies can be mutated independently without interference.

  Now consider two different \piBI processes, $P_{\lbl{a}}$ and $P_{\lbl{m}}$, with judgments:
  \begin{align*}
    \var{db}_1 : \DB \band \var{db}_2 : \DB \proves P_{\lbl{a}} :: z : C
    &&
    \var{db}_1 : \DB \bsep \var{db}_2 : \DB \proves P_{\lbl{m}} :: z : C
  \end{align*}
  $P_{\lbl{a}}$ has access to two databases that are allowed to ``overlap''
  since they are aggregated by a `$\band$'.
  In contrast, $P_{\lbl{m}}$ has access to two non-overlapping databases.
  Here ``overlapping'' has a subtle meaning:
  it refers to the \emph{provenance} of the data stored in the two databases,
  rather than the stored value itself.
  To see the difference concretely, imagine we interact, in both cases,
  with $\var{db}_1$ by weakening the `get',
  and with $\var{db}_2$ by weakening the `put' (and the continuation of `get'):
  \begin{align*}
    \var{put}_1 : \Data \to \DB \band \var{d} : \Data
      \proves P'_{\lbl{a}} :: z : C
    &&
    \var{put}_1 : \Data \to \DB \bsep \var{d} : \Data
      \proves P'_{\lbl{m}} :: z : C
  \end{align*}
  Process~$P'_{\lbl{a}}$ is now allowed to send~$d$ on channel $\var{put}_1$,
  updating the database's value to~$d$, thus inducing a flow of information
  from $\var{db}_2$ to $\var{db}_1$.
  This flow is however forbidden in the case of $P'_{\lbl{m}}$:
  the data sent through $\var{put}_1$ needs to be obtained from a resource
  that is separated with it by `$\band$' as per Rule~\ref{rule:impl-l}.
  The fact that~$d$ is separated using~`$\bsep$' fundamentally forbids
  it to flow into~$\var{put}_1$.

  Now suppose $ C = \Data * \Data $ and take $\DB$ to be the 1-unfolding of the recursive definition.
  The typing of~$P_{\lbl{m}}$ ensures that the two data values sent on
  the channel~$z$ would come one from~$\var{db}_1$ and the other from~$\var{db}_2$; the combinations where two values taken from the same database are sent on $z$ are disallowed by typing.
  As we will see in \cref{ex:provenance,ex:db-reprise},
  the denotational semantics developed in \cref{sec:denot}
  formally justifies these claims.
\end{example}

 \section{Meta-theoretical Properties}
\label{sec:meta}
A distinguishing feature of the propositions-as-sessions approach is that the main meta-theoretical properties of session-typed processes (e.g., type preservation and deadlock-freedom) follow immediately from the cut elimination property in the underlying logic.
In this section we show that \piBI satisfies these properties, which serves to validate the appropriateness of our interpretation.
We consider type preservation and deadlock-freedom, but also \emph{weak normalization}.
\appendixref{app:meta_props} gives additional properties and detailed proofs.

\subsection{Type Preservation and Deadlock-Freedom}
\label{sec:type_preservation_and_deadlock_freedom}
Essential correctness properties in session-based concurrency are that (i)~processes correctly implement the sessions specified by its types (session fidelity) and (ii)~there are no communication errors or mismatches (communication safety).
Both these properties follow from the \emph{type preservation} property, which ensures that typing is consistent across structural congruence and reduction.
\begin{restatable}{theorem}{subjectRed}
  \label{thm:subject_red}
  If $\Delta \proves P :: x:C$, then $P \congr Q$ and $P \redd Q$ imply $\Delta \proves Q :: x : C$.
\end{restatable}

\noindent
The theorem above is a consequence of the tight correspondence between \piBI and the BI proof theory, as structural congruence and reduction of typed processes correspond to proof equivalences and (principal) cut reductions in the BI sequent calculus (see \appendixref{app:subject_red:proof} for details).

Another important correctness property is \emph{deadlock-freedom}, the guarantee that processes never get stuck waiting on pending communications.
In general, deadlock-freedom holds for well-typed \piBI processes where all names are bound, except for the provided name, which must be used only to close a session.
Any process satisfying these typing conditions can then either reduce, or it is \emph{inactive}: only the closing of the session on the provided name is left, possibly prefixed by an empty spawn.
Because of bunches, a process with all names bound but one is typable in more ways than just under an empty typing context:
\begin{definition}[Empty bunch]
  \label{def:empty_bunch}
  An \emph{empty bunch}~$\emptybunch$ is a bunch such that
  $\ident(\emptybunch) = \emptyset$.
  Equivalently, a bunch is empty if each of its leaves is $\mEmpty$ or $\aEmpty$.
\end{definition}

\begin{restatable}[Deadlock-freedom]{theorem}{dlfreedomThm}
\label{t:dlfreedom}
    Given an empty bunch $\emptybunch$, if\/ $\emptybunch \proves P :: z : A$ with $A \in \set{\mOne,\aOne}$, then either (i) $P \congr \out z<>$, or (ii) $P \congr \spawn{\emptyset}.\out z<>$, or (iii) there exists $S$ such that $P \redd S$.
\end{restatable}

\noindent
The property stated above is an important feature of \piBI derived from its logical origin.
The \piBI interpretation of \Cref{rule:cut} combines restriction and parallel, ensuring that parallel processes never share more than one channel and thus preventing processes such as $\new x.\new y.(\inp y().\out x<> \| \inp x().\out y<>)$ where the subprocesses are stuck waiting for each other.
The proof follows from a property that we call \emph{progress}, which  ensures that processes of a given syntactical shape can reduce.
Although weak by itself, this property is useful in providing a reduction strategy for practical implementation of \piBI.
Moreover, it simplifies the proof of deadlock-freedom (given in \appendixref{app:dlfree:proof}), which reduces to proving that processes typable under empty bunches are in the right syntactical shape to invoke progress.

\subsection{Weak Normalization}
\label{sec:weak_normalization}
We now turn our attention to proving that our calculus is weakly normalizing, that is, for every process $P$ there exists some process $Q$ such that $P \reddStar Q \nredd$.
This is a result of independent interest, which we will use to show soundness of denotational semantics in \Cref{sec:denot}.
The normalization proof that we give here is of combinatorial nature.
Before writing out the necessary auxiliary definitions and lemmas, we first outline the main ideas.

Given a process $P$, what kind of reductions can $P$ make and can we come up with some kind of measure that would strictly decrease and disallow infinite reduction sequences?
If we did not have the spawn prefix, then the answer to this problem would be simple: each reduction is an instance of communication (or a forwarder reduction), which decreases the total number of communication prefixes in the process.
However, in presence of spawn, counting the total number of prefixes does not work.
For example, consider the following reduction, where $\fn(R) = \{x,y\}$,
\begin{equation}\label{eq:dummy_spawn_red}
\new x.\big(R \| \spawn{x -> x_1,x_2}.Q \big)
\redd
\spawn{y -> y_1,y_2}. \new x_1. (\idx{R}{1} \| \new x_2.(\idx{R}{2} \| Q)).
\end{equation}
In this reduction the prefixes in the sub-process $R$ get duplicated, so the total number of prefixes increases.
What has also changed is that the spawn prefix $\spawn{x->x_1,x_2}$ turned into the prefix $\spawn{y -> y_1,y_2}$ with a larger scope.
As a result, the communication prefixes in $Q$ went from being guarded directly by $\spawn{x->x_1,x_2}$, to being guarded by a prefix $\spawn{y -> y_1,y_2}$, with the latter prefix being ``smaller'' in the sense that it is closer to the top-level of the process.

Furthermore, if the reduction (\ref{eq:dummy_spawn_red}) occurs in some evaluation context $\ectxt$, then we can use \Cref{red-spawn-r,red-spawn-l} to actually propagate the spawn prefix to the top-level:
\begin{align}
\label{eq:dummy_spawn_red_two}
\ectxt[\new x.(R \| \spawn{x -> x_1,x_2}.Q)]
& \redd
\ectxt[\spawn{y -> y_1,y_2}. \new x_1. (\idx{R}{1} \| \new x_2.(\idx{R}{2} \| Q))] \\
& \reddStar
\spawn{y -> y_1,y_2}.\ectxt[\new x_1. (\idx{R}{1} \| \new x_2.(\idx{R}{2} \| Q))],
\nonumber
\end{align}
assuming $\ectxt$ has no other spawn prefixes that would interfere with $\spawn{y -> y_1,y_2}$.

Following this observation, the trick is to stratify the number of prefixes at each \pre\spawnsym-\emph{depth}, which is the number of spawn prefixes behind which the said prefix occurs.
So, if we examine the previous reduction sequence (\ref{eq:dummy_spawn_red_two}) and ignore the top-level spawn prefix,
the communication prefixes in $Q$ went from being at depth $n+1$ to being at depth $n$.
While the number of prefixes at depth $n$ has increased, the number of prefixes at depth $n+1$ has decreased.
This suggests that we should consider
a progress measure that aggregates the number of prefixes,
giving more weight to prefixes at greater \pre\spawnsym-depths.

Our reduction strategy for weak normalization is then as follows.
If a process can perform a communication reduction or a forwarder reduction, then we do exactly that reduction.
If a process can only perform a reduction that involves a spawn prefix, then we
\begin{enumerate*}
\item select (an active) spawn prefix with the least depth;
\item perform the spawn reduction;
\item propagate the newly created spawn prefix to the very top-level, merging it with other spawn prefixes along the way.
\end{enumerate*}

To show that this reduction strategy terminates, we adopt a measuring function that assigns to each process $P$ a finite mapping $\mu(P) : \Nat \to \Nat$ assigning to each number $n$ the number of communication prefixes at depth $n$ and above.
In order to handle the special case of a top-level prefix, the measure function simply skips it, i.e. $\mu(\spawn{\spvar}.P) = \mu(P)$ for a top-level $\spawn{\spvar}$.
We then define an ordering $\skelLt$ on such mappings
which prioritizes the number of prefixes at greater depths,
and show that it is well-founded.

Then, we argue that each clause of our reduction strategy strictly decreases the measure.
Since the relation $\skelLt$ is well-founded, it guarantees that our strategy terminates.
If we perform a communication reduction, then the number of communication prefixes at a given depth decreases, which strictly decreases the measure.
If we perform a spawn reduction, then the number of prefixes at some depth $n+1$ might decrease, but the number of prefixes at depth $n$ might increase, because of the propagated spawn prefix.
In this case, we keep propagating the spawn prefix to the top-level as much as possible, either leaving it at the top-level (to be skipped by the measure function), or merging it with an existing top-level prefix.
In both cases, the maximal prefix depth of the process decreases, which results in a strictly decreased measure.

Due to space limitations, we refer the interested reader to \appendixref{appendix:sec:wn} for the full details.
\begin{restatable}{theorem}{weakNormThm}
  \label{thm:weak_norm}
  If $\bunch \proves P :: z : A$ is a typed process, then $P$ is weakly normalizing, i.e., there exists some $Q$ such that $P \reddStar Q \nredd$.
\end{restatable}

\Cref{thm:weak_norm} thus captures the fact that, starting from a process $P$, different reductions may be applicable, or that there might be multiple spawn prefixes that can be brought to the top-most level.

  Strictly speaking, we do not require well-typedness assumptions for establishing weak  normalization; this property is enforced by the reduction semantics.
    This is a pleasant consequence of our design for the syntax of processes, which already incorporates some of the structure imposed by typing; this structure is then preserved via the correspondence between commuting conversions and reductions.
  As such, even the untyped processes are ``well-scoped'' in the sense that they conform to the tree-like structure typical of session-based interpretations of intuitionistic logics.

The weak normalization theorem is related to cut elimination in BI, but the two theorems are not equivalent.
The main discrepancy lies in the fact that not all cut reductions in BI correspond to reductions of \piBI processes;
process reductions correspond to reductions of cuts which are not guarded by an input or an output prefix.
Consecutively, we cannot directly adopt the usual cut elimination procedure for BI \cite{arisaka.qin:2012} for the purposes of showing weak normalization.

 \section{Translating the \texorpdfstring{\alcalc}{alpha-lambda-calculus} into \texorpdfstring{\piBI}{piBI}}
\label{sec:translation}
The \alcalc is a functional calculus that is in a Curry-Howard correspondence with the natural deduction representation of BI~\cite{ohearn:2003,pym:2002}.
Here we develop a type-preserving translation from the \alcalc to \piBI, and establish its correctness in a very strong sense: the translation satisfies an \emph{operational correspondence} property, which asserts how reduction steps in the source and target calculi are preserved and reflected (cf.\ \Cref{thm:completeness,thm:soundness}, respectively).

\subsection{The \texorpdfstring{\alcalc}{Alpha-Lambda-Calculus} and its Translation into \texorpdfstring{\piBI}{piBI}}
\begin{mathfig}[\small]
\adjustfigure
\rulesection*{Type system}
\begin{proofrules}
\infer*[lab=N-id]
{}
{x : A \vdash x : A}
\label{rule:N-id}

\infer*[lab=N-W]
{\Gamma(\Delta) \vdash M : A}
{\Gamma(\Delta\band \Delta') \vdash M : A}
\label{rule:N-W}

\infer*[lab=N-C]
{\Gamma(\Delta^{(1)}\band \Delta) \vdash M : A}
{\Gamma(\Delta) \vdash \substS{M}{\ident(\Delta^{(1)})}{\ident(\Delta)} : A}
\label{rule:N-C}

\infer*[lab=$\wand$I]
{\Delta\bsep x : A \vdash M : B}
{\Delta \vdash \Lam x. M : A \wand B}
\label{rule:wand-intro}

\infer*[lab=$\to$I]
{\Delta\band x : A \vdash M : B}
{\Delta \vdash \Alp x. M : A \to B}
\label{rule:impl-intro}

\infer*[lab=$\wand$E]
{\Delta_1 \vdash M : A \wand B
\and \Delta_2 \vdash N : A}
{\Delta_1\bsep \Delta_2 \vdash M\ N : B}
\label{rule:wand-elim}

\infer*[lab=$\to$E]
{\Delta_1 \vdash M : A \to B
\and \Delta_2 \vdash N : A}
{\Delta_1\band \Delta_2 \vdash M @ N : B}
\label{rule:impl-elim}
\end{proofrules}

 \rulesection{Reduction rules}
\begin{proofrules}
  \infer*[lab=red-beta-$\lambda$]
  {}
  {(\Lam x. M)\ N \alredto \substS{M}{x}{N}}
  \label{rule:red-beta-lam}

  \infer*[lab=red-beta-$\alpha$]
  {}
  {(\Alp x. M) @ N \alredto \substS{M}{x}{N}}
  \label{rule:red-beta-alp}

  \infer
  {M \alredto M'}
  {M\ N \alredto M'\ N}

  \infer
  {M \alredto M'}
  {M @ N \alredto M' @ N}
\end{proofrules}
 \caption{Selected rules of the type system and reduction rules for the \alcalc.}
\label{fig:alpha_lambda_typing}
\end{mathfig}
We first recall the statics and dynamics of the \alcalc.
Our formulation of the type system is based on the presentations by~\citet{ohearn:2003} and \citet[Chapter 2]{pym:2002}.

We use $M,N,L,\ldots$ for terms, and $a,b,c,\ldots,x,y,z,\ldots$ for variables.
The \alcalc is based on the $\lambda$-calculus, but with two separate kinds of function binders: $\Lam x.M$ with its corresponding function application $M\ N$ for the magic wand $A \wand B$, and $\Alp x.M$ with its corresponding function application $M @ N$ for the intuitionistic implication $A \to B$.
Selected typing rules are given in the top of \Cref{{fig:alpha_lambda_typing}}; the full type system can be found in \appendixref{sec:appendix:translation}.

We write $\fv(M)$ to denote the free variables of $M$.
As usual, substitution of a term $N$ for a variable $x$ in a term $M$ is denoted $\substS{M}{x}{N}$.
We write $M\subst{x_1->N_1,,x_n->N_n}$ for the sequence of substitutions $M\subst{x_1->N_1}\ldots\subst{x_n->N_n}$.
The reduction semantics of the \alcalc, denoted $\alredto$, follows a call-by-name  strategy for the $\lambda$-calculus, extended to cover two kinds of function binders.
Selected reduction rules are given in the bottom of \Cref{fig:alpha_lambda_typing}.
\paragraph{Typed translation}
\begin{figure}[t]
  \adjustfigure[\footnotesize]
  \begin{tabular}{@{}cc@{}}
  \toprule
  \alcalc typing of~$M_0$ & \piBI encoding~$\bTrans_z(M_0)$
\\ \midrule
  $ \infer{}{x : A \proves x : A} $
  &
  $ \infer{}{x : A \proves \fwd[z<-x] :: z : A} $
  \\[.7em]$ \infer{
    \Gamma(\Delta) \proves M : A
  }{
    \Gamma(\Delta \band \Delta') \proves M : A
  } $
  &
  $ \infer{
    \Gamma(\Delta) \proves \Trans_z(M) :: z : A
  }{
    \Gamma(\Delta\band \Delta') \proves
    \spawn{x -> \emptyset | x \in \ident(\Delta')}. \Trans_z(M) :: z : A
  } $
  \\[1.7em]$ \infer{
    \Delta \bsep x : A \proves M : B
  }{
    \Delta \proves \Lam x. M : A \wand B
  } $
  &
  $ \infer{
    \Delta \bsep x : A \proves \Trans_z(M) :: z : B
  }{
    \Delta \proves \inp z(x).\Trans_z(M) :: z : A \wand B
  } $
  \\[1.7em]$ \infer{
    \Delta \band x : A \proves M : B
  }{
    \Delta \proves \Alp x. M : A \to B
  } $
  &
  $ \infer{
    \Delta \band x : A \proves \Trans_z(M) :: z : B
  }{
    \Delta \proves \inp z(x).\Trans_z(M) :: z : A \to B
  } $
  \\[1.7em]$ \infer{
    \Delta_1 \proves M : A \wand B
    \and
    \Delta_2 \proves N : A
  }{
    \Delta_1 \bsep \Delta_2 \proves M\ N : B
  } $
  &
  $ \infer{
    \Delta_1 \proves \Trans_x(M) :: x : A \wand B
    \and
    \infer*{
      \Delta_2 \proves \Trans_y(N) :: y : A
      \and
      x : B \proves \fwd[z<-x] :: z : B
    }{
      x : A \wand B \bsep \Delta_2 \proves
      \out x[y].\big(\Trans_y(N) \| \fwd[z<-x]\big) :: z : B
    }
  }{
    \Delta_1 \bsep \Delta_2 \proves
    \new x.\bigl(
      \Trans_x(M) \| \out x[y].(\Trans_y(N) \| \fwd[z<-x])
    \bigr) :: z : B
  } $
  \\[1.7em]$ \infer{
    \Delta_1 \proves M : A \to B
    \and
    \Delta_2 \proves N : A
  }{
    \Delta_1 \band \Delta_2 \proves M @ N : B
  } $
  &
  $ \infer{
    \Delta_1 \proves \Trans_x(M) :: x : A \to B
    \and
    \infer*{
      \Delta_2 \proves \Trans_y(N) :: y : A
      \and
      x : B \proves \fwd[z<-x] :: z : B
    }{
      x : A \to B \band \Delta_2 \proves
      \out x[y].\big(\Trans_y(N) \| \fwd[z<-x]\big) :: z : B
    }
  }{
    \Delta_1 \band \Delta_2 \proves
    \new x.\bigl(
      \Trans_x(M) \| \out x[y].(\Trans_y(N) \| \fwd[z<-x])
    \bigr) :: z : B
  } $
\\ \bottomrule
\end{tabular}

   \caption{Translation from \alcalc to \piBI (selected clauses).}
  \label{fig:translation}
\end{figure}
Given a typed term $\Gamma \vdash M : A$ and a variable $z \notin \fv(M)$, we inductively translate the typing derivation of $M$ to a \piBI typing derivation, denoted $\Gamma \proves \Trans_z(\Gamma \proves M : A) :: z : A$.
As customary in translations of $\lambda$ into $\pi$ (cf.\ \cite{milner:1992,sangiorgi.walker:2003,wadler:2014}), the parameter $z$ is a name on which the behavior of the source term $M$ is made available.
By abuse of notation, we often write $\Gamma \proves \Trans_z(M) :: z : A$.
The translation is inspired by a canonical translation of proofs in natural deduction from into sequent calculus from (cf.~\cite[Section 6.3]{pym:2002}), and it is type-preserving by construction.
The translations of selected rules from \cref{fig:alpha_lambda_typing} is given in \cref{fig:translation}.
The identity derivation is translated into a forwarder, and the introduction rules are translated using right rules for the associated connectives.
The elimination rules are translated using the corresponding left rule in combination with a cut.
The weakening and contraction rules, which use implicit substitutions in \alcalc, are translated explicitly using the \ref{rule:struct} rule.

\begin{example}
  Consider the following \alcalc\ derivation for the term $M \is \Lam a. \Alp y. (y@a)@a$:
  \begin{derivation}
    \infer*{
    \infer*{
    \infer*{
    \infer*{
      \infer*{}{
        a_1:A  \proves a_1:A
      }
\and
\infer*{
        a_2:A  \proves a_2:A
\and
y:A \to A \to B \proves y:A \to A \to B
      }{
        a_2:A \band y:A \to A \to B \proves
        y@a_2 : A \to B
      }
    }{
      a_1:A \band a_2:A \band y:A \to A \to B \proves
      (y@a_2)@a_1 : B
    }}{
      a:A \band y:A \to A \to B \proves
      (y@a)@a : B
    }}{
      a:A \proves
        \Alp y. (y@a)@a
      : (A \to A \to B) \to B
    }}{
      \mEmpty \proves
        M \is \Lam a. \Alp y. (y@a)@a
      : A \wand (A \to A \to B) \to B
    }
  \end{derivation}
  The translation of  $M$ into \piBI is \[\Trans_z(M) = \inp z(a).\inp z(y).\spw{a->a_1,a_2}.\new x.\begin{array}[t]{@{}l@{}}
    (\new w.(\fwd[w<-y] \| \out w[a'_2].(\fwd[a'_2<-a_2] \| \fwd[x<-w])) \\
    {} \| \out x[a'_1].(\fwd[a'_1<-a_1] \| \fwd[z<-x])).
\end{array}\]
  This corresponds to the \piBI derivation
  in \cref{ex:unusual}, modulo additional cuts on forwarders due to the translation of variables and function applications.
\end{example}

\subsection{Operational Correspondence}
Here we show that the translation $\Trans_z(-)$ preserves and reflects behavior of processes and terms.
We formulate this important property in terms of an \emph{operational correspondence} result, following established criteria (cf.\ \cite{peters:2019,gorla:2010}).
Concretely, we establish the result in two parts: \emph{completeness} and \emph{soundness}.
The former states that reduction of \alcalc terms induces corresponding reductions of their process translations into \piBI; conversely, the latter states that reductions of translated terms are reflected by corresponding reductions of the source terms in the \alcalc.
\appendixref{sec:appendix:translation} gives detailed proofs.

\subsubsection{Completeness}
For completeness, we want to mimic every \alcalc reduction with one or multiple \piBI reductions.
That is, we would like to show that the translation induces a simulation.
To accurately characterize this, we need to address the discrepancy between the way the substitutions and function application are handled in \alcalc and in \piBI.
Unfortunately, the reductions of the translated term (a \piBI process) might diverge from the source term, due to the way the substitution and function application are handled in the \alcalc.
A function application $(\Alp x.M)\ N$ results in a term $\substS{M}{x}{N}$ with a substitution.
If there are multiple occurrences of $x$ in $M$---which is possible due to contraction---, they all get substituted with $N$.
On the \piBI side, substitution is represented as a composition $\new x. (\Trans_x(N) \| \Trans_z(M))$, in which \emph{one copy} of $\Trans_x(N)$ gets connected with the body $\Trans_z(M)$ through the endpoint $x$.
The contraction of the multiple occurrences of $x$ in $M$ is handled with a spawn prefix in $\Trans_z(M)$.
To address this discrepancy, we formulate completeness in a generalized way: following the approach by~\citet{toninho.etal:2012}, we define a  \emph{substitution lifting} relation which we show to be a simulation.
\begin{definition}[Substitution lifting]
    Given a term $M$ and a process $P$ of the same typing, we say $P$ lifts the substitutions of $M$, denoted $\bunch \proves P \sublift M :: z : A$, or $P \sublift M$  for short, if:
    \begin{enumerate}
        \item $P \equiv \spw{\spvar_s}.\ldots.\spw{\spvar_1}.\new x_n.(\Trans_{x_n}(N_n) \| \ldots \new x_1.(\Trans_{x_1}(N_1) \| \Trans_z(M')) \ldots)$ where for each $i \in [1,s]$, $\spvar_i = \map{y_1->\emptyset;\ldots;y_m->\emptyset}$ (only weakening) or $\spvar_i = \map{y_1->y_1,y'_1;\ldots;y_m->y_m,y'_m}$ (only contraction);
        \item $M = M'\subst{x_1 -> N_1,,x_n -> N_n}\map{\tilde{\spvar_1},\ldots,\tilde{\spvar_s}}$ where for each $i \in [1,s]$, the substitution $\tilde{\spvar_i}$ denotes a substitution corresponding to the spawn binding $\spvar_i$.
          Specifically, $\tilde{\spvar_i}$ is an empty substitution if $\spvar_i$ is weakening,
          and is the substitution $\subst{y'_1->y_1,,y'_m->y_m}$ if $\spvar_i$ is contraction.
    \end{enumerate}
    That is, both $P$ and $M$ are composed of $n$ cuts with (the translations of) the terms $M,N_1,\dots, N_n$.
Note that for any well-typed $N$ we have $\bTrans_z(N) \sublift N$.
\end{definition}
\noindent We then show the completeness result.
\begin{restatable}[Completeness]{theorem}{completenessThm}
  \label{thm:completeness}
  Given $\bunch \proves M : A$ and $\bunch \proves P :: z : A$ such that $P \sublift M$, if $M \alredto N$, then there exists $Q$ such that $P \reddStar Q \sublift N$.
\end{restatable}

\subsubsection{Soundness}
The completeness theorem shows that the reductions of terms are preserved by the translation.
We now show that reductions of translated processes are reflected by reductions of source terms.
There is a caveat, though: the translated processes are ``more concurrent'', and have more possible reductions that cannot  be immediately matched in source terms.
\begin{example}
    For some term $M\subst{x->N}$ and a corresponding substitution-lifted process $P \is \new x. (\Trans_x(N) \| \Trans_z(M))$,
    suppose that the subterm $N$ has a reduction $N \alredto N'$.
    The process $P$ can mimic this reduction:
    \[
    P \redd \new x. (Q \| \Trans_z(M)),
    \]
    for some $Q$.
    However, we do not necessarily have a corresponding reduction $M\subst{x->N} \alredto M\subst{x -> N'}$, since the variable $x$ might occur at a position where it is not enabled (e.g., under a $\lambda$-binder).
\end{example}
In order to be able to reflect all the reductions in translated processes, we state soundness in terms of an extended class of reductions for terms, denoted $\alfredto$ (with reflexive, transitive closure denoted $\alfredto*$).
To be precise, let $\mathfrak{C}$ be an arbitrary \alcalc context. In addition to the reductions in \Cref{fig:alpha_lambda_typing}, we consider reductions under arbitrary contexts:
\[
\infer
{M \alfredto M'}
{\mathfrak{C}[M] \alfredto \mathfrak{C}[M']}
\]
\begin{restatable}[Soundness]{theorem}{soundnessThm}
\label{thm:soundness}
    Given $\Delta \proves P \sublift M :: z : A$, if $P \reddStar Q$, then there exist $N$ and $R$ such that $M \alfredto* N$ and $Q \reddStar R \sublift N$.
\end{restatable}
\noindent
Note that the premise in the theorem above permits arbitrarily many reduction steps from $P$ to $Q$ (i.e., $P \reddStar Q$), assuring that \emph{every} sequence of reductions of $P$ is reflected by a corresponding sequence of reductions of the source term $M$.
The alternative with a single reduction in the premise  (i.e., $P \redd Q$) being a much weaker property.
The proof of soundness proceeds by cases on the possible reductions of $P$, informed by the structure and typing of the source term $M$.
The key point in the proof is to postpone certain independent reductions of the target process, which cannot be immediately matched by reductions in the source term.

 \section{Observational Equivalence and Denotational Semantics}
\label{sec:denot}
Here we develop the theory of observational equivalence for \piBI processes.
To this end, we first define   barbed equivalence and observational equivalence.
Then, we provide a denotational semantics and show that processes that have the same denotation are observationally equivalent.

\begin{figure}[t]
  \adjustfigure[\small]
  \begin{proofrules}
    \infer*[lab=barb-select]
    {\ell \in \{\inl, \inr\}}
    {\sel{x}{\ell}.P \barb{\sel{x}{\ell}}}
    \label{barb-select}

    \infer*[lab=barb-branch]
    {\ell \in \{\inl, \inr\}}
    {\caseLR{x}{Q_{\inl}}{Q_{\inr}} \barb{\recvIn{x}{\ell}}}
    \label{barb-branch}

    \infer*[lab=barb-send]
    {}
    {\out x[y].\big(P \| Q\big) \barb{\barbOut{x}}}
    \label{barb-send}

    \infer*[lab=barb-recv]
    {}
    {\inp x(y).Q \barb{\barbIn{x}}}
    \label{barb-recv}

    \infer*[lab=barb-wait]
    {}
    {\inp x().P \barb{\inp x()}}
    \label{barb-wait}

    \infer*[lab=barb-close]
    {}
    {\out x<> \barb{\out x<>}}
    \label{barb-close}

    \infer*[lab=barb-restr-l]
    {P \barb{\alpha} \and x \not= \chanOf(\alpha)}
    {\new x.(P \| Q) \barb{\alpha}}
    \label{barb-restr-l}

    \infer*[lab=barb-restr-r]
    {Q \barb{\alpha} \and x \not= \chanOf(\alpha)}
    {\new x.(P \| Q) \barb{\alpha}}
    \label{barb-restr-r}

    \infer*[lab=barb-spawn]
    {P \barb{\alpha} \and \chanOf(\alpha) \notin \restrOf(\spvar)}
    {\spawn{\spvar}. P \barb{\alpha}}
  \end{proofrules}
  \caption{Barbs for \piBI processes.}
  \label{fig:barbs}
\end{figure}
We first define \emph{barbs}---observations that we can make on processes.
Their formulation is  standard:
\begin{grammar}
  \alpha \is
  \selL{x}
  | \selR{x}
  | \recvIn{x}{\inl}
  | \recvIn{x}{\inr}
  | \barbIn{x}
  | \barbOut{x}
  | \inp x()
  | \out x<>
\end{grammar}
By $\chanOf(\alpha)$ we denote the channel associated to the barb $\alpha$.
We say that process $P$ has a barb $\alpha$, if the relation $P \barb{\alpha}$ is derivable from the rules in \Cref{fig:barbs}.
Now we define observational equivalence.
\begin{definition}[Barbed equivalence]
  Barbed equivalence is the largest equivalence relation $\barbeq$ on processes of the same type that is closed under reductions and that satisfies the following condition.
  If $P \barbeq Q$ and $P \downarrow_{\alpha}$ then there exists $Q'$ such that $Q \reddStar Q' \downarrow_{\alpha}$.
\end{definition}

We will be mainly concerned with barbed equivalence of \emph{closed} processes.
A process $P$ is closed if it is typeable as $\Sigma \proves P :: y : B$, where $\fn(\Sigma) = \emptyset$, i.e. $\Sigma$ is an empty bunch (cf.~\Cref{def:empty_bunch}).
Note that a closed process can only have barbs associated to its provided channel $y$.

A program context $\pctxt[\hole]$ is a \piBI process with a hole in it.
Given $\bunch \proves P :: x : A$, a \emph{closing program context} $\pctxt$ is a program context such that $\Sigma \proves \pctxt[P] :: y : B$ for some empty bunch $\Sigma$.
\begin{definition}[Observational equivalence]
\label{def:obs_equiv}
  Two processes $\bunch \proves P :: T$ and $\bunch \proves Q :: T $ are observationally equivalent, denoted $\bunch \proves P \obseq Q :: T$, if, for any closing program context $\pctxt$ and any type $A$ such that $\Sigma \proves \pctxt[P] :: z : A$ and $\Sigma \proves \pctxt[Q] :: z : A$, it is the case that $\pctxt[P]$ and $\pctxt[Q]$ are barbed equivalent.
\end{definition}

Observational equivalence is a strong notion, because it relates two processes in \emph{any} well-typed program context.
As a consequence, proving observational equivalence of two processes directly is challenging, as it requires reasoning about an arbitrary context $\pctxt$.
Next, we describe a sound, more compositional approach to proving equivalence,
based on a denotational semantics for \piBI.

\subsection{Denotational Semantics}
Our motivation for developing a denotational semantics for \piBI is two-fold. First, it will provide a sound technique for establishing observational equivalence.
Second, it will prove useful to illustrate the aspects of separation and sharing through tracking of the origins of different processes, thus explaining the fundamental differences between multiplicative and additive connectives in \piBI.
Intuitively, if we have a closed process of the type $\Sigma \proves P :: x : A \ast B$, then this process outputs a fresh channel $y$ on $x$, and then separates into two processes providing $A$ and $B$.
These two resulting processes will have a \emph{different origin}: they are not results of duplication via a spawn prefix.
Crucially, the two processes obtained from the prefix $\spawn{x \mapsto x_1, x_2}{}$ on channels $x_1$ and $x_2$ do have the same origin as the process on channel $x$.

In order to make this insight precise, we extend the type system with \emph{atomic types}, denoted $\aty_1, \aty_2, \dots$, which we use to represent abstract channels/resources.
The extension is conservative, as we do not introduce any rules for atomic types.
Also, we slightly modify our notion of empty bunches (cf.~\Cref{def:empty_bunch}) to allow names as long as they are associated with atomic types.
\begin{definition}[Atomic bunch]
  \label{def:empty_bunch}
  An \emph{atomic bunch}~$\emptybunch$ is a bunch such that
  any type assignment $x : A$ in $\emptybunch$ is such that $A$ is an atomic type.
\end{definition}
This way, e.g., we consider $\Sigma = (x : \aty_1 \band y : \aty_2) \bsep z : \aty_3$ an atomic bunch.
Since we do not add any typing rules for atomic processes, well-typed processes cannot ``break down'' sessions $\aty_i$ and cannot communicate/block on the channels associated with atomic types.
In fact, the only thing that a well-typed process can do with a channel associated to an atomic type is forwarding.
Hence, this extension retains the essential properties of the system (e.g., \Cref{{t:dlfreedom}}).

We now define the denotational semantics of types and processes.
We start with a fixed set $\Tag$ of primitive tags.
A tag represents an origin of a process, as in, e.g., the ID of a node in which a particular process is executed.
As such, these tags will represent different origins/provenances of processes.
We interpret every type as a $\pset{\Tag}$-valued set: $\Sem{A} : \pset{\Tag} \to \Set$.
We have:
\begin{align*}
\Sem{1_m}(D) & = \singleton
&
\Sem{\aty}(D) &= D
\\
\Sem{1_a}(D) &= \singleton
&
\Sem{A \vee B}(D) & = \Sem{A}(D) + \Sem{B}(D)\\
\Sem{A \wedge B}(D) & = \Sem{A}(D) \times \Sem{B}(D)
& \Sem{A \sep B}(D) & = \Sigma_{D = D_1 \dunion D_2}.\Sem{A}(D_1) \times \Sem{B}(D_2)  \\
\Sem{A \to B}(D) & = \Sem{A}(D) \to \Sem{B}(D)
& \Sem{A \wand B}(D) & = \Pi_{D' \inters D = \emptyset}. \Sem{A}(D') \to \Sem{B}(D \union D')
\end{align*}
where
  $\singleton$ is the (terminal) set with one element, and
      $D = D_1 \dunion D_2$ holds if
    $D_1 \cap D_2 = \emptyset$ and
    $D = D_1 \union D_2$.
The interpretation of bunches $\Sem{\Delta}$ is defined analogously
by treating~`$\band$' and `$\bsep$' as `$\land$' and `$\sep$', respectively.
A process is interpreted as a function polymorphic in a set of tags $D$:
\[
  \Sem{\Delta \proves P :: x : A}_D : \Sem{\Delta}(D) \to \Sem{A}(D).
\]
The interpretation follows the standard interpretation of BI in doubly closed categories~\cite[Chapter~3.3]{pym:2002}.
Specifically, we interpret types as presheaves $\Set^{\pset{\Tag}}\!$, where $\pset{\Tag}$ is interpreted as a discrete category.
The interpretation of type formers corresponds to the Cartesian closed structure and a closed monoidal structure on $\Set^{\pset{\Tag}}\!$.
As the construction is standard, we omit the details in the interest of space.

\paragraph{The meaning of multiplicative and additive types, formally}
As we claimed in \cref{ex:db-example}, the difference between the multiplicative
and additive types can be explained in terms of data flow.
Equipped with the denotational semantics,
we can now make this intuition formal.
The idea is that applying contraction to some resource~$A$,
leaves a trace in the form of
the `$\band$' in the resulting resource $A\band A$.
This records the fact that,
although we can interact with each copy of~$A$ independently,
the duplicated resources share a common provenance.

Atomic types represent the base types, whose provenance
we want to track.
Assigning the same tag to two atomic types indicates that they
have a common provenance.
For example, to inhabit $ \Sem{\aty_1 * \aty_2}(\set{t_1, t_2}) $
a term needs to decide how to split~$\set{t_1, t_2}$ into two disjoint sets,
therefore assigning $t_1$ to $\aty_1$ and $t_2$ to $\aty_2$ (or vice versa),
forcing the two atomic types to have different provenance
(as expected for separation).
The following examples show what this means for \piBI\ processes.

\begin{example}
\label{ex:provenance}
  Let $\Delta \proves P :: x : \aty_1 \ast \aty_2$.
  Then for any $D \in \pset{\Tag}$, and $v \in \Sem{\Delta}(D)$,
  $
    \Sem{P}_D(v) = (t_1, t_2)
  $
  with $t_1 \neq t_2$.
  In other words, the two sessions $\aty_1$ and $\aty_2$ that are sent over the channel~$x$ by~$P$ have a disjoint provenance.
  Note that this would not hold for a process typed
  $\Delta \proves Q :: x : \aty_1 \land \aty_2$.
  For example,
  if~$Q$ is
  $
    {\spawn{y->y_1,y_2}.\out x<y_1>.\fwd[x<-y_2]}
  $,
  we have $\Sem{y : \aty \proves Q :: x : \aty \land \aty}_D(t) = (t,t)$,
  for all~$t\in D$,
  \ie the sessions that are sent over the channel~$x$ by~$Q$ do share their provenance.

  Furthermore, since the denotational semantics is compositional, we can generalize the argument above.
  Suppose we place~$P$ in some enclosing context $\pctxt$ such that
  $y : \aty'_1 \bsep z : \aty'_2 \proves \pctxt[P] :: x : \aty_1 \sep \aty_2$.
  Then
  $\Sem{y : \aty'_1 \bsep z : \aty'_2 \proves \pctxt[P] :: x : \aty_1 \sep \aty_2}_{\set{t_1, t_2}}(t_1, t_2) = (t'_1, t'_2)$ and either ${t_1 = t'_1},{t_2 = t'_2}$ or ${t_1 = t'_2}, {t_2 = t'_1}$,
  as these are the only available functions in the denotational semantics.
  That is, the process $\pctxt[P]$ sends over $x$ the sessions either with the same provenance as $y$ and $z$, or with swapped provenance.
  If we instead consider a program typed with additive conjunction
  $y : \aty'_1 \band z : \aty'_2 \proves \pctxt[Q] :: x : \aty_1 \land \aty_2$,
  then it may further send over $x$ two sessions both with the same provenance
  (either the one of~$y$ or of~$z$).
\end{example}

\begin{example}
\label{ex:db-reprise}
  Recall the database type from \cref{ex:db-example}.
  Here we let the data stored in the database be of atomic type~$\aty$,
  and we consider the 1-unfolding of the original type:
  $
    \DB_{\aty} \is (\aty \to \aOne) \land (\aty \land \aOne).
  $
  In \cref{ex:db-example} we claimed that a process with access to two databases
  can generate a flow of data from one to the other only if they
  are additively composed.
  Let us see how the denotational semantics makes this evident at the type level.
Consider the additive case, symbolized by process~$P'_{\lbl{a}}$.
  For any set of tags~$D$, we have:
  \begin{align*}
    \Sem{
      \var{put}_1 : \aty \to \aOne \band d : \aty
        \proves \ghost[l]{P'_{\lbl{m}}}{P'_{\lbl{a}}} :: z : C
    }_{D}
    &:
    ((D \to \singleton) \times D) \to \sem{C}_D
\intertext{Therefore,
  the denotation of $\var{put}_1$, with domain~$D$, can be applied to the tag of $d$, which is a member of~$D$.
  Data flow from one database to the other is allowed
  since the databases are
  already declared to have shared provenance.
In contrast, for the multiplicative case represented by~$P'_{\lbl{m}}$,
  we have:
  }
\Sem{
      \var{put}_1 : \aty \to \aOne \bsep d : \aty
        \proves P'_{\lbl{m}} :: z : C
    }_{D}
    &:
    ((D_1 \to \singleton) \times D_2) \to \sem{C}_{D},
  \end{align*}
  for sets of tags $D_1$ and $D_2$ such that $D_1 \dunion D_2 = D$.
  In this case, because $\var{put}_1$ and $d$ are separated by `$\bsep$', the tags associated with $\aty$ in these types must be disjoint.
  Therefore, in the case of $P'_{\lbl{m}}$ it is impossible to apply~$\var{put}_1$
  (the denotation of which has domain~$D_1$)
  to the data from the second database
  (which comes from~$D_2$).
\end{example}

\subsection{Properties}
We write $\Sem{P}$ for $\Sem{\Delta \proves P :: x : A}$ when $\Delta$ and $x:A$ are unambiguous.
Our interpretation is indeed a valid model of \piBI, as it satisfies the following lemmas:
\begin{lemma}
  \label{lem:denot_congr}
  Let $\bunch \proves P :: z : A$ and $\bunch \proves Q :: z : A$ be processes such that $P \congr Q$. Then $\Sem{P} = \Sem{Q}$.
\end{lemma}
\begin{lemma}
  \label{lem:denot_redd}
  Let $\bunch \proves P :: z : A$ and $\bunch \proves Q :: z : A$ be processes such that $P \redd Q$. Then $\Sem{P} = \Sem{Q}$.
\end{lemma}

\paragraph{Denotational semantics and observational equivalence.}
As already mentioned, we will use denotational semantics to verify observational equivalences of processes.
Formally, we have:
\begin{theorem}
  \label{thm:denot_equiv_sound}
  Given two processes $\bunch \proves P :: z : C$ and $\bunch \proves Q :: z : C$, if $\Sem{P} = \Sem{Q}$, then $\bunch \proves P \obseq Q ::  z : C$.
\end{theorem}

\noindent
In order to prove this theorem we will need the following two lemmas:
\begin{restatable}{lemma}{observeWeak}
  \label{l:observeWeak}
  Suppose given a process $P$ such that $\Gamma \vdash P :: z:C$, where $P \nredd$ and $P$ does not have any barbs on channels from $\Gamma$.
  Then $P$ has a barb on the channel $z$.
\end{restatable}
\begin{lemma}[Observability]\label{t:observability}
    Suppose $\Sigma \proves P :: z : C$ such that $\Sigma$ is an atomic bunch.
    Then there exists a process $Q$ such that $P \reddStar Q$ where $Q \barb{\alpha(z)}$.
\end{lemma}
\begin{proof}
    By weak normalization and subject reduction (\Cref{thm:weak_norm,thm:subject_red}), and \Cref{l:observeWeak}.
\end{proof}
\noindent We can use the observability lemma to show \Cref{thm:denot_equiv_sound}:
\begin{proof}[Proof (of \Cref{thm:denot_equiv_sound})]
  Let $\pctxt$ be a closing program context.
  We are to show $\Sigma \proves \pctxt[P] \barbeq \pctxt[Q] :: z : D$ for some context $\pctxt$.

  Since the denotational semantics is compositional, we have $\Sem{\pctxt[P]} = \Sem{\pctxt[Q]}$.
  Furthermore, the relation $P, Q \mapsto \Sem{\pctxt[P]} = \Sem{\pctxt[Q]}$ is an equivalence relation and is closed under reductions (\Cref{lem:denot_redd}).
  Therefore, it suffices to consider only the main clause of barbed equivalence.
  That is, if $\Sem{P} = \Sem{Q}$ and $P \downarrow_{\alpha(z)}$, then there exists $Q'$ such that $Q \reddStar Q' \barb{\alpha(z)}$.

  For simplicity, let us consider a case where the type $D$, on which we are making observations, is of the form $A \vee B$.
  Then the only possible observations for $P$ and for $Q$ are $\sendInl{z}$ and $\sendInr{z}$.
  Suppose, without loss of generality, that $P \barb{\sendInl{z}}$.
  Then, by \Cref{t:observability}, we have $Q \reddStar Q' \barb{\alpha(z)}$.

  By \Cref{lem:denot_redd}, we have $\Sem{P} = \Sem{Q} = \Sem{Q'}$.
  Because $P$ has a barb $\sendInl{z}$, the interpretation $\Sem{P}$ must be of the form
  $\inl \circ (\dots)$, where $\inl$ is the embedding  $\Sem{A} \to \Sem{A} + \Sem{B}$.
  It the must be the case that $\Sem{Q'}$ is also of the same shape, and, hence, $Q'\barb{\sendInl{z}}$.
\end{proof}

As we have seen, the proof of \Cref{thm:denot_equiv_sound}
relies on \Cref{t:observability},
which in turn relies on weak normalization (\Cref{thm:weak_norm}).
In extensions of \piBI for which weak normalization does not hold, the proof strategy would need to appeal to observability based on other techniques, such as logical relations.
We elaborate on this in \Cref{{sec:rel_work}}.

\paragraph{Equivalence induced from the translation of the \alcalc.}
We close this section by demonstrating an application of denotational semantics in the context of the correctness of the translation from \Cref{sec:translation}.
Specifically, we show that the relation $\sublift$ from \Cref{sec:translation} decomposes as a translation $\bTrans_z(-)$ and an observational equivalence $\obseq$:
\begin{theorem}
 If $\bunch \proves P \sublift M :: z : A$, then $\Sem{P} = \Sem{\Trans_z(M)}$, and, consequently, $P \obseq \Trans_z(M)$.
 \label{thm:sublift_denot_eq}
\end{theorem}
\begin{proof}
  Essentially, we need to show that for any $\bctxt(x : A) \proves M : B$ and $\bunch \proves N : A$, we have
  \[
    \Sem{\Trans_z(M\subst{{x}->{N}})} = \Sem{\new x.(\Trans_x(N) \| \Trans_z(M))}.
  \]
  We do this by induction on the typing derivation, generalizing to multiple substitutions.
\end{proof}
Recall that in \Cref{sec:translation} we could not use the translation function $\bTrans_z(-)$ itself to establish a simulation; instead we had to take a coarser relation $\sublift$.
\Cref{{thm:sublift_denot_eq}} shows that this does not introduce any observable difference.

 \section{Related Work}
\label{sec:rel_work}
We have already discussed some of the most closely related works, and we have given some comparisons with previous works by means of examples in \Cref{sec:examples-comparisons}.
Here we discuss other related literature along several dimensions.

\paragraph{BI and process calculi.}
To our knowledge, the work of \citet{anderson.pym:2016} is the only prior work that connects BI with process calculi. Their technical approach and results are very different from ours. They introduce a process calculus (a synchronous CCS) with an explicit representation of (bunched) resources, in which processes and resources evolve hand-in-hand.
Rather than a typed framework for processes or an interpretation in the style of propositions-as-types, they use a logic related to BI to specify rich properties of processes, in the style of Hennessy-Milner logic.

\paragraph{BI and Curry-Howard correspondence}
The works of \citet{ohearn:2003} and \citet{pym:2002} are, to our knowledge, the only prior investigations into (non-concurrent) Curry-Howard correspondences based on BI.
These works were later extended to cover polymorphism~\cite{collinson.pym.robinson:2008} and store with strong update~\cite{berdine.ohearn:2006}.
An extension $\lambda_{\mathsf{sep}}$ of an affine version of the \alcalc with a more fine-grained notion of separation was studied by \citet{atkey:2004,atkey:2006}.

\paragraph{Previous works on propositions-as-sessions.}
Starting with the works by~\citet{caires.pfenning:2010} and \citet{wadler:2012}, the line of work on propositions-as-sessions has exclusively relied on (variants of) LL, which is incomparable to BI; this immediately separates those prior works from our novel approach based on BI.

Our work adapts to the BI setting  key design principles in~\cite{caires.pfenning:2010,wadler:2012}: the interpretation of multiplicative conjunction as output, linear implication as input, and the interpretation of `cut' as the coalescing of restriction and parallel composition.
Those works use input-guarded replication to accommodate non-linear sessions, typed with the modality $!A$; in contrast, \piBI handles structural principles directly at the process level with the new spawn prefix.

Our adaptation is novel and non-trivial, and cannot be derived from prior interpretations based on LL.
Still, certain aspects of \piBI bear high-level similarities with elements from those interpretations.
The semantics of our spawn prefix borrows inspiration from the treatment of aliases in Pruiksma and Pfenning's interpretation of asynchronous binary sessions, based on adjoint logic, in which structural rules are controlled via modalities~\cite{pruiksma.pfenning:2021}.
Thanks to spawn binders (\Cref{def:spawn-env}), our semantics explicitly handles duplication and disposal of services; this is similar in spirit to the syntax and semantics of replicated servers in HCP, an interpretation based on a hypersequent presentation of classical LL~\cite{kokke.montesi.peressotti:2019}.
The behavioral theory of HCP consists of a labeled transition semantics for processes, a denotational semantics for processes, and a full abstraction result.
The work of~\citet{qian.kavvos.birkedal:2021} extends linear logic with coexponentials with the aim of capturing client-server interactions not expressible in preceding interpretations of linear logic.
Precise comparisons between the expressivity of such interactions and the connection patterns enabled by our spawn prefix remains to be determined. Concerning failures, as discussed in \Cref{sec:examples-comparisons}, the work of \citet{fowler.etal:2019} develops a linear functional language with asynchronous communication and support for failure handling, closely related to Wadler's CP.

\paragraph{Observational equivalence}
Observational equivalence compares the behaviour of two processes
in \emph{every} (well-typed) program context.
This universal quantification makes  direct proofs of equivalence
very hard and non-compositional.
Observational equivalence is therefore usually established
using more compositional methods that do not involve reasoning
about a program in a context.
Examples of these methods are bisimulations and logical relations; in the session-typed setting, such methods have been addressed in~\cite{kouzapas.yoshida.honda:2011} and \cite{caires.perez.pfenning.toninho:2013,perez.etal:2014,DBLP:conf/esop/Atkey17,derakhshan.balzer.jia:2021}, respectively.
In \cref{sec:denot} we followed an approach based on denotational semantics, exploiting a canonical construction. 
Our denotational semantics also serves to
elucidate the difference between
the additive and multiplicative types.
Our proof of adequacy of the denotational semantics
for proving observational equivalence (\cref{thm:denot_equiv_sound})
relies on weak normalization (\cref{thm:weak_norm}).
In extensions of \piBI for which weak normalization does not hold,
our proof strategy would need to be revised.
Recent work on denotational semantics and logical relations for session-typed languages~\cite{derakhshan.balzer.jia:2021,Kavanagh22}
could provide the basis for handling such extensions.

 \section{Concluding Remarks and Future Perspectives}
\label{sec:conclusion}
In this paper we present a fresh look at logical foundations for message-passing concurrency.
We have cast the essential principles of propositions-as-sessions, initially developed upon LL, in the unexplored context of BI.
We introduced the typed process calculus \piBI, explored its operational and type-theoretical contents, illustrated its expressiveness, and established the meta-theoretical framework needed to study the behavioral consequences of the BI typing discipline for concurrency.

Our results unlock a number of enticing future directions.
First, because \piBI targets binary session types (between two parties) with synchronous communication, it would be interesting to study variants of \piBI with multiparty, asynchronous communication~\cite{honda.yoshida.carbone:2008,scalas.yoshida:2019}.
An asynchronous version of \piBI could be defined by following the work of~\citet{deyoung.etal:2012} to maximize concurrency.
Also, the works~\cite{carbone.etal:2016,caires.perez:2016} already provide insights on how to exploit \piBI to analyze  multiparty protocols.

Second, variations and extensions of BI could provide new insights.
For example, the $\bang{A}$ modality is not incompatible with BI,
and can be added to obtain a type $\bang{A} \simeq A * \dots * A $.
The intuitive interpretation is that the provider of~$\bang{A}$
can create an instance of~$A$ from scratch,
thus not sharing its origin with the other instances.
This new type would seem incomparable with the corresponding modality of LL,
which makes it interesting to study what  interpretations could admit.

\begin{acks}
  The authors would like to thank
  Revantha Ramanayake for providing comments on an early version of the paper
  and
  the anonymous referees for their valuable comments and helpful suggestions.
  The authors also would like to thank Vasilios Andrikopoulos for helpful discussions.

This work was supported by the
  \grantsponsor{NWO}{Dutch Research Council (NWO)}{}
  under project No.~\grantnum{NWO}{016.Vidi.189.046}
  (Unifying Correctness for Communicating Software);
  and by a
  \grantsponsor{ERC}{European Research Council}{https://erc.europa.eu/} (ERC)
  Consolidator Grant for the project ``PERSIST'' under the
  European Union's Horizon 2020 research and innovation programme
  (grant agreement No.~\grantnum{ERC}{101003349}).
\end{acks}

\setlabel{LAST}\label{paper-last-page}

\appendix

\section{Meta-theoretical properties}
\label{app:meta_props}

In this section we first introduce auxiliary lemmas about (typed) spawn bindings and free/bound names of processes.
We then present omitted proofs of subject congruence and subject reduction in \Cref{{app:subject_red:proof}}.
After that we prove deadlock-freedom (\Cref{t:dlfreedom}) in \Cref{app:dlfree:proof}, introducing a \emph{progress} lemma.
Finally, in \Cref{{appendix:sec:wn}} we detail our proof of weak normalization.

\subsection{Spawn bindings}
We will need the following properties for proving subject reduction, all shown by induction on spawn binding typing.
\begin{lemma}
  \label{lem:typedS_ctx}
  If $\typedS{\spvar}{\bunch_2}{\bunch_1}$, then $\typedS{\spvar}{\bctxt(\bunch_2)}{\bctxt(\bunch_1)}$ for any bunched context $\bctxt$.
\end{lemma}
\begin{lemma}
  \label{lem:typedS_cut_nil}
  Suppose that $
    \spvar \from \bctxt_1(x:A) \tobunch \bunch_2
  $ and $x \notin \dom(\spvar)$.
  Then $\bunch_2 = \bctxt_2(x:A)$ for some $\bctxt_2$.
  Furthermore, for any~$\bunch$ we have
  $
    \spvar \from \bctxt_1(\bunch) \tobunch \bctxt_2(\bunch).
  $
\end{lemma}
\begin{proof}
  We proceed by induction on $\typedS{\spvar}{\bunch_2}{\bctxt_1(x:A)}$.
  \begin{induction}
    \step[Case \ref{rule:spawn-weaken}]
    If $
      \map{x -> \emptyset | {x \in \ident(\bunch)}}
        \from \bctxt(\bunch; \bunch') \tobunch \bctxt(\bunch')
    $, then, from the assumption that $x \notin \dom(\spvar)$,
    we know that~$x$ does not occur~$\bunch$.
    That means that~$x$ it is either part of~$\bunch$ or~$\bctxt$.

    In either case, we can freely replace~$x:A$
    with an arbitrary bunch.

    \step[Case \ref{rule:spawn-contract}]
    Similar to the previous case.

    \step[Case \ref{rule:spawn-merge}]
    Suppose we have
    $
      (\spvar_1 \merge \spvar_2) \from
        \bctxt_0(x:A) \tobunch \bunch_2
    $
    with $
      \spvar_1 \from \bctxt_0(x:A) \tobunch \bunch_1
    $
    for some intermediate bunch~$\bunch_1$.

    Since $x \notin \dom(\spvar_1 \merge \spvar_2)$,
    we know that $x \notin \dom(\spvar_1)$.
    Hence,
    ${\bunch_1 = \bctxt_1(x:A)}$
    by induction hypothesis,
    and we have $\spvar_1\from \bctxt_1(\bunch) \tobunch \bctxt_2(\bunch)$
    for any bunch~$\bunch$.

    Furthermore, $x \notin \restrOf(\spvar_1)$
    (otherwise we would not be able to replace~$x : A$
    with an arbitrary bunch~$\bunch$).
    Hence, $x \notin \dom(\spvar_2)$,
    and by induction hypothesis we have
    $\bunch_2 = \bctxt_2(x:A)$ for some $\bctxt_2$, and
    $
      \spvar_2 \from \bctxt_1(\bunch) \tobunch \bctxt_2(\bunch)
    $ for any bunch~$\bunch$.
    The desired result then follows by using \ref{rule:spawn-merge} again.
  \qedhere
  \end{induction}
\end{proof}
\begin{lemma}
  \label{lem:typedS_cut}
  Suppose that
    $\spvar \from \bunch_1 \tobunch \bunch_2$ and
    $\spvar(x) = \set{x_1, \dots, x_n}$.
  Then
    $\bunch_1 = \bctxt_1(x:A)$ and
    $\bunch_2 = \bctxt_2(x_1 : A \mid \dots \mid x_n : A)$,
  for some~$\bctxt_1$ and~$\bctxt_2$.
Furthermore, for any bunch~$\bunch$ we have:
  \[
    (\spvar \setminus \set{x})
    \union
    \map{ y -> \set{y_1, \dots, y_n} | {y \in \ident(\bunch)} }
    \from
      \bctxt_1(\bunch).
    \tobunch
      \bctxt_2\bigl(\idx{\bunch}{1} \mid \dots \mid \idx{\bunch}{n}\bigr)
  \]
\end{lemma}
\begin{proof}
  Similar to the previous lemma.
\end{proof}

\subsection{Names and substitutions}
\begin{lemma}\label{l:fn_ident}
If $\Delta \proves P :: x :C$, then $\fn(P) = \fn(\Delta) \cup \{ x \}$.
\end{lemma}
\begin{lemma}
  \label{lem:fn_typed_cut2}
  If $\new x.(P \| \new y. (Q \| R))$ is a well-typed process, then $x$ is shared either between $P$ and $Q$, or between $P$ and $R$, but not between all the three subprocesses, and $y$ is shared between $Q$ and $R$.
  That is, either $x \in \fn(P) \cap \fn(Q)$ and $x \notin \fn(R)$, or $x \in \fn(P) \cap \fn(R)$ and $x \notin \fn(Q)$.
  And in both cases we have $y \notin \fn(P)$
\end{lemma}
This lemma implies that whenever we have a typed process $\new x.(P \| \new y. (Q \| R))$, then one of the congruences \ref{cong-assoc-l} or \ref{cong-assoc-r} apply.
\begin{lemma}
\label{lem:fn_typed_cut_spawn}
If $\new x.(P \| \spawn{\spvar}.Q)$ is a well-typed process, and $x \notin \spvar$, then $\fn(P) \cap \rng(\spvar) = \emptyset$.
\end{lemma}
Similarly to the previous lemma, this lemma implies that for a well-typed process of the form $\new x.(P \| \spawn{\spvar}.Q)$ either \ref{red-spawn} or \ref{red-spawn-r} apply.
\begin{lemma}
\label{lem:spawn_typed_indep}
If $\new x. (\spawn{\spvar_1}.P \| \spawn{\spvar_2}.Q)$ is a well-typed process, then $\spvar_1$ and $\spvar_2$ are independent.
\end{lemma}
\begin{lemma}
\label{lem:inj_subst}
  If $\Delta \proves P :: x : C$ and $\theta$ is an injective substitution, then
  $\Delta\theta \proves P\theta :: \theta(x) : C$.
\end{lemma}

\subsection{Proof of subject reduction}
\label{app:subject_red:proof}

\subjectRed*

\begin{proof}[Proof (structural congruence)]
We proceed by induction on $P \congr Q$, examining the possible typing derivations of $P$.

\textbf{Case \ref{cong-assoc-l}.}
This congruence states that the order of independent cuts does not matter.
Corresponds to the following proof conversion:
\begin{mathpar}
\infer*
  {{\bunch_1 \proves P :: x : A}
    \and
    \infer*
    {\bunch_2 \proves Q :: y : B
      \and \Gamma(x : A \mid y : B) \proves R :: z : C}
    {\Gamma(x : A \mid \bunch_2) \proves \new y. (Q \|_y R) :: z : C}}
  {\Gamma(\bunch_1 \mid \bunch_2) \proves \new x.\big(P \|_x \new y.(Q \|_y R)\big) :: z : C}
\end{mathpar}
{\Huge
$$\leftrightsquigarrow$$
\normalsize}
\begin{mathpar}
\infer*
  {\bunch_2 \proves Q :: y : B
    \and
\infer*
    {\bunch_1 \proves P :: x : A \and \Gamma(x: A \mid y : B) \proves R :: z : C}
    {\Gamma(\bunch_1 \mid y : B) \proves \new x. (P \|_x R) :: z : C}
   }
  {\Gamma(\bunch_1 \mid \bunch_1) \proves \new y.\big(Q \|_y \new x. (P \|_x R)\big) :: z : C}
\end{mathpar}
\smallskip

\textbf{Case \ref{cong-assoc-r}.}
This congruence states that the order of subsequent cuts does not matter.
Corresponds to the following proof conversion:
\begin{mathpar}
\infer*
  {{\bunch_1 \proves P :: x : A}
    \and
    \infer*
    {\bunch_2(x: A) \proves Q :: y : B
      \and \Gamma(y : B) \proves R :: z : C}
    {\Gamma(\bunch_2(x:A)) \proves \new y. (Q \|_y R) :: z : C}}
  {\Gamma(\bunch_2(\bunch_1)) \proves \new x.\big(P \|_x \new y.(Q \|_y R)\big) :: z : C}
\end{mathpar}
{\Huge
$$\leftrightsquigarrow$$
\normalsize}
\begin{mathpar}
\infer*
  {\infer*
    {\bunch_1 \proves P :: x : A \and \bunch_2(x: A) \proves Q :: y : B}
    {\bunch_2(\bunch_1) \proves \new x. (P \|_x Q) :: y : B}
   \and \Gamma(y : B) \proves R :: z : C}
  {\Gamma(\bunch_2(\bunch_1) \proves \new y.\big(\new x. (P \|_x Q) \|_y R\big) :: z : C}
\end{mathpar}
\smallskip

\textbf{Case \ref{cong-spawn-swap}.}
Similarly to the previous case, if two spawn bindings $\spvar$ and $\spvar'$ are independent, than they correspond to two independent applications of \ruleref{struct} that can be commuted past each other.
For example,
\begin{mathpar}
\infer*
{ \infer*
  {\Gamma(\empA \mid \idx{\bunch}{1} \band \idx{\bunch}{2}) \proves P :: z : C}
  {\Gamma(\empA \mid \bunch) \proves \spawnQ{x \mapsto \{x_1, x_2\}}{x \in \fn(\bunch)} P :: z : C}}
{\Gamma(\bunch' \mid \bunch) \proves \spawnQ{x \mapsto \emptyset}{x \in \fn(\bunch')}\spawnQ{x \mapsto \{x_1, x_2\}}{x \in \fn(\bunch)} P :: z : C}
\end{mathpar}
{\Huge
$$\leftrightsquigarrow$$
\normalsize}
\begin{mathpar}
\infer*
{ \infer*
  {\Gamma(\empA \mid \idx{\bunch}{1}\band \idx{\bunch}{2}) \proves P :: z : C}
  {\Gamma(\bunch' \mid \idx{\bunch}{1}\band \idx{\bunch}{2}) \proves \spawnQ{x \mapsto \emptyset}{x \in \fn(\bunch')}P :: z : C}}
{\Gamma(\bunch' \mid \bunch) \proves \spawnQ{x \mapsto \{x_1, x_2\}}{x \in \fn(\bunch)}\spawnQ{x \mapsto \emptyset}{x \in \fn(\bunch')} P :: z : C}
\end{mathpar}
\smallskip

\textbf{Closure under program contexts.}
 Using the induction hypothesis.
\end{proof}

\begin{proof}[Proof (reduction)]
  By induction on the reduction relation $P \redd Q$ and the typing derivation.
  The case \ref{red-cong} follows using the previous theorem and the induction hypothesis.
  The case \ref{red-eval-ctxt} for evaluation contexts of the form $\ectxt = \new x. (\ectxt'[\hole] \| P)$ and $\ectxt = \spawn{\spvar}. \ectxt'[\hole]$ follows from the induction hypothesis.

  \textbf{Case \ref{red-unit-l}.}
It corresponds to the following reduction of proofs, or its additive version:
\begin{mathpar}
  \begin{array}[t]{c c c}
\infer*
{\empM \proves \out x<> :: x : \Emp
\and
\infer*{\Gamma(\empM) \proves Q :: z : C}
       {\Gamma (x:\Emp) \proves \inp x().Q :: z  : C}}
{\Gamma(\empM) \proves \new x.(\out x<> \|_x \inp x().Q) :: z : C}
    & \BigArrow
& \axiom{\Gamma(\empM) \proves Q :: z : C}
  \end{array}
\end{mathpar}
\smallskip

\textbf{Case \ref{red-comm-l}.}
It corresponds to the following reduction of proofs, or its additive version:
\begin{mathpar}
\infer*
{ \infer*
  {\bunch_1 \proves P_1 :: y : A
  \and
  \bunch_2 \proves P_2 :: x : B}
  {\bunch_1\bsep \bunch_2 \proves \out x[y].(P_1 \| P_2) :: x : A \ast B}
  \and
  \infer*
  {\Gamma(y : A\bsep x : B) \proves Q :: z : C}
  {\Gamma (x : A \ast B) \proves \inp x(y).Q :: z : C}
}
{\Gamma(\bunch_1\bsep \bunch_2) \proves
\new x.\big(\out x[y].(P_1 \| P_2) \| \inp x(y).Q\big) :: z : C
}
\end{mathpar}
\begin{center}
{\BigArrow}
\end{center}
\begin{mathpar}
\infer*
{
  \bunch_2 \proves P_2 :: x : B
\and
 \infer*
  {\bunch_1 \proves P_1 :: y : A
  \and \Gamma(y : A\bsep x : B) \proves Q :: z : C}
  {\Gamma(\bunch_1\bsep x : B) \proves \new y. ( P_1 \| Q) :: z : C}
}
{\Gamma(\bunch_1\bsep \bunch_2) \proves \new x. \big( P_2 \| \new y. ( P_1 \| Q)\big) :: z : C}
\end{mathpar}
\smallskip

\textbf{Case \ref{red-comm-r}.}
It corresponds to the following reduction of proofs, or its additive version:
\begin{mathpar}
\infer*
{ \infer*
  {\bunch\bsep y : A \proves Q :: x : B}
  {\bunch \proves \inp x(y).Q :: x : A \wand B}
  \and
  \infer*
  {\bunch_1 \proves P_1 :: y : A
  \and
  \Gamma(x:B\bsep \bunch_2) \proves P_2 :: z : C}
  {\Gamma(\bunch_1\bsep x : A \wand B\bsep \bunch_2) \proves \out x[y].(P_1 \| P_2) :: z : C}
}
{\Gamma(\bunch_1\bsep \bunch\bsep \bunch_2) \proves
\new x.\big(\inp x(y).Q \| \out x[y].(P_1 \| P_2) \big) :: z : C
}
\end{mathpar}
\begin{center}
{\BigArrow}
\end{center}
\begin{mathpar}
\infer*
{ \infer*
  { \bunch_1 \proves P_1 :: y : A
    \and {\bunch\bsep y : A \proves Q :: x : B}
  }
  {\bunch\bsep \bunch_1 \proves \new y.(P_1 \| Q) :: x : B}
  \and
  \Gamma(x : B\bsep \bunch_2) \proves P_2 :: z : C
}
{\Gamma(\bunch\bsep \bunch_1\bsep \bunch_2) \proves
\new x.\big(\new y.(P_1 \| Q) \| P_2\big) :: z : C}
\end{mathpar}
\smallskip

\textbf{Case \ref{red-case}.}
Corresponds to the following reduction of proofs (for $\ell = \inl$):
\begin{mathpar}
\infer*
{\infer*{\bunch \proves P :: x : A}
        {\bunch \proves \selL{x}.P :: x : A \vee B}
\and
\infer*{\bctxt(x : A) \proves Q_1 :: z : C
        \and \bctxt(x : B) \proves Q_2 :: z : C}
       {\bctxt (x : A \vee B) \proves \caseLR{x}{Q_1}{Q_2} :: z : C}}
{\bctxt(\bunch) \proves \new x.\big(\selL{x}.P \| \caseLR{x}{Q_1}{Q_2} \big) :: z : C}
\end{mathpar}
\begin{center}
{\BigArrow}
\end{center}
\begin{mathpar}
\infer*
{{\bunch \proves P :: x : A} \and {\bctxt(x : A) \proves Q_1 :: z : C}}
{\bctxt(\bunch) \proves \new x.\big(P \| Q_1 \big) :: z : C}
\end{mathpar}
\smallskip

\textbf{Case \ref{red-fwd-l}.}
It corresponds to the following reduction of proofs, using \Cref{lem:inj_subst}:
\begin{mathpar}
\begin{array}[t]{c c c}
\infer*
{y : A \proves \fwd[x<-y] :: x : A
\and
\bunch(x:A) \proves P :: z : C}
{\bunch(y:A) \proves \new x.\big(\fwd[x<-y] \| P\big) :: z : C}
  & \BigArrow
  & \axiom{\bunch(y:A) \proves \substS{P}{x}{y} :: z : C}
\end{array}
\end{mathpar}
\smallskip

\textbf{Case \ref{red-fwd-r}.}
It corresponds to the following reduction of proofs:
\begin{mathpar}
\begin{array}[t]{c c c}
\infer*
  {\bunch \proves P :: x : A
  \and
  x : A \proves \fwd[y<-x] :: y : A
  }
  {\bunch \proves \new x.\big(P \| \fwd[y<-x]\big) :: y : A}
  & \BigArrow
  & \axiom{\bunch \proves \substS{P}{x}{y} :: y : A}
\end{array}
\end{mathpar}
\smallskip

\textbf{Case \ref{red-spawn}.}
This corresponds to pushing application of \ruleref{struct} past the cut.
By \Cref{lem:typedS_cut}, the derivation that we have is of the following shape:
\begin{equation*}
\infer*
{\bunch \proves P :: x : A
  \and
  \infer*
  {\Gamma(x_1 : A \mid \dots \mid x_n : A) \proves Q :: z : C
    \and \typedS{\spvar}{\Gamma(x_1 : A \mid \dots \mid x_n : A)}{\Gamma'(x:A)}}
  {\Gamma'(x : A) \proves \spawn{\spvar}.Q :: z : C}}
{\Gamma'(\bunch) \proves \new x.\big(P \| \spawn{\spvar}.Q\big) :: z : C}
\end{equation*}
\begin{center}
{\BigArrow}
\end{center}
\begin{equation*}
\small
\infer*
{
 \infer*
 {\infer*
    {\infer*
      {\idx{\bunch}{n} \proves \idx{P}{n} :: x_{n} : A \and
         \Gamma(x_1 : A \mid \dots \mid x_n : A) \proves Q :: z : C}
       {\Gamma(x_1 : A \mid \dots \mid \bunch^{(n)}) \proves \new {x_{n}}.(\idx{P}{n} \| Q) :: z : C}}
    {\bunch^{(1)} \proves P^{(1)} :: x_1 : A \and \dots}}
 {\Gamma(\bunch^{(1)} \mid \dots \mid \bunch^{(n)}) \proves \new {x_1}.(P^{(1)} \| \dots \new {x_{n}}.( P^{(n)} \| Q) \dots ) :: z : C}
  \and
 \typedS{\spvar'}{\Gamma(\bunch^{(1)} \mid \dots \mid \bunch^{(n)})}{\Gamma'(\bunch)}}
{\Gamma'(\bunch) \proves \spawn{\spvar'}.\new x_1.(P^{(1)} \| \dots \new x_{n}.( P^{(n)} \| Q) \dots ) :: z : C}
\end{equation*}
The spawn binding $\spvar' = (\spvar \setminus x) \cup \big[y \mapsto \{y_1, \dots, y_n\}\big]_{y \in \fn(\Pi)}$ is well-typed according to \Cref{lem:typedS_cut}.
\smallskip

\textbf{Case \ref{red-spawn-l}.}
As in the previous case, this reduction corresponds to moving \ruleref{struct} past a cut.
Since $\spawn{\spvar}.P$ appears on the left side of the composition, we know that the application of \ruleref{struct} was independent from $Q$ and from the cut.
The corresponding proof transformation is as follows:
\begin{mathpar}
\infer*
{\infer*
 {\Pi \proves P ::x : A \and \typedS{\spvar}{\Pi}{\Pi'}}
 {\Pi' \proves \spawn{\spvar}.P :: x : A}
\and
\Gamma(x : A) \proves Q :: z : C
}
{\Gamma(\Pi') \proves \new x.\big(\spawn{\spvar}.P \| Q\big) :: z : C}
\end{mathpar}
\begin{center}
{\BigArrow}
\end{center}
\begin{mathpar}
\infer*
{\infer*{\Pi \proves P ::x : A
        \and
        \Gamma(x : A) \proves Q :: z : C}
      {\Gamma(\Pi) \proves \new x.(P \| Q) :: z :C}
 \and \typedS{\spvar}{\Gamma(\Pi)}{\Gamma(\Pi')}}
{\Gamma(\Pi') \proves \spawn{\spvar}. \new x.\big(P \| Q\big) :: z : C}
\end{mathpar}
Where we know that $\typedS{\spvar}{\Gamma(\Pi)}{\Gamma(\Pi')}$ by \Cref{lem:typedS_ctx}.
\smallskip

\textbf{Case \ref{red-spawn-r}.}
Similar to the previous case.
Since $x \notin \spvar$ we know that the application of \ruleref{struct} is independent of the cut.
That is, by \Cref{lem:typedS_cut_nil}, the contexts for $Q$ are of the shape $\Gamma(x:A)$ and $\Gamma'(x:A)$.
\begin{mathpar}
\infer*
{\bunch \proves P :: x : A
\and
\infer*
 {\Gamma(x:A) \proves Q :: z : C \and \typedS{\spvar}{\Gamma(x:A)}{\Gamma'(x:A)}}
 {\Gamma'(x:A) \proves \spawn{\spvar}.Q :: z : C}
}
{\Gamma(\bunch) \proves \new x.\big(P \| \spawn{\spvar}.Q\big) :: z : C}
\end{mathpar}
\begin{center}
{\BigArrow}
\end{center}
\begin{mathpar}
\infer*
{
  \infer*
  {\bunch \proves P :: x : A \and \Gamma(x : A) \proves Q :: z : C}
  {\Gamma(\bunch) \proves \new x.(P \| Q) :: z : C}
  \and
  \typedS{\spvar}{\Gamma(\bunch)}{\Gamma'(\bunch)}
}
{\Gamma'(\bunch) \proves \spawn{\spvar}.\big(\new x.(P\| Q)\big) :: z : C}
\end{mathpar}
Where $\spvar$ has the appropriate type by \Cref{lem:typedS_cut_nil}.
\smallskip

\textbf{Case \ref{red-spawn-merge}.}
Directly using the rules for spawn prefix typing.

\end{proof}

\subsection{Proof of deadlock-freedom}
\label{app:dlfree:proof}
To prove deadlock-freedom, we first need to analyze when a process is \emph{not} stuck, i.e.\ when it can reduce.
We define the \emph{readiness} of a process, which is a means to syntactically determine whether a process can reduce.
This notion of readiness\footnote{In some literature this notion is referred to as ``liveness'', but we did not want to confused it here with a more semantic notion of liveness.} is useful when implementing \piBI as, e.g., a programming language: a reduction can be derived by simply analyzing the syntax of a program.

To define readiness, we need to know which names can be used for a communication.
We define this as a process' set of \emph{active names}: free names used for communication prefixes not guarded by other communication prefixes.
\begin{definition}[Active Names]\label{d:activeNames}
    Given a process $P$, we define the set of \emph{active names} of $P$, denoted $\an(P)$, as follows:
    \begin{align*}
        \an(\out x<>)
        &= \{x\}
        &
        \an(\new x.(P \|_x Q))
        &= (\an(P) \cup \an(Q)) \setminus \{x\}
        \\
        \an(\inp x().P)
        &= \{x\}
        &
        \an(\spawn{\spvar}.P)
        &= \dom(\spvar) \cup (\an(P) \setminus \rng(\spvar))
        \\
        \an(\out x[y].(P \| Q))
        &= \{x\}
        &
        \an(\inp x(y).P)
        &= \{x\}
        \\
        \an(\selL{x}.P)
        &= \an(\selR{x}.P) = \{x\}
        &
        \an(\caseLR{x}{P}{Q})
        &= \{x\}
        \\
        \an(\fwd[x<-y])
        &= \{x,y\}
\end{align*}
\end{definition}

\begin{lemma}
  If $ P \congr Q $ then $ \an(P) = \an(Q) $.
\end{lemma}
\begin{proof}
    There are no rules of structural congruence that add or remove prefixes.
    Moreover, the only rule of structural congruence that affects names only affects \emph{bound} names, and active names are free by definition.
\end{proof}

\begin{definition}[Ready process]
\label{def:ready}
  A process $P$ is \emph{ready}, denoted $\ready(P)$, if it is expected to reduce.
  Formally, the $\ready$ predicate is defined by the following rules:
  \begin{proofrules}
    \infer{
      \ready(Q)
    }{
      \ready(\spw\spb.Q)
    }

    \infer{ }{
      \ready(\spw\spb.\spw\spb'.Q)
    }

    \infer{
      \ready(Q)
      \\
      P \congr Q
    }{
      \ready(P)
    }

    \infer{
      x \in \an(P) \inters \an(Q)
    }{
      \ready(\new x.(P \|_x Q))
    }

    \infer{
      \ready(P)
      \lor
      \ready(Q)
    }{
      \ready(\new x.(P \|_x Q))
    }

    \infer{ }{
      \ready(\new x.(\spawn{\spvar}.P \|_x Q))
    }

    \infer{ }{
      \ready(\new x.(P \|_x \spawn{\spvar}.Q))
    }

    \infer{ }{
      \ready(\new x.(\fwd[x<- y] \|_x Q))
    }

    \infer{ }{
      \ready(\new x.(P \|_x \fwd[x<- y]))
    }
  \end{proofrules}
\end{definition}

The following then assures that well-typed, ready processes can reduce:
\begin{lemma}[Progress]\label{t:progress}
    Suppose given a process $P$ such that $\Gamma \proves P :: z:C$ and $\ready(P)$.
    Then, there exists a process $S$ such that $P \redd S$.
\end{lemma}
\begin{proof}
    By induction on the derivation of $\ready(P)$.
    \begin{casesplit}
        \case[$P = \spw{\spvar}.Q$ and $\ready(Q)$]
        By the IH, there exists $S'$ such that $Q \redd S'$.
        By Rule~\ref{red-eval-ctxt}, $P \redd \spw{\spvar}.S'$.

        \case[$P = \spw{\spvar}.\spw{\spvar'}.Q$]
        By Rule~\ref{red-spawn-merge}, $P \redd \spw{\spvar \merge \spvar'}.Q$.

        \case[$P \congr Q$ and $\ready(Q)$]
        By the IH, there exists $S$ such that $Q \redd S$.
        By Rule~\ref{red-cong}, $P \redd S$.

        \case[$P = \new x.(Q \| R)$ and $x \in \an(Q) \cap \an(R)$]
        We have $\Gamma = \Gamma'(\Delta)$ where $\Delta \proves Q :: x:A$ and $\Gamma'(x:A) \proves R :: z:C$.
        Since $x \in \an(Q) \cap \an(R)$, there is an unguarded prefix with subject $x$ in both $Q$ and $R$.
        Being unguarded, the prefix in $Q$ appears inside a sequence of $n$ cuts and spawns.
        Similarly, the prefix in $R$ appears inside a sequence of $m$ cuts and spawns.
        By induction on $n$ and $m$, we show that there exists a process $S$ such that $P \redd S$.
        \begin{itemize}
            \item
            If $n = 0$ and $m = 0$, the analysis depends on whether $Q = \fwd[x<-y]$ or $R = \fwd[y<-x]$, or neither.
                If so, this case is analogous the appropriate of the latter two cases of this proof.

                Otherwise, neither $Q$ nor $R$ is a forwarder.
                In that case, $Q$ is typable with a right rule for output, input, selection, or branching on $x$, depending on the type $A$.
                Similarly, $R$ is typable with a dual left rule on $x$.
                Suppose, as a representative example, that $A = B_1 \ast B_2$.
                Then, $Q$ is typable with Rule~\ref{rule:sep-r}, i.e.\ $Q = \out x[y].(Q_1 \| Q_2)$.
                Similarly, $R$ is typable with Rule~\ref{rule:sep-l}, i.e.\ $R = \in x(y').R'$.
                Then $P = \new x.(\out x[y].(Q_1 \| Q_2) \| \in x(y').R)$.
                Let $S = \new x.(Q_2 \| \new y.(Q_1 \| R'\subst{y'->y}))$.
                By Rule~\ref{red-comm-l}, $P \redd S$.

            \item
                If $n = n' + 1$, then the analysis depends on whether the outermost construct in $Q$ is a cut or a spawn.
                We thus consider these two cases:
                \begin{itemize}
                    \item
                        If the outermost construct is a cut, then $P = \new x.(\new w.(Q_1 \| Q_2) \| R)$.
                        The prefix on $x$ appears in $Q_2$, under a sequence of $n'$ sequence of and spawns.
                        Since $x \in \fn(\new w.(Q_1 \| Q_2))$, we know $w \neq x$.
                        This means that $x \notin \fn(Q_1)$.
                        Hence, by Rule~\ref{cong-assoc-r}, $P \congr \new w.(Q_1 \| \new x.(Q_2 \| R))$.
                        By the IH, there exists $S'$ such that $\new x.(Q_2 \| R) \redd S'$.
                        Let $S = \new w.(Q_1 \| S')$.
                        Then, by Rules~\ref{red-eval-ctxt} and~\ref{red-cong}, $P \redd S$.

                    \item
                        If the outermost construct is a spawn, then $Q = \spw{\spvar}.Q'$ and the proof follows as in the case where $P = \new x.(\spw{\spvar}.Q \| R)$.
                \end{itemize}

            \item
                If $m = m' + 1$, the analysis is analogous to the case above.
        \end{itemize}

        \case[$P = \new x.(Q \| R)$ and $\ready(Q)$ or $\ready(R)$]
        W.l.o.g., assume $\ready(Q)$.
        By the IH, there exists $S'$ such that $Q \redd S'$.
        By Rule~\ref{red-eval-ctxt}, $P \redd \new x.(S' \| R)$.

        \case[$P = \new x.(\spw{\spvar}.Q \| R)$]
        By typability, $x \notin \spvar$.
        Hence, by Rule~\ref{red-spawn-l},
        \[P \redd \spw{\spvar}.\new x.(Q \| R).\]

        \case[$P = \new x.(Q \| \spw{\spvar}.R)$]
        The analysis depends on wheter $x \in \spvar$ or not, so we consider two cases:
        \begin{itemize}
            \item
                If $x \in \spvar$, then $\spvar(x) = \set{x_1, \ldots, x_n}$.
                Let $\spvar' = (\spvar \setminus \set{x}) \cup \map{w=>\set{w_1, \ldots, w_n} | w \in (\fn(Q) \setminus \set{x})}$ and $S = \spw{\spvar'}.\new x_1.(Q^{(1)} \| \ldots \new x_n.(Q^{(n)} \| R) \ldots )$.
                By Rule~\ref{red-spawn}, $P \redd S$.

            \item
                If $x \notin \spvar$, by Rule~\ref{red-spawn-r}, $P \redd \spw{\spvar}.\new x.(Q \| R)$.
        \end{itemize}

        \case[$P = \new x.(\fwd[x<-y] \| Q)$]
        By Rule~\ref{red-fwd-l}, $P \redd Q\subst{x->y}$.

        \case[$P = \new x.(Q \| \fwd[y<-x])$]
        By Rule~\ref{red-fwd-r}, $P \redd Q\subst{x->y}$.
        \qedhere
    \end{casesplit}
\end{proof}

\dlfreedomThm*
\begin{proof}
  If the process $P$ is ready, then the result follows from \Cref{t:progress}.
  Otherwise, towards a contradiction, assume $P \not\congr \out z<>$ and $P \not\congr \spawn{\emptyset}.\out z<>$.
    W.l.o.g., assume $P$ is not prefixed by an empty spawn.
    Since $\Sigma$ contains no names, and $P$ is not an empty output on $z$, the only possibility is that $P$ is a cut: $P \congr \new x. (Q \|_x R)$.
    There are several possibilities for $Q$ and $R$: they can be communication prefixes on $x$, they can be spawn prefixes with only $x$ in the domain, they can be cuts, or they can be forwarders on $x$.
    \begin{itemize}
        \item $Q$ or $R$ is a spawn prefix with only $x$ in the domain.
            By definition, $P$ is ready.
        \item $Q$ or $R$ is a forwarder on $x$.
            By definition, $P$ is ready.
        \item $Q$ is a communication prefix on $x$.
            By typability, $x$ must be free in $R$, so $R$ must contain an action on $x$: a spawn prefix with $x$ in the domain, a forwarder on $x$, or a communication prefix on $x$.
            If $x \in \an(R)$, then $x \in \an(Q) \cap \an(R)$, so $P$ is ready by definition.
            Otherwise, there is a name $y$ such that the action on $x$ in $R$ is guarded by a spawn prefix with $y$ in the domain or by a communication prefix on $y$.
            Either way, there must be a cut on $y$ in $R$, i.e., $R \congr \new y.(R_1 \|_y R_2)$.
            We show by induction on the structures of $R_1$ and $R_2$ that $R$ is ready.
            \begin{itemize}
                \item If $R_2$ is a spawn with $y$ in the domain, then $R$ is ready by definition; $R_2$ cannot be a spawn with $y$ in the domain.
                \item If $R_1$ or $R_2$ is a forwarder on $x$, then $R$ is ready by definition.
                \item If $R_1$ is a communication prefix on $y$, the analysis depends on whether $y \in \an(R_2)$.
                    If so, $R$ is ready by defition.
                    Otherwise, there is a name $z$ such that the action on $y$ in $R_2$ is guarded by a spawn prefix with $z$ in the domain or by a communication prefix on $z$.
                    Either way, there must be a cut on $z$ in $R_2$, i.e., $R_2 \congr \new z.(R'_2 \|_z R''_2)$.
                    By the IH, $R'_2$ or $R''_2$ is ready, so $R_2$ is ready by definition.
                    Hence, $R$ is ready by definition.
                \item The case where $R_2$ is a communication prefix on $y$ is analogous to the previous case.
                \item If $R_1$ or $R_2$ is a cut, then by the IH, $R_1$ or $R_2$ is ready, so $R$ is ready by definition.
            \end{itemize}
            Since $R$ is ready, also $P$ is ready.
        \item $R$ is a communication prefix on $x$.
            This case is analogous to the previous case.
    \end{itemize}
    In each case, the assumption that $P$ is not ready is contradicted, so $P \congr \out z<>$ or $P \congr \spawn{\emptyset}.\out z<>$.
\end{proof}

\subsection{Weak normalization}
\label{appendix:sec:wn}
Recall our normalization strategy:
If a process can perform a communication reduction or a forwarder reduction, then we do exactly that reduction.
If a process can only perform a reduction that involves a spawn prefix, then we
\begin{enumerate*}
\item select (an active) spawn prefix with the least depth;
\item perform the spawn reduction;
\item propagate the newly created spawn prefix to the very top-level, merging it with other spawn prefixes along the way.
\end{enumerate*}

We will show that this reduction strategy terminates, by assigning a particular lexicographical measure to the processes and showing that our strategy strictly reduces this measure.
This measure counts the number of communication prefixes in a process at a given depth, where depth is determined by spawn prefixes.
Let us make this precise.

For a process $P$ we consider its \emph{skeleton} $\skelOf(P)$, which is a finite map assigning to each number $n$ the amount of communication prefixes at depth $n$ and above.
Since the processes are finite, each communication prefix occurs at a finite depth.
That means that $\skelOf(P)(k) = 0$ for any $k$ greater than the maximal depth of the process.
Formally, we define skeletons as follows.
\begin{definition}[Skeleton]
  A function $\skel \from \Nat \to \Nat$
  is a \emph{skeleton of depth~$k$},
  if $ \A i>k. \skel(i) = 0 $.

  We define $\Skel_k$ as the set of all skeletons of depth~$k$.
  Moreover, we define $\Skel \is \Union_{k\in\Nat} \Skel_k$
  and equip it with the strict quasi order $\skelLt$ such that
  \[
  \skel_1 \skelLt \skel_2
  \iff
    \E j. \bigl(
      \skel_1(j) < \skel_2(j)
      \land
      \A i>j. \skel_1(i) = \skel_2(i)
    \bigr)
  \]
  We also define $
    \skel_1 \skelLe \skel_2 \is
      (\skel_1 = \skel_2 \lor \skel_1 \skelLt \skel_2).
  $
\end{definition}

The following lemmas allow us to do well-founded recursion on skeletons.
\begin{lemma}
\label{lm:skel-lt-pres-k}
  If $\skel_2 \in \Skel_k$ and $\skel_2 \skelGt \skel_1$
  then $ \skel_1 \in \Skel_k $.
\end{lemma}
\begin{proof}
  From $\skel_2 \skelGt \skel_1$ we get some~$j$ such that
  $
    \skel_1(j) < \skel_2(j)
  $ and $
    \A i>j. \skel_1(i) = \skel_2(i)
  $.
  \begin{casesplit}
  \case[$j \leq k$]
    Then $ \A i>k. \skel_1(i) = \skel_2(i) = 0 $
    which proves $ \skel_1 \in \Skel_k $.
  \case[$j > k$]
    This is impossible as we would have
    $\skel_1(j) < \skel_2(j) = 0$.
  \qedhere
  \end{casesplit}
\end{proof}
\begin{lemma}
\label{lm:skel-k-wf}
  $ (\Skel_k, \skelLt) $ is well-founded.
\end{lemma}
\begin{proof}
  Towards a contradiction,
  assume $ \skel_0 \skelGt \skel_1 \skelGt \dots $
  is an infinite descending chain of $\Skel_k$.
  Let~$j_n$ be the witness for $\skel_n \skelGt \skel_{n+1}$,
  \ie $ \skel_n(j_n) > \skel_{n+1}(j_n) $ and
      $ \A i>j_n. \skel_n(i) = \skel_{n+1}(i) $.
  From $s_n, s_{n+1} \in \Skel_k$ we get $ j_n < k $.
  Since there are finitely many natural numbers below~$k$,
  by the pigeonhole principle,
  the sequence $j_0, j_1, \dots$
  contains at least one number that repeats infinitely often.
  Among the ones that do, pick the greatest to be~$m$.
  By definition, all the numbers larger than~$m$ appear finitely often
  in $j_0, j_1, \dots$ and so there is a position~$ p $
  such that~$\A n\geq p.j_n \leq m$.
  We obtain that
  $
    \A n \geq p.
      \skel_n(m) \geq \skel_{n+1}(m)
  $.
  Moreover, let $ i_0 < i_1 < \dots $
  be such that~$j_{i_0}, j_{i_1}, \dots$
  consists of the infinite subsequence
  of the occurrences of~$m$ in $j_p, j_{p+1}, \dots$,
  \ie $m = j_{i_0} = j_{i_1} = \dots$.
  We have
  $
    \skel_{i_0}(m)
    >
    \skel_{i_0+1}(m)
    \geq
    \dots
    \geq
    \skel_{i_1}(m)
    >
    \skel_{i_1+1}(m)
    \geq
    \dots
  $.
  We obtain that
  $ \skel_{i_0}(m) > \skel_{i_1}(m) > \dots $
  is an infinite descending chain of~$\Nat$,
  which is a contradiction.
\end{proof}

\begin{restatable}{lemma}{skelwf}
  \label{lm:skel-wf}
  $ (\Skel, \skelLt) $ is well-founded.
\end{restatable}
\begin{proof}
  Towards a contradiction,
  assume $ \skel_0 \skelGt \skel_1 \skelGt \dots $
  is an infinite descending chain of $\Skel$.
  Since $\skel_0 \in \Skel_k$ for some~$k$,
  by \cref{lm:skel-lt-pres-k}
  $ \A i. \skel_i \in \Skel_k $.
  Therefore we have an infinite descending chain
  of $\Skel_k$ which contradicts \cref{lm:skel-k-wf}.
\end{proof}

For a process $P$, its \emph{skeleton} $\skelOf(P)$ is a finite map assigning to each number $n$ the amount of communication prefixes at depth $n$.
\begin{definition}[Skeleton of~$P$]
\label{def:skel-of-P}
  Given $\skel, \skel_1,\skel_2\in\Skel$,
  we define:
  \begin{align*}
\skelOne(i) &\is
      \begin{cases}
        1 \CASE i=0
        \\
        0 \OTHERWISE
      \end{cases}
    &
    (\skel_1 \skelPlus \skel_2)(i) &\is
      \skel_1(i) + \skel_2(i)
    &
    (\skelShift{\skel})(i) &\is
      \begin{cases}
        \skel(0) \CASE i=0
        \\
        \skel(i-1) \CASE i>0
      \end{cases}
  \end{align*}

  The \emph{skeleton of a process~$P$},
  written $\skelOf(P)$,
  is then defined as:
  \begin{align*}
    \skelOf(P) &\is
    \begin{cases}
      \skelOne
        \CASE P = \fwd[x<-y] \lor P = \out x<>
      \\
      \skelOne \skelPlus \skelOf(Q_1) \skelPlus \skelOf(Q_2)
        \CASE P = \out x[y].(Q_1 \| Q_2) \lor P = \caseLR{x}{Q_1}{Q_2}
      \\
      \skelOne \skelPlus \skelOf(Q)
        \CASE P = \inp x().Q \lor P = \inp x(y).Q \lor P = \sel{x}{\ell}.Q
      \\
      \skelOf(Q_1) \skelPlus \skelOf(Q_2)
        \CASE P = \new x.(Q_1\|Q_2)
      \\
      \skelShift{\skelOf(Q)}
        \CASE P = \spw\senv.Q
    \end{cases}
  \end{align*}

Note that if $ \skel = \skelOf(P) $ then
  $ \A i.\skel(i) \geq \skel(i+1) $.
\end{definition}
For example, the skeleton
\[
  \skelOf(\new x.\big(\out x<> \| \spawn{\spvar}.\new y.(\out y<>\|\spawn{\spvar'}.\inp x().\inp y().\out k<>)\big))
  =
  \map{0 -> 5, 1 -> 4; 2 -> 3; \wtv -> 0}.
\]

\paragraph{The measure.}
Recall from the main part of the paper, that when computing a measure associated to the process we have to take special care of the top-level spawn prefix.
We define the measure function $\mu$ as follows.1
\begin{equation*}
    \mu(P) \is
    \begin{cases}
      \skelOf(Q)
        \CASE P = \spw\senv.Q
      \\
      \skelOf(P)
        \OTHERWISE
    \end{cases}
\end{equation*}

\begin{restatable}{lemma}{measureClosedCong}
  \label{lem:measure_closed_congr}
  If $P \congr Q$  then $\mu(P) = \mu(Q)$.
\end{restatable}
\begin{proof}
  None of the congruences can change whether the top-level construct is a spawn prefix.
  Furthermore, none of the congruences change depth of any communication prefixes.
\end{proof}

\begin{restatable}{lemma}{measureCommReductions}
  \label{lem:measure_comm_reductions}
  The communication reductions strictly decrease the measure.
  That is, reductions \ref{red-unit-l}, \ref{red-comm-l}, \ref{red-comm-r}, \ref{red-case}, \ref{red-fwd-l}, \ref{red-fwd-r} decrease the measure $\mu$, even when occurring under arbitrary evaluation contexts.

  Similarly, reduction \ref{red-spawn-merge} strictly decreases the measure.
\end{restatable}
\begin{proof}
  Each of those reductions reduce the amount of communication prefixes at a given depth, and, as such, decrease the skeleton of the process.
  The only thing that we need to note is the special spawn prefix condition on $\mu$ in cases \ref{red-unit-l}, \ref{red-fwd-l}, and \ref{red-fwd-r}.
  In those cases, the reduction might introduce a spawn prefix in front of the process.
  However, in that case the measure $\mu$ will still strictly decrease.
\end{proof}

As we have seen, the spawn reductions \ref{red-spawn}, \ref{red-spawn-l}, and \ref{red-spawn-r} might temporarily increase the measure, but if we repeat them long enough then the measure will actually decrease.
\begin{restatable}{lemma}{measureRhoReductions}
  \label{lem:measure_rho_reductions}
  Let $\ectxt_0[\hole]$ be a non-empty evaluation context which may contain a $(\spawn{\spvar_0}.[\hole])$ sub-context only at the top level.
  Let $\ectxt_0[\spawn{\spvar}.Q]$ be a process.
  In other words, $\spawn{\spvar}$ is an active prefix spawn at the least depth in $\ectxt_0[\spawn{\spvar}.Q]$.
  Then there exists a spawn binding $\spvar'$, and an evaluation context $\ectxt_1[\hole]$ which is free of $(\spawn{\spvar_1}.[\hole])$ sub-contexts for any $\spvar_1$, such that
  \begin{align*}
    \ectxt_0[\spawn{\spvar}.Q] \redto* \spawn{\spvar'}.\ectxt_1[Q]
    &&
  \text{and}
    &&
    \mu(\ectxt_0[\spawn{\spvar}.Q]) > \mu(\spawn{\spvar'}.\ectxt_1[Q]).
  \end{align*}
\end{restatable}
\begin{proof}
  We first show that the last condition follows from the previous ones.
  By definition, ${\mu(\spawn{\spvar'}.\ectxt_1[Q]) = \skelOf(\ectxt_1[Q])}.$
  We then consider two situations.
  If $\ectxt_0[\spawn{\spvar}.Q]$ does not have a spawn prefix at the top level, then
  $$\mu(\ectxt_0[\spawn{\spvar}.Q]) = \skelOf(\ectxt_0[\spawn{\spvar}.Q]) > \skelOf(\ectxt_1[Q]),$$
  as the later process has less spawn prefixes.
  On the other hand, if $\ectxt_0[\spawn{\spvar}.Q]$ begins with a spawn prefix at the top level, that prefix cannot be $\spawn{\spvar}$ itself, as $\ectxt_0$ is non-empty.
  Then the process $\ectxt_0[\spawn{\spvar}.Q]$ is of the form $\spawn{\spvar'}.\ectxt[\spawn{\spvar}.Q]$, and we have
  \[
    \mu(\spawn{\spvar'}.\ectxt[\spawn{\spvar}.Q]) = \skelOf(\ectxt[\spawn{\spvar}.Q]) > \skelOf(\ectxt_1[Q]).
  \]

  Thus, we only need to find an adequate context $\ectxt_1[\hole]$ and establish the reduction.
  We prove this by induction on the size of the evaluation context $\ectxt_0[\hole]$.
  We do a case analysis on the ``tail'' of the evaluation context.

  \begin{induction}
    \step[Case $\ectxt_0$ is of the form ${\ectxt_0'[\spawn{\spvar'}.{[\hole]}]}$]
    If $\ectxt_0$ contains the $\spawn{\spvar'}.[\hole]$, then by our assumption, it is on the top level.
    That means that $\ectxt_0'$ is empty.
    We then apply the reduction \ref{red-spawn-merge}:
    \[
      \spawn{\spvar'}.\spawn{\spvar}.Q \redd \spawn{\spvar' \merge \spvar}.Q.
    \]
    Then pick $\ectxt_1$ to be empty.
    \step[Case ${\ectxt_0[\spawn{\spvar}.Q]}$ is of the form ${\ectxt_0'[\new x. (P \| \spawn{x \mapsto x_1,x_2}.Q)]}$]
    We then have a reduction
    \[
      \ectxt_0'[\new x. (P \| \spawn{x -> x_1,x_2}.Q)]
      \redd
      \ectxt_0'[
        \spawn{z \mapsto z_1, z_2}.
          \new x_1.(
             \idx{P}{1}
          \| \new x_2. (\idx{P}{2} \| Q)
          )
      ],
    \]
    if $\fn(P) = \{x, z\}$.

    If $\ectxt_0'$ is empty, then we are done.
    If it is not, then by the induction hypothesis we then have
    \begin{multline*}
      \ectxt_0'[
        \spawn{z \mapsto z_1, z_2}.
          \new x_1.(\idx{P}{1} \| \new x_2. (\idx{P}{2} \| Q))
      ]
      \redto*
      {}\\
      \spawn{\spvar'}.
        \ectxt_1[
          \new x_1.(\idx{P}{1} \| \new x_2. (\idx{P}{2} \| Q))
        ],
    \end{multline*}
    which we chain with the original \ref{red-spawn} reduction.

    Other cases are handled similarly.
    \qedhere
  \end{induction}
\end{proof}

\weakNormThm*
\begin{proof}
  We give a normalization procedure as follows.
  Given a process $P$, we consider its possible reductions, and apply them in order that would decrease the measure $\mu$.
  We repeat this until we reach a normal form.
  Since the measure is strictly decreasing, this procedure will terminate by \Cref{lm:skel-wf}.

  Thanks to \Cref{lem:measure_closed_congr} we can consider possible reductions of $P$ up to congruence.
  Let us consider which reductions can apply to $P$.
  \begin{casesplit}
    \case[\ref{red-spawn}, \ref{red-spawn-l}, \ref{red-spawn-r}, or \ref{red-spawn-merge}]
    In that case we find a spawn prefix, involved with such a reduction, with the least depth.
    Then this spawn prefix will satisfy the conditions of \Cref{lem:measure_rho_reductions}, and we pull out this active prefix upfront, decreasing the measure.
    \case[communication or forwarder reductions]
    In that case we apply that exact reduction, which by \Cref{lem:measure_comm_reductions} will decrease the measure. \qedhere
  \end{casesplit}
\end{proof}

 \section{Encoding the \alcalc}
\label{sec:appendix:translation}
The \alcalc type system follows the presentations of \alcalc by \citet{ohearn:2003} and \citet[Chapter 2]{pym:2002}, but we adjusted the elimination rule for additive units to match the corresponding rule for multiplicative units.
For reasons of space, we have omitted products and coproducts from \Cref{sec:translation}; these are present here.
\begin{figure}[t]
\adjustfigure
\begin{proofrules}
\infer*[lab=N-id]
{}
{x : A \vdash x : A}
\label{rule:N-id}

\infer*[lab=N-$\equiv$]
{\Delta \vdash M : A \and \Delta \equiv \Gamma}
{\Gamma \vdash M : A}
\label{rule:N-equiv}

\infer*[lab=N-W]
{\Gamma(\Delta) \vdash M : A}
{\Gamma(\Delta\band \Delta') \vdash M : A}
\label{rule:N-W}

\infer*[lab=N-C]
{\Gamma(\Delta^{(1)}\band \Delta) \vdash M : A}
{\Gamma(\Delta) \vdash \substS{M}{\ident(\Delta^{(1)})}{\ident(\Delta)} : A}
\label{rule:N-C}

\infer*[lab=$\wand$I]
{\Delta\bsep x : A \vdash M : B}
{\Delta \vdash \Lam x. M : A \wand B}
\label{rule:wand-intro}

\infer*[lab=$\to$I]
{\Delta\band x : A \vdash M : B}
{\Delta \vdash \Alp x. M : A \to B}
\label{rule:impl-intro}

\infer*[lab=$\wand$E]
{\Delta_1 \vdash M : A \wand B
\and \Delta_2 \vdash N : A}
{\Delta_1\bsep \Delta_2 \vdash M\ N : B}
\label{rule:wand-elim}

\infer*[lab=$\to$E]
{\Delta_1 \vdash M : A \to B
\and \Delta_2 \vdash N : A}
{\Delta_1\band \Delta_2 \vdash M @ N : B}
\label{rule:impl-elim}

\infer*[lab=$\mOne$I]
{}
{\empM \vdash \unitM : \mOne}
\label{rule:mone-intro}

\infer*[lab=$\aOne$I]
{}
{\empA \vdash \unitA : \aOne}
\label{rule:aone-intro}

\infer*[lab=$\mOne$E]
{\Delta \proves M : \mOne
\and \Gamma(\empM) \proves N : A}
{\Gamma(\Delta) \proves \Let \unitM = M in N : A}
\label{rule:mone-elim}

\infer*[lab=$\aOne$E]
{\Delta \proves M : \aOne
\and \Gamma(\empA) \proves N : A}
{\Gamma(\Delta) \proves \Let \unitA = M in N : A}
\label{rule:aone-elim}

\infer*[lab=$\ast$I]
{\Delta_1 \proves M : A \and
  \Delta_2 \proves N : B}
{\Delta_1\bsep \Delta_2 \proves \langle M, N \rangle : A \ast B}
\label{rule:ast-intro}

\infer*[lab=$\wedge$I]
{\Delta_1 \proves M : A \and
  \Delta_2 \proves N : B}
{\Delta_1\band \Delta_2 \proves (M, N) : A \wedge B}
\label{rule:conj-intro}

\infer*[lab=$\ast$E]
{\Delta \proves M : A \ast B \and
  \Gamma(x : A\bsep y : B) \proves N : C}
{\Gamma(\Delta) \proves \Let \langle x, y\rangle = M in N : C}
\label{rule:ast-elim}

\infer*[lab=$\wedge$E]
{\Delta \proves M : A_1 \wedge A_2 \and i \in \set{0,1}
}
{\Delta \proves \pi_i(M) : A_i}
\label{rule:conj-elim}

\infer*[lab=$\vee$I]
{\Delta \proves N : A_i \and i \in \{1,2\}}
{\Delta \proves \inj_i(N) : A_1 \vee A_2}
\label{rule:disj-intro}

\infer*[lab=$\vee$E]
{\Delta \proves M : A_1 \vee A_2
\and
\Gamma(x_1 : A_1) \proves N_1 : C
\and
\Gamma(x_2: A_2) \proves N_2 : C}
{\Gamma(\Delta) \proves \Case M of x_1 => N_1 or x_2=> N_2}
\label{rule:disj-elim}

\infer*[lab=N-cut]
{\Delta \proves M : A \and
\Gamma(x : A) \proves N : C}
{\Gamma(\Delta) \proves \substS{N}{x}{M} : C}
\label{rule:N-cut}
\end{proofrules}

 \caption{Typing system for \alcalc.}
\label{fig:alpha_lambda_typing_full}
\end{figure}
The type system is given in \cref{fig:alpha_lambda_typing_full}.
Note that the \ruleref{N-cut} is admissible (cf.\ \cite{ohearn:2003}).
\begin{figure}
\adjustfigure
\rulesection*{Primitive reductions}
\begin{proofrules}
  \infer*[lab=red-beta-$\lambda$]
  {}
  {(\Lam x. M)\ N \alredto \substS{M}{x}{N}}
  \label{rule:red-beta-lam}

  \infer*[lab=red-beta-$\alpha$]
  {}
  {(\Alp x. M) @ N \alredto \substS{M}{x}{N}}
  \label{rule:red-beta-alp}

  \infer*[lab=red-proj,Right=$i \in \set{1,2}$]
  { }
  {\pi_i(M_1, M_2) \alredto M_i}
  \label{rule:red-proj}

  \infer*[lab=red-unitM]
  {}
  {\Let \unitM = \unitM in M \alredto M}
  \label{rule:red-unitM}

  \infer*[lab=red-unitA]
  {}
  {\Let \unitA = \unitA in M \alredto M}
  \label{rule:red-unitA}

  \infer*[lab=red-pair]
  {}
  {\Let \langle x,y \rangle = \langle M_1 , M_2 \rangle in N \alredto N\subst{x->M_1,y->M_2}}
  \label{rule:red-pair}

  \infer*[lab=red-case,Right=$i \in \set{1,2}$]
  { }
  {\Case \inj_i(M) of x_1 => N_1 or x_2 => N_2 \alredto N_i\subst{x_i -> M}}
  \label{rule:red-case}
\end{proofrules}

\rulesection{Lifted reductions}
\begin{proofrules}
  \infer
  {M \alredto M'}
  {M\ N \alredto M'\ N}

  \infer
  {M \alredto M'}
  {M @ N \alredto M' @ N}

  \infer
  {M \alredto M'}
  {\Let p = M in N \alredto \Let p = M' in N}

  \infer
  {M \alredto M'}
  {\pi_i M \alredto \pi_i M'}

  \infer
  {M \alredto M'}
  {\Case M of x_1 => N_1 or x_2 => N_2 \alredto \Case M' of x_1 => N_1 or x_2 => N_2}
\end{proofrules}
\rulesectionend
 \caption{Reduction rules for \alcalc.}
\label{fig:alcalc_red_rules_full}
\end{figure}
We consider call-by-name reduction strategy, the reduction relation for which is given in \cref{fig:alcalc_red_rules_full}.

The translation function $\bTrans_z(-)$ is defined by recursion on the typing derivation and is given in \cref{fig:translation_big,fig:translation_big_two}.
\begin{figure}
  \adjustfigure[\footnotesize]
  \centering
  \begin{tabular}{@{}c c@{}}
  \toprule
  \alcalc typing of~$M_0$ & \piBI encoding~$\bTrans_z(M_0)$
\\ \midrule
  $ \infer{}{x : A \proves x : A} $
  &
  $ \infer{}{x : A \proves \fwd[z<-x] :: z : A} $
  \\[1em]$ \infer{
    \Gamma(\Delta) \proves M : A
  }{
    \Gamma(\Delta \band \Delta') \proves M : A
  } $
  &
  $ \infer{
    \Gamma(\Delta) \proves \Trans_z(M) :: z : A
  }{
    \Gamma(\Delta\band \Delta') \proves
    \spawn{x -> \emptyset | x \in \ident(\Delta')}. \Trans_z(M) :: z : A
  } $
  \\[2em]$ \infer{
    \Gamma(\idx{\Delta}{1} \band \idx{\Delta}{2}) \proves M : A
  }{
    \Gamma(\Delta) \proves M\subst{\idx{\Delta}{1}->\Delta,\idx{\Delta}{2}->\Delta} : A
  } $
  &
  $ \infer{
    \Gamma(\idx{\Delta}{1} \band \idx{\Delta}{2}) \proves \proves \Trans_z(M) :: z : A
  }{
    \Gamma(\Delta) \proves
    \spawn{x -> x_1,x_2 | x \in \ident(\Delta)}. \Trans_z(M) :: z : A
  } $
  \\[2em]$ \infer{
    \Delta \bsep x : A \proves M : B
  }{
    \Delta \proves \Lam x. M : A \wand B
  } $
  &
  $ \infer{
    \Delta \bsep x : A \proves \Trans_z(M) :: z : B
  }{
    \Delta \proves \inp z(x).\Trans_z(M) :: z : A \wand B
  } $
  \\[2em]$ \infer{
    \Delta \band x : A \proves M : B
  }{
    \Delta \proves \Alp x. M : A \to B
  } $
  &
  $ \infer{
    \Delta \band x : A \proves \Trans_z(M) :: z : B
  }{
    \Delta \proves \inp z(x).\Trans_z(M) :: z : A \to B
  } $
  \\[2em]$ \infer{
    \Delta_1 \proves M : A \and \Delta_2 \proves N : B
  }{
    \Delta_1 \bsep \Delta_2 \proves \langle M, N \rangle : A \ast B
  } $
  &
  $ \infer{
    \Delta_1 \proves \Trans_y(M) :: y : A
    \and
    \Delta_2 \proves \Trans_z(N) :: z : B
  }{
    \Delta_1 \bsep \Delta_2 \proves
    \out z[y].(\Trans_y(M) \| \Trans_z(N)):: z :: A \ast B
  } $
  \\[2em]$ \infer{
    \Delta_1 \proves M : A \and \Delta_2 \proves N : B
  }{
    \Delta_1 \band \Delta_2 \proves (M, N) : A \land B
  } $
  &
  $ \infer{
    \Delta_1 \proves \Trans_y(M) :: y : A
    \and
    \Delta_2 \proves \Trans_z(N) :: z : B
  }{
    \Delta_1 \band \Delta_2 \proves
    \out z[y].(\Trans_y(M) \| \Trans_z(N)):: z :: A \land B
  } $
  \\ \bottomrule
\end{tabular}

   \caption{Translation from \alcalc to \piBI (1/2).}
  \label{fig:translation_big}
\end{figure}
\begin{sidewaysfigure}
  \adjustfigure[\footnotesize]
  \centering
  \begin{tabular}{@{}c c@{}}
  \toprule
  \alcalc typing of~$M_0$ & \piBI encoding~$\bTrans_z(M_0)$
  \\ \midrule
  $ \infer{
    \Delta_1 \proves M : A \wand B
    \and
    \Delta_2 \proves N : A
  }{
    \Delta_1 \bsep \Delta_2 \proves M\ N : B
  } $
  &
  $ \infer{
    \Delta_1 \proves \Trans_x(M) :: x : A \wand B
    \and
    \infer*{
      \Delta_2 \proves \Trans_y(N) :: y : A
      \and
      x : B \proves \fwd[z<-x] :: z : B
    }{
      x : A \wand B \bsep \Delta_2 \proves
      \out x[y].\big(\Trans_y(N) \| \fwd[z<-x]\big) :: z : B
    }
  }{
    \Delta_1 \bsep \Delta_2 \proves
    \new x.\bigl(
      \Trans_x(M) \| \out x[y].(\Trans_y(N) \| \fwd[z<-x])
    \bigr) :: z : B
  } $
  \\[2em]$ \infer{
    \Delta_1 \proves M : A \to B
    \and
    \Delta_2 \proves N : A
  }{
    \Delta_1 \band \Delta_2 \proves M\ N : B
  } $
  &
  $ \infer{
    \Delta_1 \proves \Trans_x(M) :: x : A \to B
    \and
    \infer*{
      \Delta_2 \proves \Trans_y(N) :: y : A
      \and
      x : B \proves \fwd[z<-x] :: z : B
    }{
      x : A \to B \band \Delta_2 \proves
      \out x[y].\big(\Trans_y(N) \| \fwd[z<-x]\big) :: z : B
    }
  }{
    \Delta_1 \band \Delta_2 \proves
    \new x.\bigl(
      \Trans_x(M) \| \out x[y].(\Trans_y(N) \| \fwd[z<-x])
    \bigr) :: z : B
  } $
  \\[2em]$ \infer{
    \Delta \proves M : A_1 \land A_2
  }{
    \Delta \proves \proj_1(M) : A_1
  } $
  &
  $ \infer{
    \Delta \proves \Trans_{x_2}(M) :: x_2 :  A_1 \land A_2
    \and
    \infer*{
      \infer*{
        x_1 : A_1 \proves \fwd[z<-x_1] :: z : A_1
      }{
        x_1 : A_1 \band x_2 : A_2 \proves
        \spawn{x_{2} -> \emptyset}.\fwd[z<-x_1] :: z : A_1
      }
    }{
      x : A_1 \land A_2 \proves
      x_2(x_1).\spawn{x_{2} -> \emptyset}.\fwd[z<-x_1] :: z : A_1
    }
  }{
    \Delta \proves
    \new x_2. \big(
      \Trans_{x_2}(M) \| x_2(x_1).\spawn{x_{2} -> \emptyset}. \fwd[z<-x_1]
    \big) :: z : A_1
  } $
  \\[2em]$ \infer{
    \Delta \proves M : A_1 \ast A_2
  \and \bctxt(x : A_1\bsep y:A_2) \proves N : C
  }{
    \bctxt(\Delta) \proves \Let \langle x, y\rangle = M in N : C
  } $
  &
  $ \infer{
    \Delta \proves \Trans_{y}(M) :: y :  A_1 \ast A_2
    \and
    \infer*{
        \bctxt(x : A_1 \bsep y : A_2) \proves \Trans_z(N) :: z : C
    }{
        \bctxt(y : A_1 \ast A_2) \proves \inp y(x).\Trans_z(N) :: z : C
    }
    }{
        \bctxt(\Delta) \proves \new y. (\Trans_y(M) \| \inp y(x).\Trans_z(N)) :: z : C
    } $
  \\[2em]$ \infer{
    \Delta \proves N : A_1
  }{
    \Delta \proves \inj_1(N) : A_1 \lor A_2
  } $
  &
  $ \infer{
    \Delta \proves \Trans_{z}(N) :: z :  A_1
    }{
    \Delta \proves \sendInl{z}.\Trans_z(N) :: z : A_1 \lor A_2
    } $
  \\[2em]$ \infer{
    \Delta \proves M : A_1 \lor A_2
  \and \bctxt(x_1 : A_1) \proves N_1 : C
  \and \bctxt(x_2 : A_2) \proves N_2 : C
  }{
  \bctxt(\Delta) \proves \Case M of x_1 => N_1 or x_2 => N_2
  } $
  &
  $ \infer{
    \Delta \proves \Trans_{x}(M) :: x :  A_1 \lor A_2
    \infer*{
        \bctxt(x_1 : A_1) \proves \Trans_{z}(N_1) :: z :  C
         \and
        \bctxt(x_2 : A_2) \proves \Trans_z(N_2) :: z :  C
      }{
      \bctxt(x : A_1 \lor A_2) \proves \caseLR{x}{\Trans_{z}(N_1)\subst{x_1 -> x}}{\Trans_{z}(N_2)\subst{x_2-> x}} :: z : C
      }
    }{
    \bctxt(\Delta) \proves \new x. (\Trans_{x}(M) \| \caseLR{x}{\Trans_{z}(N_1)\subst{x_1 -> x}}{\Trans_{z}(N_2)\subst{x_2-> x}})
    :: z : C
    } $
  \\ \bottomrule
\end{tabular}

   \caption{Translation from \alcalc to \piBI (2/2).}
  \label{fig:translation_big_two}
\end{sidewaysfigure}

\subsection{Operational correspondence}
We split the proof of completeness into two parts.
First, we show that if a term can reduce, then this reduction is matched by the translated process, and the resulting term and process diverge up to substitution lifting.
\begin{restatable}[Basic completeness]{lemma}{basicCompleteness}
  \label{lem:basic_completeness}
  Given $\bunch \proves M : A$, if $M \alredto N$, then there exists $Q$ such that $\Trans_z(M) \reddStar Q \sublift N$.
\end{restatable}
\begin{proof}
    By induction on the derivation of $M \alredto N$.
    There are six base cases:
    \begin{casesplit}
        \case[\ref{rule:red-beta-lam}]
        We have $(\Lam x. M)\ N \alredto \substS{M}{x}{N}$.
        We then have
        \begin{align*}
          \Trans_z((\Lam x. M)\ N) & = \new y.(\inp y(x).\Trans_y(M) \| \out y[x].(\Trans_x(N) \| \fwd[z<-y])) \\
          & \redto \new y.(\new x.(\Trans_x(N) \| \Trans_y(M)) \| \fwd[z<-y]) \\
          & \redto \new x.(\Trans_x(N) \| \Trans_z(M))
            \sublift \substS{M}{x}{N}.
        \end{align*}

        \case[\ref{rule:red-beta-alp}]
        Analogous to Case~\ref{rule:red-beta-lam}.

        \case[\ref{rule:red-proj}]
        We have $\pi_i(M_1,M_2) \alredto M_i$ for $i \in \set{1,2}$.
        Let $i' \in \set{1,2} \setminus \set{i}$.
        \begin{align*}
            \Trans_z(\pi_i(M_1,M_2)) &\eqto \new x_1.(\out {x_1}[x_2].(\Trans_{x_2}(M_2) \| \Trans_{x_1}(M_1)) \| \inp x_1(x_2).\spw{x_{i'}->\emptyset}.\fwd[z<-x_i])
            \\
            &\redto \new x_1.(\Trans_{x_1}(M_1) \| \new x_2.(\Trans_{x_2}(M_2) \| \spw{x_{i'}->\emptyset}.\fwd[z<-x_i]))
            \\
            &\redto \spw{z->\emptyset | z \in \fn(M_{i'})}.\new x_i.(\Trans_{x_i}(M_i) \| \fwd[z<-x_i])
            \\
            &\redto \spw{z->\emptyset | z \in \fn(M_{i'})}.\Trans_z(M_i)
            \\
            &\sublifto M_i
        \end{align*}

        \case[\ref{rule:red-unitM}]
        We have $\Let \unitM = \unitM in M \alredto M$.
        \begin{align*}
            \Trans_z(\Let \unitM = \unitM in M \redd M)
            &\eqto \new x.(\out x<> \| \inp x().\Trans_z(M))
            \\
            &\redto \Trans_z(M)
            \\
            &\sublifto M
        \end{align*}

        \case[\ref{rule:red-pair}]
        We have $\Let \langle x,y \rangle = \langle M_1, M_2 \rangle in N \alredto N\subst{x->M_1,y->M_2}$.
        \begin{align*}
            \Trans_z(\Let \langle x,y \rangle = \langle M_1, M_2 \rangle in N) &\eqto \new x.(\out x[y].(\Trans_y(M_2) \| \Trans_x(M_1)) \| \inp x(y).\Trans_z(N))
            \\
            &\redto \new x.(\Trans_x(M_1) \| \new y.(\Trans_y(M_2) \| \Trans_z(N)))
            \\
            &\congrto \new y.(\Trans_y(M_2) \| \new x.(\Trans_x(M_1) \| \Trans_z(N)))
            \\
            &\sublifto N\subst{x->M_1,y->M_2}
        \end{align*}

        \case[\ref{rule:red-case}]
        We have,  for ${i \in \set{1,2}}$, that \[(\Case \inj_i(M) of x_1 => N_1 or x_2 => N_2) \alredto \substS{N_i}{x_i}{M}.\]
        Expanding the translation:
        \begin{align*}
            &\phantom{\eqto} \Trans_z(\Case \inj_i(M) of x_1 => N_1 or x_2 => N_2)
            \\
            &\eqto \new x_1.(\Trans_{x_1}(\inj_i(M)) \| \caseLR{x_1}{\Trans_z(N_1)}{\new x_2.(\fwd[x_2<-x_1] \| \Trans_z(N_2))})
        \end{align*}
        There are two cases for $i \in \set{1,2}$.
        \begin{casesplit}
            \case[$i = 1$]
            \begin{align*}
                &\phantom{\eqto} \new x_1.(\Trans_{x_1}(\inj_1(M)) \| \caseLR{x_1}{\Trans_z(N_1)}{\new x_2.(\fwd[x_2<-x_1] \| \Trans_z(N_2))})
                \\
                &\eqto \new x_1.(\selL{x_1}.\Trans_{x_1}(M) \| \caseLR{x_1}{\Trans_z(N_1)}{\new x_2.(\fwd[x_2<-x_1] \| \Trans_z(N_2))})
                \\
                &\redto \new x_1.(\Trans_{x_1}(M) \| \Trans_z(N_1))
                \\
                &\sublifto \substS{N_1}{x_1}{M}
            \end{align*}

            \case[$i = 2$]
            \begin{align*}
                &\phantom{\eqto} \new x_1.(\Trans_{x_1}(\inj_2(M)) \| \caseLR{x_1}{\Trans_z(N_1)}{\new x_2.(\fwd[x_2<-x_1] \| \Trans_z(N_2))})
                \\
                &\eqto \new x_1.(\selR{x_1}.\Trans_{x_1}(M) \| \caseLR{x_1}{\Trans_z(N_1)}{\new x_2.(\fwd[x_2<-x_1] \| \Trans_z(N_2))})
                \\
                &\redto \new x_1.(\Trans_{x_1}(M) \| \new x_2.(\fwd[x_2<-x_1] \| \Trans_z(N_2)))
                \\
                &\redto \new x_1.(\Trans_{x_1}(M) \| \substS{\Trans_z(N_2)}{x_2}{x_1})
                \\
                &\congrto \new x_2.(\Trans_{x_2}(M) \| \Trans_z(N_2))
                \\
                &\sublifto \substS{N_2}{x_2}{M}
            \end{align*}
        \end{casesplit}
    \end{casesplit}

    The inductive cases all concern the lifted reductions.
    Each case is analogous, so we only detail the arbitrarily chosen case of reduction lifting under $\lambda$-application.
    Assume $M \alredto M'$.
    We have $M\ N \alredto M'\ N$.
    By the IH, $\Trans_y(M) \reddStar P \sublift M'$.
    Hence, assuming $M' = M''\subst{x_1->N_1,,x_n->N_n}$, we have $P \congr \new x_n.(\Trans_{x_n}(N_n) \| \ldots \new x_1.(\Trans_{x_1}(N_1) \| \Trans_y(M'')) \ldots )$.
    Moreover,
    \[
        M'\ N = (M''\subst{x_1->N_1,,x_n->N_n})\ N = (M''\ N)\subst{x_1->N_1,,x_n->N_n}.
    \]
    We have the following:
    \begin{align*}
        \Trans_z(M\ N) &\eqto \new y.(\Trans_y(M) \| \inp y(w).\Trans_z(N))
        \\
        &\redto* \new y.(\new x_n.(\Trans_{x_n}(N_n) \| \ldots \new x_1.(\Trans_{x_1}(N_1) \| \Trans_y(M'')) \ldots ) \| \inp y(w).\Trans_z(N))
        \\
        &\congrto \new x_n.(\Trans_{x_n}(N_n) \| \ldots \new x_1.(\Trans_{x_1}(N_1) \| \new y.(\Trans_y(M'') \| \inp y(w).\Trans_z(N))) \ldots )
        \\
        &\sublifto M'\ N
        \tag*{\qedhere}
    \end{align*}
\end{proof}

The statement of \Cref{lem:basic_completeness} cannot be chained to form a simulation diagram, since the premise does not start with the substitution relation as in the result.
The full version of completeness starts with a term and a process that are related via substitution lifting:
\completenessThm*
\begin{proof}
  Since $P \sublift M$, we can write the latter as $M'\subst{x_1->N_1,,x_n->N_n}$.
  We then consider two cases, depending on whether the reduction $M \alredto N$ already happens in $M'$, or whether this reduction is triggered by one of the substitutions.

  In the former case we already have a reduction $M' \alredto N'$ that is ``lifted'' to the reduction $M'\subst{x_1->N_1,,x_n->N_n} \alredto N'\subst{x_1->N_1,,x_n->N_n}$.
  We can then appeal directly to the previous lemma to obtain a process $Q$ such that $\Trans_z(M') \reddStar Q \sublift N'$.
  Then,
  \begin{align*}
    & \new x_1.(\Trans_{x_1}(N_1 \| \dots \new x_n.(\Trans_{x_n}(N_n) \| \Trans_z(M')) \dots ) \reddStar \\
    & \new x_1.(\Trans_{x_1}(N_1 \| \dots \new x_n.(\Trans_{x_n}(N_n) \| Q) \dots ) \sublift N'\subst{x_1->N_1,,x_n->N_n}.
  \end{align*}

  In the second case, the reduction in the term is only enabled after some substitution $\subst{x_i->N_i}$ is performed.
  The idea is to reduce this to the first case, by explicitly performing the substitution $\subst{x_i->N_i}$ in the corresponding processes.

  If $\subst{x_i->N_i}$ is the substitution that enables the reduction $M'\subst{x_1->N_1,,x_n->N_n} \alredto N$, then the variable $x_i$ is located at a head position in the term $M'$.
  This means that in the translation, the corresponding process $\Trans_c(x_i)$ will not occur under an input/output prefix, which will allow us to eagerly perform the substitution by using the forwarder reduction, combined with the structural congruences and \ref{red-spawn-r}.

  Let us demonstrate what we mean by  an example.
  Suppose that $M = (x_1\ M'')\subst{x_1->N_1,x_2->N_2}$ and $N_1 = \Lam a. T$.
  Clearly, in this case the beta reduction is enabled only after the substitution.
  The corresponding substitution-lifted process can reduce as follows:
  \begin{align*}
  & \new x_2.(\Trans_{x_2}(N_2) \| \new x_1.(\Trans_{x_1} (N_1) \|  {\Trans_z(x_1\ M'')})) = \\
  & \new x_2.(\Trans_{x_2}(N_2) \| \new x_1.(\Trans_{x_1} (N_1) \|
    {\new c.(\Trans_c(x_1) \| \out c[b].(\Trans_b(M'') \| \fwd[z<-c]))})) =\\
  & \new x_2.(\Trans_{x_2}(N_2) \| \new x_1.(\Trans_{x_1} (N_1) \|
    {\new c.(\fwd[c <- x_1] \| \out c[b].(\Trans_b(M'') \| \fwd[z<-c]))})) \redd\\
  & \new x_2.(\Trans_{x_2}(N_2) \| \new x_1.(\Trans_{x_1} (N_1) \|
    \out {x_1}[b].(\Trans_b(M'') \| \fwd[z<-x_1]))) =\\
  & \new x_2.(\Trans_{x_2}(N_2) \| \Trans_z(N_1\ M'')) \sublifto (N_1\ M'')\subst{x_2->N_2}.
  \end{align*}
  The forwarder reduction in that sequence correspond to explicitly performing the substitution $\subst{x_1->N_1}$.
  After that we get a term $(N_1\ M'')\subst{x_2->N_2}$ in which the reduction is enabled prior to the substitution, thus leaving us with the scenario from the case of this theorem.
\end{proof}

\soundnessThm*
\begin{proof}
    By definition, $P \equiv \spw{\spvar_s}.\ldots\spw{\spvar_1}.\new x_n.(\Trans_{x_n}(M_n) \| \ldots \new x_1.(\Trans_{x_1}(M_1) \| \Trans_z(M')) \ldots )$ where $M = M'\subst{x_1->M_1,,x_n->M_n}\map{\spvar_1,\ldots,\spvar_s}$.
    Let us consider possible reductions of $P$.
    First, each parallel subprocess of $P$ may reduce internally.
    Second, one of the subprocesses may be a forwarder, in which case a forwarder reduction is applicable.
    Third, one of the subprocesses may start with a spawn prefix, which can interact with the cuts.
    Note that no message-passing communication between the subprocesses of $P$ is possible, as follows from the definition of the translation.
    We discuss each possible case:
    \begin{itemize}
        \item
            $\Trans_{x_i}(M_i)$ for $i \in [1,n]$ reduces internally, i.e., $\Trans_{x_i}(M_i) \redd Q_i$.
            We apply induction on the derivation of $\Delta \proves M : A$.
            Clearly, $\Trans_{x_i}(M_i) \sublift M_i$.
            Since the typing derivation of $M_i$ is a sub-derivation of the typing derivation of $M$, the IH applies: there exist $N_i$ and $R_i$ such that $M_i \alfredto* N_i$ and $Q_i \redto* R_i \sublift N_i$.
Let \[Q' = \spw{\spvar_s}.\ldots\spw{\spvar_1}.\new x_n.(\Trans_{x_n}(M_n) \| \ldots \new x_i.(Q_i \| \ldots \new x_1.(\Trans_{x_1}(M_1) \| \Trans_z(M')) \ldots ) \ldots );\] we have $P \redto Q'$.
            From $Q'$, all the reductions that were possible from $P$ are still possible.
            However, these reductions are all independent, so we postpone all but further reductions of $Q'$.
            Let \[R = \spw{\spvar_s}.\ldots\spw{\spvar_1}.\new x_n.(\Trans_{x_n}(M_n) \| \ldots \new x_i.(R_i \| \ldots \new x_1.(\Trans_{x_1}(M_1) \| \Trans_z(M')) \ldots ) \ldots );\] we have $Q' \redto* R$.

            By definition, \[R_i \congrto \spw{\spvar'_t}.\ldots\spw{\spvar'_1}.\new y_m.(\Trans_{y_m}(L_n) \| \ldots \new y_1.(\Trans_{y_1}(L_1) \| \Trans_{x_i}(N'_i)) \ldots )\] where $N_i = N'_i\subst{y_1->L_1,,y_m->L_m}\map{\spvar'_1,\ldots,\spvar'_t}$.
            Let $M_0 = M'\subst{x_1->M_1,,x_i->N_i,,x_n->M_n}$; we have $M \alfredto* M_0$.
            Due to the shape of $R_i$, which includes substitutions, weakenings, and contractions in $N_i$, $R$ is not yet of a shape that we can relate to $M_0$.

            First, we have to move the spawn prefixes in $R_i$ to the sequence of spawn prefixes at the beginning of $R$.
            The procedure depends on whether there are $x_{i+1},\ldots,x_n$ that are weakened or contracted by the spawn prefixes in $R_i$.
            This is largely analogous to the latter cases of spawn prefixes commuting and interacting, so here we assume that no weakening or contraction takes place.
            By typability, none of the substitutions in $N_i$ touch the variable $x_i$, so we can commute the cuts in $R_i$ past the cut on $x_i$ in $R$.
            Let \begin{multline*}
              R' =
                \spw{\spvar_s}.\ldots
                  \spw{\spvar_1}.\spw{\spvar'_t}.\ldots
                    \spw{\spvar'_1}.\new x_n.(
                      \Trans_{x_n}(M_n) \| \ldots
                        \new y_m.(
                          \Trans_{y_m}(L_m) \| \ldots
                    \\
                            \new y_1.(
                              \Trans_{y_1}(L_1) \|
                              \new x_i.(
                                \Trans_{x_i}(N'_i) \| \ldots
                                  \new x_1.(\Trans_{x_1}(M_1) \| \Trans_z(M'))
\ldots))
                        \ldots )
                        \ldots )
            \end{multline*}
            We have $R \redto* R'$.
            Moreover, \[M_0 = M'\subst{x_1->M_1,,x_i->N'_i,y_1->L_1,,y_m->L_m,,x_n->M_n}\map{\spvar_1,\ldots,\spvar_s,\spvar'_1,\ldots,\spvar'_t},\] and thus $R' \sublift M_0$.
            This proves the thesis.

        \item
            $\Trans_z(M')$ reduces internally.
            We apply induction on the derivation of $\Delta \proves M : A$ (\ih1); there is a case for typing rule, although not all cases may yield a reduction in $P$.
            In each case, we additionally apply induction on the number $k$ of reductions from $P$ to $Q$ (\ih2), i.e., $P \redd^k Q$.
            Depending on the shape of $P$, and relying on the independence of reductions, we then isolate $k'$ reductions $P \redd^{k'} Q'$ such that $Q' \sublift N'$ and $M \alfredto N'$ (where $k'$ may be different in each case).
            We then have $Q' \redd^{k-k'} Q$, so it follows from \ih2 that $N' \alfredto* N$ and $Q' \redd* R$ such that $R \sublift N$.

            Note that applications of \ih1 yield processes with spawn prefixes that need to be commuted past cuts to bring them to the front of the process, while some of them apply weakening or contraction when meeting certain cuts.
            We explain such procedures in the latter cases of this proof, so here we assume that \ih1 yields processes without spawn prefixes.

            \begin{casesplit}
                \case[\ref{rule:N-id}]
                We have $M' = y$ and $\Trans_z(M') = \fwd[z<-y]$; no reductions are possible.

                \case[\ref{rule:N-equiv}]
                The thesis follows from \ih1 directly.

                \case[\ref{rule:N-W}]
                We have $\Trans_z(M') = \spw{x->\emptyset | x \in \fn(\Delta')}.\Trans_z(M')$.
                There is only one possibility of reduction: $\spw{x->\emptyset | x \in \fn(\Delta')}.\Trans_z(M') \redd \spw{x->\emptyset | x \in \fn(\Delta')}.Q'$.

                By \ih1, there exist $L$ and $R'$ such that $M' \alfredto* L$ and $Q' \redto* R' \sublift L$.
                Then $$R' \congrto \new y_m.(\Trans_{y_m}(L_m) \| \ldots \new y_1.(\Trans_{y_1}(L_1) \| \Trans_z(L')) \ldots )$$ and $L' = L'\subst{y_1->L_1,,y_m->L_m}$.

                Let $R_0 = \new x_n.(\Trans_{x_n}(M_n) \| \ldots \new x_1.(\Trans_{x_1}(M_1) \| \spw{x->\emptyset | x \in \fn(\Delta')}.R'))$; we have $P \redto* R_0$.
                Also, let $M_0 = L'\subst{y_1->L_1,,y_m->L_m,x_1->M_1,,x_n->M_n}$; we have $M \alfredto* M_0$.

                At this point, $R_0$ is not of appropriate shape to relate it to $M_0$ through substitution lifting, because the weakening spawn prefix is in the middle of the substitutions.
                There are two possibilities for reduction here: the spawn interacts with one of the cuts on $x_i$, or the spawn commutes past them all.
                The former is analogous to the case of a spawn prefix in $\Trans_z(M')$ interacting with a cut, which follows the current case.
                In the latter case, let $R = \spw{x->\emptyset | x \in \fn(\Delta')}.\new x_n.(\Trans_{x_n}(M_n) \| \ldots \new x_1.(\Trans_{x_1}(M_1) \| R'))$.
                Now, $R_0 \redto* R$ and $R \sublift M_0$, proving the thesis.

                \case[\ref{rule:N-C}]
                Analogous to Case~\ref{rule:N-W}.

                \case[\ref{rule:wand-intro}]
                We have $M' = \Lam x.M''$ and $\Trans_z(M') = \inp z(x).\Trans_z(M'')$; no reductions are possible.

                \case[\ref{rule:impl-intro}]
                Analogous to Case~\ref{rule:wand-intro}.

                \case[\ref{rule:wand-elim}]
                We have $M' = L_1\ L_2$ and $\Trans_z(M') = \new x.(\Trans_x(L_1) \| \out x[y].(\Trans_y(L_2) \| \fwd[z<-x]))$.
                There are three possible reductions: $\Trans_x(L_1)$ reduces internally, $\Trans_x(L_1)$ is prefixed by a spawn which commutes past the restriction on $x$, or the output on $x$ synchronizes with an input on $x$ in $\Trans_x(L_1)$.
                \begin{itemize}
                    \item
                        $\Trans_x(L_1)$ reduces internally, i.e., $\Trans_x(L_1) \redto Q'$.
                        By \ih1, there exist $N$ and $R'$ such that $L_1 \alfredto* N$ and $Q' \redto* R' \sublift N$.
                        Then $$R' \congrto \new y_m.(\Trans_{y_m}(N_m) \| \ldots \new y_1.(\Trans_{y_1}(N_1) \| \Trans_x(N')) \ldots )$$ and $N = N'\subst{y_1->N_1,,y_m->N_m}.$
                        Let
                        \begin{align*}
                          R_0 =
                            \new x_n.(&
                              \Trans_{x_n}(M_n) \| \ldots
                              \\&\quad
                              \new x_1.(\Trans_{x_1}(M_1) \| \new x.(R' \| \out x[y].(\Trans_y(L_2) \| \fwd[z<-x]))) \ldots );
                        \end{align*}
                        we have $P \redto* R_0$.
                        Also, let 
                        \begin{align*}
                          M_0 & = (N'\subst{y_1->N_1, , y_m->N_m}\ L_2)\subst{x_1->M_1, , x_n->M_n}\\
                              & = (N'\ L_2)\subst{y_1->N_1, , y_m->N_m, x_1->M_1, , x_n->M_n}
                        \end{align*}
                        we have $M \alfredto* M_0$.
                        We have
                        \begin{align*}
                        R_0 \congrto
                            \new x_n.(&
                              \Trans_{x_n}(M_n) \|
                                \ldots
                                \\&\quad
                                \new x_1.(
                                  \Trans_{x_1}(M_1) \|
                                  \new y_m.(
                                    \Trans_{y_m}(N_m) \|
                                    \ldots
                                    \\&\mkern40mu
                                    \new y_1.(
                                      \Trans_{y_1}(N_1) \|
                                      \new x.(
                                        \Trans_x(N') \|
                                        \out x[y].(
                                          \Trans_y(L_2) \| \fwd[z<-x]
                                        )
                                      )
                                    )
                                    \ldots
                                  )
                                )
                                \ldots
                            ),
                        \end{align*}
                        so $R_0 \sublift M_0$.
                        This proves the thesis.

                    \item
                        $\Trans_x(L_1)$ is prefixed by a spawn which commutes past the restriction on $x$.
                        This case is analogous to the case of a spawn prefix in $\Trans_x(M')$ commuting past cuts, which follows the current case.

                    \item
                        The output on $x$ synchronizes with an input on $x$ in $\Trans_x(L_1)$.
                        By typability, then $L_1 = \Lam y.L'_1$ and $\Trans_x(L_1) = \inp x(y).\Trans_x(L'_1)$.
                        Let $Q'_0 = \new x.(\new y.(\Trans_y(L_2) \| \Trans_x(L'_1)) \| \fwd[z<-x])$.
                        Then $\Trans_x(M') \redd Q'_0$.

                        From $Q'_0$, there may be similar reductions as from $\Trans_x(M')$, with an additional forwarder reduction possible.
                        All of these reductions are independent, so we postpone all but the forwarder reduction.
                        Let $Q'_1 = \new y.(\Trans_y(L_2) \| \Trans_z(L'_1))$.
                        Then $Q'_0 \redd Q'_1$.

                        Let $Q' = \new x_n.(\Trans_{x_n}(M_n) \| \ldots \new x_1.(\Trans_{x_1}(M_1) \| Q'_1))$; we have $P \redd^2 Q'$.
                        Also, let $M_0 = M'_1\subst{y->L_2,x_1->M_1,,x_n->M_n}$; we have $M \alfredto M_0$.
                        Since $Q' \redd^{k-2} Q$, the thesis then follows from \ih2.
                \end{itemize}

                \case[\ref{rule:impl-elim}]
                Analogous to Case~\ref{rule:wand-elim}.

                \case[\ref{rule:mone-intro}]
                We have $M' = \unitM$ and $\Trans_z(M') = \out z<>$; no reductions are possible.

                \case[\ref{rule:aone-intro}]
                Analogous to Case~\ref{rule:mone-intro}.

                \case[\ref{rule:mone-elim}]
                We have $M' = \Let \unitM = L_1 in L_2$ and $\Trans_z(M') = \new x.(\Trans_x(L_1) \| \inp x().\Trans_z(L_2))$.
                There are three possible reductions: $\Trans_x(L_1)$ reduces internally, $\Trans_x(L_2)$ is prefixed by a spawn which commutes past the restriction on $x$, or the empty input on $x$ synchronizes with an empty output on $x$ in $\Trans_x(L_1)$.
                The former two sub-cases are analogous to the similar sub-cases in Case~\ref{rule:wand-elim}.
                In the latter case, by typability, we have $L_1 = \unitM$ and $\Trans_x(L_1) = \out x<>$.

                Let $Q_0 = \Trans_z(L_2)$; we have $\Trans_z(M') \redd Q_0$.
                Let $$Q' = \new x_n.(\Trans_{x_n}(M_n) \| \ldots \new x_1.(\Trans_{x_1}(M_1) \| Q_0) \ldots );$$ we have $P \redd Q'$.

                {Let ${M_0 = L_2\subst{x_1->M_1,,x_n->M_n}}$; we} have $M \alfredto M_0$ and $Q' \sublift M_0$.
                Since $Q' \redd^{k-1} Q$, the thesis follows from \ih2.

                \case[\ref{rule:aone-elim}]
                Analogous to Case~\ref{rule:mone-elim}.

                \case[\ref{rule:ast-intro}]
                We have $M' = \langle L_1,L_2 \rangle$ and $\Trans_z(M') = \out z[y].(\Trans_y(L_1) \| \Trans_z(L_2))$; no reductions are possible.

                \case[\ref{rule:conj-intro}]
                Analogous to Case~\ref{rule:ast-intro}.

                \case[\ref{rule:ast-elim}]
                We have $M' = \Let \langle x,y \rangle = L_1 in L_2$ and $\Trans_z(M') = \new y.(\Trans_y(L_1) \| \inp y(x).\Trans_z(L_2))$.
                There are three possible reductions: $\Trans_y(L_1)$ reduces internally, $\Trans_y(L_1)$ is prefixed by a spawn which commutes past the restriction on $y$, or the input on $y$ synchronizes with an output on $y$ in $\Trans_y(L_1)$.
                The former two sub-cases are analogous to the similar sub-cases in Case~\ref{rule:wand-elim}.
                In the latter case, by typability, we have $L_1 = \langle K_1,K_2 \rangle$ and $\Trans_y(L_1) = \out y[x].(\Trans_x(K_1) \| \Trans_y(K_2))$.

                Let $Q_0 = \new y.(\Trans_y(K_2) \| \new x.(\Trans_x(K_1) \| \Trans_z(L_2)))$; we have $\Trans_z(M') \redd Q_0$.
                Let \[Q' = \new x_n.(\Trans_{x_n}(M_n) \| \ldots \new x_1.(\Trans_{x_1}(M_1) \| Q_0) \ldots );\] we have $P \redd Q'$.

                Let $M_0 = L_2\subst{x->K_1,y->K_2,x_1->M_1,,x_n->M_n}$; we have $M \alfredto M_0$ and $Q' \sublift M_0$.
                Since $Q' \redd^{k-1} Q$, the thesis follows from \ih2.

                \case[\ref{rule:conj-elim}]
                We have $M' = \proj_i L$ and $\Trans_z(M') = \new x_2.(\Trans_{x_2}(L) \| \inp x_2(x_1).\spw{x_{i'}->\emptyset}.\fwd[z<-x_i])$ for $i \in \set{1,2}$ and $i' \in \set{1,2} \setminus \set{i}$.
                W.l.o.g., let $i = 1$ and $i' = 2$.
                There are three possible reductions: $\Trans_{x_2}(L)$ reduces internally, $\Trans_{x_2}(L)$ is prefixed by a spawn which commutes past the restriction on $x_2$, or the input on $x_2$ synchronizes with an output on $x_2$ in $\Trans_{x_2}(L)$.
                The former two sub-cases are analogous to the similar sub-cases in Case~\ref{rule:wand-elim}.

                In the latter case, by typability, we have $L = (L_1,L_2)$ and $\Trans_{x_2}(L) = \out {x_2}[x_1].(\Trans_{x_1}(L_1) \| \Trans_{x_2}(L_2))$.
                Let $Q_0 = \new x_2.(\Trans_{x_2}(L_2) \| \new x_1.(\Trans_{x_1}(L_1) \| \spw{x_2->\emptyset}.\fwd[z<-x_1]))$; we have $\Trans_z(M') \redd Q_0$.

                From $Q_0$, there may be internal reductions of $\Trans_{x_2}(L_2)$ or $\Trans_{x_1}(L_1)$, a spawn prefix in $\Trans_{x_2}(L_2)$ may commute past the restriction on $x_2$, a spawn prefix in $\Trans_{x_1}(L_1)$ may commute past the restrictions on $x_1$ and $x_2$, and the spawn prefix $\spawn{x_2->\emptyset}$ may commute past the restriction on $x_1$.
                All these reductions are independent, so we postpone all but the commute of the spawn prefix $\spawn{x_2->\emptyset}$.
                Let $Q_1 = \new x_2.(\Trans_{x_2}(L_2) \| \spw{x_2->\emptyset}.\new x_1.(\Trans_{x_1}(L_1) \| \fwd[z<-x_1]))$; we have $Q_0 \redd Q_1$.

                From $Q_1$ we have the same possible reductions as from $Q_0$, except that there may also be weakening of $x_2$ due to the spawn prefix $\spawn{x_2->\emptyset}$ interacting with the restriction on $x_2$.
                Again, we postpone all but the latter reduction.
                Let $Q_2 = \spw{y->\emptyset | y \in \fv(L_2)}.\new x_1.(\Trans_{x_1}(L_1) \| \fwd[z<-x_1])$; we have $Q_1 \redd Q_2$.

                From $Q_2$, we again have the reductions that were available from $Q_1$, but also the reduction of the forwarder $\fwd[z<-x_1]$.
                We postpone all but the latter.
                Let \[Q_3 = \spw{y->\emptyset | y \in \fv(L_2)}.\Trans_z(L_1);\] we have $Q_2 \redd Q_3$.

                Let $Q'_0 = \new x_n.(\Trans_{x_n}(M_n) \| \ldots \new x_1.(\Trans_{x_1}(M_1) \| Q_3) \ldots )$; we have $P \redd^4 Q'_0$.
                Also, let $M_0 = L_1\subst{x_1->M_1,,x_n->M_n}$, where the resources used by $L_2$ have been weakened.
                At this point, $Q'_0$ is not of appropriate shape to relate it to $M_0$ through substitution lifting, because of the spawn prefix in $Q_3$.
                There are two possibilities for reduction here: the spawn interacts with one of the cuts on $x_i$, or the spawn commutes past them all.
                The former is analogous to the case of a spawn prefix in $\Trans_z(M')$ interacting with a cut, which follows the current case.
                In the latter case, let \[Q' = \spw{y->\emptyset | y \in \fv(L_2)}.\new x_n.(\Trans_{x_n}(M_n) \| \ldots \new x_1.(\Trans_{x_1}(M_1) \| \Trans_z(L_1)) \ldots ).\]
                Now $Q'_0 \redd^n Q'$ and $Q' \sublift M_0$.

                Since $Q' \redd^{k-4-n} Q$, the thesis follows from \ih2.

                \case[\ref{rule:disj-intro}]
                We have $M' = \inj_i(N)$ for $i \in \set{1,2}$.
                Depending on the value of $i$,
                $\Trans_z(M')$ is either
                $\selL{z}.\Trans_z(N)$
                or $\selR{z}.\Trans_z(N)$.
                Either way, no reductions are possible.

                \case[\ref{rule:disj-elim}]
                We have
                \begin{align*}
                M' &= \Case N of y_1 => L_1 or y_2 => L_2
                \\
                \Trans_z(M') &=
                    \new y_1.(\Trans_{y_1}(N) \| \caseLR{y_1}{\Trans_z(L_1)}{\new y_2.(\fwd[y_2<-y_1] \| \Trans_z(L_2))})
                \end{align*}.
                There are three possible reductions: $\Trans_{y_1}(N)$ reduces, $\Trans_{y_1}(N)$ is prefixed by a spawn which commutes past the restriction on $y_1$, or the case on $y_1$ synchronizes with a select on $y_1$ in $\Trans_{y_1}(N)$.
                The former two sub-cases are analogous to the similar sub-cases in Case~\ref{rule:wand-elim}.

                In the latter case, by typability, we have $N = \inj_i(N')$ for $i \in \set{1,2}$.
                The rest of the analysis depends on the value of $i$:
                \begin{casesplit}
                    \case[$i = 1$]
                    We have $\Trans_{y_1}(N) = \selL{y_1}.\Trans_{y_1}(N')$.
                    Let $Q_0 = \new y_1.(\Trans_{y_1}(N') \| \Trans_z(L_1))$; we have $\Trans_z(M') \redd Q_0$.
                    Let $Q' = \new x_n.(\Trans_{x_n}(M_n) \| \ldots \new x_1.(\Trans_{x_1}(M_1) \| Q_0))$; we have $P \redd Q'$.
                    Also, let $M_0 = L_1\subst{y_1->N',x_1->M_1,,x_n->M_n}$; we have $M \alfredto M_0$.
                    Moreover, $Q' \sublift M_0$.
                    Since $Q' \redd^{k-1} Q$, the thesis follows from \ih2.

                    \case[$i = 2$]
                    We have $\Trans_{y_1}(N) = \selR{y_1}.\Trans_{y_1}(N')$.
                    Let \[Q_0 = \new y_1.(\Trans_{y_1}(N') \| \new y_2.(\fwd[y_2<-y_1] \| \Trans_z(L_2)))\] for which we have $\Trans_z(M') \redd Q_0$.

                    From $Q_0$ several reductions are possible: $\Trans_{y_1}(N')$ or $\Trans_z(L_2)$ reduce internally, a spawn prefix in $\Trans_{y_1}(N')$ commutes past the restriction on $y_1$, a spawn prefix in $\Trans_z(L_2)$ commutes past or interacts with the restriction on $y_2$, or the forwarder $\fwd[y_2<-y_1]$ interacts with the restriction on $y_1$.
                    These reductions are independent, so we postpone all but the forwarder reduction.
                    Let $Q_1 = \new y_2.(\Trans_{y_2}(N') \| \Trans_z(L_2))$; we have $Q_0 \redd Q_1$.

                    Let $Q' = \new x_n.(\Trans_{x_n}(M_n) \| \ldots \new x_1.(\Trans_{x_1}(M_1) \| Q_1) \ldots )$; we have $P \redd^2 Q'$.
                    Also, let $M_0 = L_2\subst{y_2->N',x_1->M_1,,x_n->M_n}$; we have $M \alfredto M_0$.
                    Moreover, $Q' \sublift M_0$.
                    Since $Q' \redd^{k-2} Q$, the thesis follows from \ih2.
                \end{casesplit}
            \end{casesplit}

        \item
            A forwarder in $\Trans_{x_i}(M_i)$ for $i \in [1,n]$ interacts with a cut.
            We apply induction on the number $k$ of reductions from $P$ to $Q$, i.e., $P \redd^k Q$.

            We have $\Trans_{x_i}(M_i) = \fwd[x_i<-y]$ for some $y$.
            Hence, $M_i = y$.
            Let \[Q' = \new x_n.(\Trans_{x_n}(M_n) \| \ldots \new x_1.(\Trans_{x_1}(M_1) \| \Trans_z(M'\subst{x_i->y})) \ldots )\] without the cut on $x_i$.
            Since $\Trans_z(M')\subst{x_i->y} = \Trans_z(M'\subst{x_i->y})$, we have $P \redto Q'$.
            By typability, none of the $M_1,\ldots,M_{i-1}$ can contain the variable $y$, so we have \[M'\subst{x_1->M_1,,x_i->y,,x_n->M_n} = (M'\subst{x_i->y})\subst{x_1->M_1,,x_n->M_n}\] where the latter substitutions do not contain the substitution on $x_i$.
            Hence, $Q' \sublift M$.
            Since $Q' \redd^{k-1} Q$, by the IH, there exist $N$ and $R$ such that $M \alfredto* N$ and $Q' \redto* R \sublift N$.
            This proves the thesis.

        \item
            A forwarder in $\Trans_z(M')$ interacts with a cut.
            We apply induction on the number $k$ of reductions from $P$ to $Q$, i.e., $P \redd^k Q$.

            We have $\Trans_z(M') = \fwd[z<-y]$ for some $y$ and there exist $i \in [1,n]$ such that $x_i = y$.
            Hence, ${M' = y}$ and $M = y\subst{x_1->M_1,,y->M_i,,x_n->M_n}$.
            Let \[Q' = \new x_n.(\Trans_{x_n}(M_n) \| \ldots \new x_1.(\Trans_{x_1}(M_1) \| \Trans_z(M_i)) \ldots ).\]
            Since $\Trans_{y}(M_i)\subst{y->z} = \Trans_z(M_i)$, we have $P \redto Q'$.
            By typability, none of the $M_1,\ldots,M_{i-1}$ can contain the variable $y$, so we have \[y\subst{x_1->M_1,,y->M_i,,x_n->M_n} = M_i\subst{x_1->M_1,,x_n->M_n}\] where the latter substitutions do not contain the substitution on $x_i$.
            Hence, $Q' \sublift M$.
            Since $Q' \redd^{k-1} Q$, by the IH, there exist $N$ and $R$ such that $M \alfredto* N$ and $Q' \redto* R \sublift N$.
            This proves the thesis.

        \item
            A spawn prefix in $\Trans_z(M')$ commutes past or interacts with a cut.
            We apply induction on the number $k$ of reductions from $P$ to $Q$, i.e., $P \redd^k Q$.
            The last applied rule in the typing derivation of $\Trans_z(M')$ is~\ref{rule:N-W} or~\ref{rule:N-C}; the rest of the analysis depends on which:
            \begin{casesplit}
                \case[\ref{rule:N-W}]
                The rule weakens the variables $y_1,\ldots,y_m$.
                This case follows by commuting the spawn prefix past cuts on $x_i \notin \set{y_1,\ldots,y_m}$ and performing the weakening when the spawn prefix meets cuts on $x_i \in \set{y_1,\ldots,y_m}$.
                Other reductions that were possible from $P$ remain possible throughout this process, but they are independent of these reductions, so we can postpone them.
                After $k'$ steps of reduction, we reach from $P$ a process $Q'$ with the spawn commuted to the top of the process, and some cuts removed.
                The cuts that were removed concern substitutions of weakened variables, so removing these substitutions from $M$ makes no difference.
                Similarly, the cuts that were commuted past concern substitutions that are independent of the weakening.
                Hence, $Q' \sublift M$.
                Since $M \alfredto* M$ and $Q' \redd^{k-k'} Q$, the thesis follows from the IH.

                \case[\ref{rule:N-C}]
                This case is largely analogous to the case of~\ref{rule:N-W}, except that interactions of the spawn with cuts duplicates the cuts, and commuting past cuts moves the substitutions related to the contraction towards the end of the list of substitutions applied in $M$.
            \end{casesplit}

        \item
            A spawn prefix in $\Trans_{x_i}(M_i)$ for $i \in [1,n]$ commutes past a cut.
            This case is largely analogous to the previous case: first, the spawn can always commute past the cut on $x_i$, after which it may commute further or interact with cuts.
            \qedhere
    \end{itemize}
\end{proof}

 \section{Denotational semantics}
The interpretation of \piBI in $\Set^{\pset{\Tag}}$ essentially follows the interpretation of BI proofs in doubly-closed categories (DCCs).
Forwarders are interpreted as identity morphisms, and cut is interpreted as composition.
Suppose we are given morphisms $\Sem{P} : \Sem{\bunch} \to \Sem{A}$ and $\Sem{Q} : \Sem{\bctxt(x : A)} \to \Sem{C}$.
First, note that we can write $\Sem{\bctxt(x : A)}$ as $\Sem{\bctxt}(\Sem{A})$ where $\Sem{\bctxt}$ is the obvious endofunctor interpreting bunched contexts.
The we interpret the cut $\Sem{\new x. (P \| Q)}$ as $\Sem{Q} \circ \Sem{\bctxt}(\Sem{P})$.
Let us spell the cases for some other propositions.

For \textbf{separating conjunction} we have:
\[
\Sem{\Delta_1 \bsep \Delta_2 \proves \out x[y].(P_1 \| P_2) :: x : A \ast B} = \Sem{P_1} \ast \Sem{P_2},
\]
where $\Sem{P_1} \ast \Sem{P_2}$  is a monoidal product of two morphisms:
\begin{align*}
  \Sem{\Delta_1} \xrightarrow{\Sem{P_1}} \Sem{A} &&
  \Sem{\Delta_2} \xrightarrow{\Sem{P_2}} \Sem{B}
\end{align*}

And for the left rule:
\[
\Sem{\bctxt(x : A \ast B) \proves x(y).P :: z : C} = \Sem{P},
\]
since the type $\Sem{A \ast B}$ and a context $\Sem{y : A\bsep x : B}$ are isomorphic.

For the \textbf{magic wand} we interpret the right rule by currying:
\[
\Sem{\bunch \proves x(y).P :: x : A \wand B}(d)(a) = \Sem{P}(d,a),
\]
where $d \in \Sem{\bunch}$ and $a \in \Sem{A}$.
For the left rule we have:
\[
\Sem{\bctxt (\bunch \bsep x : A \wand B) \proves \out x[y].(P \| Q) :: z : C} =
\Sem{Q} \circ \Sem{\bctxt}(\mathit{ev} \circ \Sem{P}),
\]
where $\mathit{ev}$ is the evaluation morphism $A \ast (A \wand B) \to B$.

For \textbf{intuitionistic conjunction and implication} the interpretation is the as above, except we are using the closed Cartesian structure on $\Set^{\pset{\Tag}}$.

In order to interpret the \textbf{spawn prefix}, we need formulate the semantics of (typed) spawn bindings.
Each $\senv \from \bunch_1 \tobunch \bunch_2$ induces a map $\Sem{\senv} : \Sem{\bunch_1} \to \Sem{\bunch_2}$.
Then we interpret the \Cref{rule:struct} as
\[
\Sem{\spawn{\senv}.P} = \Sem{P} \circ \Sem{\senv}.
\]

\textbf{Units and disjunction} are interpreted as usual in Cartesian closed categories.

\subsection{Observational equivalence}

\observeWeak*
\begin{proof}
  Note that the condition that $P$ does not have any barbs on channels from $\Gamma$ is equivalent to the statement
  $\ident(\Gamma) \cap \an(P) = \emptyset$.
  We then prove this reformulated statement by induction on the structure of $P$, using \Cref{t:progress} when needed.
    \begin{itemize}
        \item
            If $P$ is a communication prefix, then $P \equiv \alpha(x).P'$.
            Clearly, $x \in \an(P)$.
            Since, by assumption, $\ident(\Gamma) \cap \an(P) = \emptyset$, this means that $x \notin \ident(\Gamma)$.
            By definition, $x \in \fn(P)$, so it follows by \Cref{l:fn_ident}, that $x = z$.
            Then $P \barb{\alpha(z)}$.
        \item
            If $P$ is a forwarder, then $P \equiv \fwd[x<-y]$.
            By typability, $x = z$ and $y \in \ident(\Gamma)$.
            By definition, $y \in \an(P)$.
            This contradicts the assumption that $\ident(\Gamma) \cap \an(P) = \emptyset$, so this case is invalid.

        \item
            If $P$ is a cut, then $P \equiv \new x.(Q \|_x R)$.
            We have that $P$ is not live: otherwise, $P \nredd$ would contradict \Cref{t:progress} (Progress).
            By inversion of typing, we have $\Gamma = \Gamma_1(\Gamma_2)$ where $\Gamma_2 \proves Q :: x:A$ and $\Gamma_1(x:A) \proves R :: z:C$.

            By definition, $\an(Q) \subseteq \an(P) \cup \{x\}$ and $\ident(\Gamma_2) \subseteq \ident(\Gamma)$.
            Moreover, by typability, $x \notin \ident(\Gamma_2)$.
            Since, by assumption, $\ident(\Gamma) \cap \an(P) = \emptyset$, it follows that $\ident(\Gamma_2) \cap \an(Q) = \emptyset$.
            Because $P \nredd$, we have $Q \nredd$.
            Also, because $P$ is not live, $Q \not\equiv \spawn{\sigma}.Q'$.
            Hence, by the IH, $Q \barb{\alpha'(x)}$.
            Clearly, this means that $x \in \an(Q)$.

            By definition, $\an(R) \subseteq \an(P) \cup \{x\}$ and $\ident(\Gamma_1(\hole)) \subseteq \ident(\Gamma)$.
            Since, by assumption, $\ident(\Gamma) \cap \an(P) = \emptyset$, it follows that $\ident(\Gamma_2(x:A)) \cap \an(R) \subseteq \{x\}$.
            Because $P$ is not live, we have $x \notin \an(Q) \cap \an(R)$.
            Since $x \in \an(Q)$, it follows that $x \notin \an(R)$.
            Hence, $\ident(\Gamma_2(x:A)) \cap \an(R) = \emptyset$.
            Because $P \nredd$, we have $R \nredd$.
            Also, because $P$ is not live, $R \not\equiv \spawn{\sigma}.R'$.
            Hence, by the IH, $R \barb{\alpha(z)}$.
            Therefore, $P \barb{\alpha(z)}$.
          \item The case when $P$ is a spawn prefix is similar to the previous case.
            \qedhere
    \end{itemize}
\end{proof}

\end{document}